%% file: main.tex
\documentclass[11pt]{article}

\newif\ifcomments
\commentstrue

\input{header}

\title{Near-Optimal Deterministic Vertex-Failure Connectivity Oracles}
\author{Yaowei Long\\ University of Michigan \and Thatchaphol Saranurak\\ University of Michigan}
\date{}

\begin{document}

\maketitle
\pagenumbering{gobble}
\input{0-abstact}

\newpage
\setcounter{tocdepth}{2}
\tableofcontents
\newpage

\pagenumbering{arabic}
\input{1-intro}

\input{2-overview}

\input{3-body}

\input{9-open}

\appendix
\input{A-Appendix}

\input{main.bbl}

\end{document}

%% file: header.tex
\usepackage{fullpage}
\usepackage[utf8]{inputenc}

\usepackage{graphicx}
\usepackage{amsmath,amsthm,amssymb}
\usepackage{algorithm}
\usepackage{algpseudocode}
\usepackage{multirow}
\usepackage{makecell}

\usepackage{thm-restate,color,xspace}
\usepackage{comment}
\usepackage{thmtools}
\usepackage{xcolor}
\usepackage{nameref}
\usepackage{array}

\definecolor{ForestGreen}{rgb}{0.1333,0.5451,0.1333}
\definecolor{DarkRed}{rgb}{0.65,0,0}
\definecolor{Red}{rgb}{1,0,0}
\usepackage[linktocpage=true,
pagebackref=true,colorlinks,
linkcolor=DarkRed,citecolor=ForestGreen,
bookmarks,bookmarksopen,bookmarksnumbered]
{hyperref}
\usepackage{cleveref}

\declaretheorem[numberwithin=section]{theorem}
\declaretheorem[numberlike=theorem]{lemma}
\declaretheorem[numberlike=theorem,name=Lemma]{lem}

\declaretheorem[numberlike=theorem]{proposition}

\declaretheorem[numberlike=theorem]{corollary}

\declaretheorem[numberlike=theorem]{conjecture}

\declaretheorem[numberlike=theorem]{claim}

\declaretheorem[numberlike=theorem,style=definition]{definition}
\declaretheorem[numberlike=theorem,style=definition]{remark}

\global\long\def\polylog{\mathrm{polylog}}
\global\long\def\polyloglog{\mathrm{polyloglog}}
\global\long\def\CMG{\mathrm{CMG}}
\global\long\def\glo{\mathrm{glo}}
\global\long\def\ET{\mathrm{ET}}
\global\long\def\Tag{\mathrm{Tag}}
\global\long\def\Id{\mathrm{Id}}

\global\long\def\Ans{\mathrm{Ans}}
\global\long\def\d{\boldsymbol{\mathrm{d}}}
\global\long\def\littleo{\tilde{\tilde{o}}}
\global\long\def\OMv{{\sf OMv}}
\global\long\def\OuMv{{\sf OuMv}}
\global\long\def\Boruvka{Bor$\mathring{\text{u}}$vka}
\global\long\def\iin{{\rm in}}
\global\long\def\out{{\rm out}}

\global\long\def\full{{\rm full}}
\global\long\def\LCA{{\rm LCA}}
\global\long\def\ssum{{\rm sum}}

\newcommand{\Decomp}{\textsc{Decomp}}

\newcommand{\threeDRangeSum}{\textsc{3DRangeSum}}

\newcommand{\FindNeighbors}{\textsc{FindNeighbor}}
\newcommand{\BatchConnCheck}{\textsc{BatchConnCheck}}

\newcommand{\Ehat}{\hat{E}}
\newcommand{\Fhat}{\hat{F}}
\newcommand{\Otil}{\tilde{O}}
\newcommand{\Ohat}{\hat{O}}
\newcommand{\Omegahat}{\hat{\Omega}}
\global\long\def\poly{\mathrm{poly}}
\global\long\def\mbar{\bar{m}}

\global\long\def\A{\mathcal{A}}
\global\long\def\I{\mathcal{I}}

\global\long\def\cS{\mathcal{S}}
\global\long\def\Gaug{\hat{G}}
\global\long\def\Table{\mathsf{Table}}
\global\long\def\fault{\mathrm{fault}}

%% file: 0-abstact.tex
\begin{abstract}
We revisit the \emph{vertex-failure connectivity oracle }problem \cite{duan2010connectivity,borradaile2012connectivity,DuanP20}.
This is one of the most basic graph data structure problems under
\emph{vertex updates}, yet its complexity is still not well-understood.
We essentially settle the complexity of this problem by showing a
new data structure whose space, preprocessing time, update time, and
query time are simultaneously optimal up to sub-polynomial factors
assuming popular conjectures. Moreover, the data structure is deterministic. 

More precisely, for any integer $d_{\star}$, the data structure preprocesses
a graph $G$ with $n$ vertices and $m$ edges in $\Ohat(md_{\star})$
time and uses $\Otil(\min\{m,nd_{\star}\})$ space. Then, given the
vertex set $D$ to be deleted where $|D|=d\le d_{\star}$, it takes
$\Ohat(d^{2})$ updates time. Finally, given any vertex pair $(u,v)$,
it checks if $u$ and $v$ are connected in $G\setminus D$ in $O(d)$
time. This improves the previously best deterministic algorithm by
Duan and Pettie \cite{DuanP20} in both space and update time by a factor of $d$.
It also significantly speeds up the $\Omega(\min\{mn,n^{\omega}\})$
preprocessing time of all known (even randomized) algorithms with update
time at most $\Otil(d^{5})$.

\end{abstract}

%% file: 1-intro.tex
\section{Introduction}

We revisit the \emph{vertex-failure connectivity oracle }problem \cite{duan2010connectivity,borradaile2012connectivity,DuanP20}.
This is one of the most basic graph data structure problems under
\emph{vertex updates}, yet its complexity is still not well-understood.
In this paper, we essentially settle the complexity of this problem
by showing a new data structure whose space, preprocessing time, update
time, and query time are simultaneously optimal up to sub-polynomial
factors assuming popular conjectures. Moreover, the data structure
is deterministic. 

More precisely, in the vertex-failure connectivity oracle problem,
there are three phases. First, we preprocess an undirected graph $G=(V,E)$
with $n$ vertices and $m$ edges and an integer $d_{\star}$. Second,
given a vertex set $D\subset V$ of size $|D|=d\le d_{\star}$, we
update the data structure. Finally, given a vertex pair $(u,v)$,
we return whether $u$ and $v$ are connected in the updated graph
$G\setminus D$. The \emph{preprocessing},\emph{ update, }and \emph{query
time }are the time used by the data structured in the first, second,
and third phases respectively. Put it another way, this problem is
just the well-known decremental connectivity problem when all vertex
deletions in given in one batch, instead of an online sequence of
updates.

To see the context, we also mention its sister problem, the \emph{edge}-failure
connectivity oracle problem, where we are given $d$ edges to be deleted
instead of vertices. Since 1997, Pastrascu and Thorup \cite{patrascu2007planning}
already showed a near-optimal deterministic edge-failure oracle, except
that it has large preprocessing time. This is later improved to a
near-optimal oracle with $\Otil(m)$ space\footnote{The algorithm uses $\Otil(n)$ space outside storing the graph. In \cite{DuanP20}, they claimed $\Otil(n)$ total space by assuming that $D \subseteq E$ and each edge in $D$ is specified by its two endpoints.   
Note that, in general, $\Omega(m)$ space is unavoidable as the oracles can be used to reconstruct any connected graph: we can query if $(u,v)\in E$, by deleting $\binom{V}{2} \setminus (u,v)$ and querying if $u$ and $v$ are connected.}, $\Otil(m)$ preprocessing, $\Otil(d)$ update, and $O(\log \log n)$
query time using the linear sketching techniques (see Threoem 7.9 of \cite{DuanP20} and also \cite{kapron2013dynamic,gibb2015dynamic}).\footnote{Throughout the paper, we use $\Otil(\cdot)$ to hide a $\polylog(n)$
factor, use $\Ohat(\cdot)$ to hide a $n^{o(1)}$ factor and use $\bar{O}(\cdot)$ to hide a $\polyloglog(n)$ factor.}

The state of the art for vertex-failure connectivity oracles is much
worse. The high-level reason why vertex failures are more challenging
is because, while $d$ edge deletions cause at most $d$ new connected
components, just a single vertex deletion may create a lot of new
connect components. For very small $d_{\star}\le2$, classical graph
decompositions such that block trees and SQRT trees yield optimal
vertex-failure oracles. When $d_{\star}=3$, the data structure by
Kanevsky et al. \cite{kanevsky1991line} implies a near-optimal oracle
with $O(n)$ space, $\Otil(m)$ preprocessing, $O(1)$ update and
query time. It was until 2010 when Duan and Pettie \cite{duan2010connectivity}
gave the first solution for general $d_{\star}$ by showing a deterministic
data structure with $\Otil(md_{\star}^{1-\frac{2}{c}}n^{\frac{1}{c}-\frac{1}{c\log(2d_{\star})}})$
space and preprocessing time, $\Otil(d^{2c+4})$ update and $O(d)$
query time. Later in 2017, they \cite{DuanP20} showed
two improved algorithms: (1) a deterministic data structure with $\Otil(md_{\star})$
space, $\Otil(mn)$ preprocessing, $\Otil(d^{3})$ update, and $O(d)$
query time, and (2) a randomized Monte Carlo data structure with $\Otil(m)$
space, $\Otil(mn)$ preprocessing, $\Otil(d^{2})$ update, and $O(d)$
query time. More recently, oracles whose update and query time depend
only on $d_{\star}$ or $d$ were shown \cite{van2019sensitive,pilipczuk2021algorithms},
but the bounds are significantly slower than the ones in \cite{DuanP20}
when $\polylog(n)$ factors are ignored.

Given the big discrepancy between the edge-failure and vertex-failure
cases in all time bounds, a natural question is whether we can match
them. Unfortunately, the conditional lower bounds \cite{HKNS15}
show that query time must be $\Omegahat(d)$ as long as the algorithm
has $\poly(n)$ preprocessing and $\poly(dn^{o(1)})$ update time.
This separates the vertex case from the edge case. The next question
is then 
\begin{center}
\emph{What is the right complexity for vertex-failure connectivity oracles}? 
\par\end{center}

Before our paper, the answer of this question was not well-understood.
For example, three questions are posted by Duan and Pettie \cite{DuanP20}
where they wrote \emph{``the following open problems are quite challenging''.
}The first question is about improving preprocessing time. Second,
can we have both $\Otil(d)$ update and query time? Third, is there
a deterministic data structure with $\Otil(m)$ space? We resolve
all these questions in this paper. In fact, we essentially settle
the complexity of all parameters of vertex-failure connectivity oracles
up to sub-polynomial factors assuming popular conjectures.\emph{}%

{} Our near-optimal deterministic oracle is stated below.
\begin{theorem}
\label{thm:main}There exists a deterministic vertex failure connectivity
oracle that uses $O(m\log^{*}n)$ space, $\Ohat(m)+\Otil(md_{\star})$
preprocessing time, $\Ohat(d^{2})$ update time, and $O(d)$ query
time. 
\end{theorem}

All the $m$ factors can be replaced by $\mbar=\min\{m,n(d_{\star}+1)\}$
by using the standard sparsification algorithm by Nagamochi and Ibaraki
\cite{nagamochi1992linear} to first preprocess the graph in $O(m)$ time. 
(This also works for all known oracles discussed above.)
So the bounds can be improved to $O(\mbar \log^* n)$ space\footnote{The $\log^* n$ factor in our space is arbitrary. We can choose any growing function, which only affects the $n^{o(1)}$ factor in the update time.}
and $O(m)+\Ohat(\mbar)+\Otil(\mbar d_{\star})$ preprocessing time.
Also, all the $n^{o(1)}$ factors in \Cref{thm:main} come solely from
the overhead factor of the current deterministic vertex expander decomposition.
It is believable that there exists a vertex expander decomposition
algorithm $\polylog(n)$ overhead factor, which would automatically
improve all the $n^{o(1)}$ factors to $\polylog(n)$.\footnote{But this would increase our space to $O(m\log^{2}n)$.} See \Cref{thm:MainDetailed} for a detailed version of \Cref{thm:main}.

Comparing with known results (see \Cref{tab:compare}), \Cref{thm:main}
improves the previous best deterministic algorithm by \cite{DuanP20}
by a factor of $d$ in both space and update time. It also significantly speeds up the $\Omega(\min\{mn,n^{\omega}\})$
preprocessing time of all known (randomized) algorithms with update time at most $\Otil(d^{5})$. %
Lastly, our space complexity also beats the state-of-the-art of $O(m\log^6 n)$ in the randomized oracle by 
\cite{DuanP20}.

Next, we state conditional lower bounds that is tight with \Cref{thm:main}. 
\begin{theorem}
\label{thm:main lower}Let $\A$ be any vertex-failure connectivity
oracle with $t_{p}$ preprocessing time, $t_{u}$ update time, and
$t_{q}$ query time bound. Assuming popular conjectures, we have the
following:
\begin{enumerate}
\item $\A$ must take $\Omega(\min\{m,nd_{\star}\})$ space. (See \cite{DuanP20}) 
\item If $t_{p}=\poly(n)$, then $t_{u}+t_{q}=\Omegahat(d^{2})$.
(See \Cref{thm:UpdateLowerBound})
\item If $t_{p}=\poly(n)$ and $t_{u}=\poly(dn^{o(1)})$, then $t_{q}=\Omegahat(d)$.
(Corollary 3.12 of \cite{HKNS15})
\item If $t_{u},t_{q}=\poly(dn^{o(1)})$ and $\A$ is a combinatorial algorithm\footnote{``Combinatorial'' algorithms \cite{AW14} are algorithms
that do not use fast matrix multiplication.}, then $t_{p}=\Omegahat(md)$. (See \Cref{thm:ConstructionLowerBound})
\end{enumerate}
\end{theorem}

The above result basically says that (1) we cannot improve the space
bound, (2) either update or query time must be at least $\Omegahat(d^{2})$,
(3) query time must be at least $\Omegahat(d)$, and (4) the preprocessing
time must be at least $\Omegahat(md)$. This means that all parameters
from \Cref{thm:main} are simultaneously tight up to sub-polynomial
factors. Note, however, that the preprocessing time is tight only
for combinatorial algorithms.

\begin{table}
\footnotesize{

\begin{tabular}{|>{\raggedright}p{0.18\textwidth}|c|c|c|c|c|}
\hline 
 & {Det./Rand.} & Space & Preprocessing & Update & Query\tabularnewline
\hline 
\hline 
Block trees, \\
SQRT trees, and \cite{kanevsky1991line} \\
only when $d_{\star}\le3$ & Det. & $O(n)$ & $\Otil(m)$ & $O(1)$ & $O(1)$\tabularnewline
\hline 
Duan \& Pettie \cite{duan2010connectivity} for $c\ge1$ & Det. & $\Otil(\mbar d_{\star}^{1-\frac{2}{c}}n^{\frac{1}{c}-\frac{1}{c\log(2d_{\star})}})$ & linear in space & $\Otil(d^{2c+4})$ & $O(d)$\tabularnewline
\hline 
\multirow{2}{0.18\textwidth}{Duan \& Pettie \cite{DuanP20}} & Det. & $O(\mbar d_{\star}\log n)$ & $O(\mbar n\log n)$ & $O(d^{3}\log^{3}n)$ & $O(d)$\tabularnewline
\cline{2-6} 
 & Rand. & $O(\mbar\log^{6}n)$ & $O(\mbar n\log n)$ & $\bar{O}(d^{2}\log^{3}n)$ w.h.p. & $O(d)$\tabularnewline
\hline 
van den Brand \& Saranurak \cite{van2019sensitive}
& Rand. & $O(n^{2})$ & $O(n^{\omega})$ & $O(d^{\omega})$ & $O(d^{2})$\tabularnewline
\hline 
\multirow{2}{0.18\textwidth}{Pilipczuk et al. \cite{pilipczuk2021algorithms}} & Det. & $\mbar2^{2^{O(d_{\star})}}$ & $\mbar n^{2}2^{2^{O(d_{\star})}}$ & $2^{2^{O(d_{\star})}}$ & $2^{2^{O(d_{\star})}}$\tabularnewline
\cline{2-6} 
 & Det. & $n^{2}\poly(d_{\star})$ & $\poly(n)2^{O(d_{\star}\log d_{\star})}$ & $\poly(d_{\star})$ & $\poly(d_{\star})$\tabularnewline
\hline 
\multirow{2}{0.18\textwidth}{\textbf{This paper}} 
& Det. & $O(\mbar\log^{3}n)$ & $O(\mbar n\log n)$ & $\bar{O}(d^{2}\log^{3}n\log^{4}d)$ & $O(d)$\tabularnewline
\cline{2-6}
& Det. & $O(\mbar\log^{*}n)$ & $\Ohat(m)+\Otil(\mbar d_{\star})$ & $\Ohat(d^{2})$ & $O(d)$\tabularnewline
\hline 
\end{tabular}

}

\caption{Complexity of known vertex-failure connectivity oracles. $\protect\mbar=\min\{m,n(d_{\star}+1)\}$\label{tab:compare}. The randomized algorithms are all Monte Carlo. The $\log^{*}n$ factor in the space of the last result can be substituted with any slowly growing function.}
\end{table}

As a side result, we also study a related problem of \emph{vertex-failure
global connectivity oracles.} This oracle is similar to a vertex-failure
connectivity oracle, except that in the second phase, after we are
given $D\subset V$, we only need to return if $G\setminus D$ is
connected and there is no third phase. Surprisingly, although this
problem might seem easier because there is only one binary number
to maintain instead of $\binom{n}{2}$ possible queries, we show that
there is no oracle with reasonable update time for this problem. 
\begin{theorem}
\label{thm:main lower global}Let $\A$ be a vertex-failure global
connectivity oracle with $\poly(n)$ preprocessing time. Assuming
SETH, $\A$ must spend $\Omegahat(n)$ update time even when the deletion
set $D$ has size $|D|=n^{o(1)}$.
\end{theorem}

This is in stark contrast with the edge version; edge-failure connectivity
oracles of \cite{patrascu2007planning,kapron2013dynamic,gibb2015dynamic}
with $\Otil(d)$ update time can also answer global connectivity and
even count the number of connected components.

\paragraph{Technical Highlights.}

The technical novelty of this paper lies in the new oracle from \Cref{thm:main}
with two parts to be highlighted. 

First, we give a $\Ohat(m)$-time algorithm for computing the \emph{low-degree
hierarchy }of \cite{DuanP20}. This step previously requires
$O(mn\log n)$ time in the preprocessing algorithm by Duan and Pettie
\cite{DuanP20} and it was their bottleneck. This new
algorithm then leads to faster preprocessing time. Actually, similar hierarchies are the key structures behind most vertex-failure
oracles (for the connectivity problem \cite{duan2010connectivity}
and for the shortest path problem \cite{duan2021approximate}) Our
construction is based on expander decomposition, which is very versatile,
and so we expect that the approach can be adapted to solves more problems
under vertex failures. 

Towards our construction, we give the first almost-linear-time algorithm for computing \emph{vertex-expander decomposition}.
Given that in the last decade numerous fast graph algorithms are based on edge-expander decomposition (see e.g.~\cite{ST04,kelner2014almost,racke2014computing,nanongkai2017dynamic,chalermsook2021vertex}), our fast vertex-expander decomposition will likely very useful.
To facilitate future applications, we show a general decomposition that works even for general vertex weights and demands. We even generalize the technique to obtain hypergraph-expander decomposition. See \Cref{lem:overview:decomp} for the basic version,  \Cref{thm:ComplexDecomp} for the weighted version and \Cref{thm:HyperDecomp} for the hypergraph version.

Second, given a near-optimal low degree hierarchy, the main challenge
is designing an oracle with small space and update time based on it.
There were two previous solutions by Duan and Pettie \cite{DuanP20}:
(1) a randomized Monte-Carlo near-optimal algorithm with $\Otil(m)$
space and $\Otil(d^{2})$ update time, and (2) a deterministic algorithm
with $\Otil(md_{\star})$ space and $\Otil(d^{3})$ update time. We
show how to achieve the best of both, near-optimal and deterministic. The techniques we developed can be interpreted as a deterministic \textit{repairable hypergraph-to-graph connectivity-sparsifer} which transforms a hypergraph to a graph with the same vertex set and almost preserves pairwise connectivity in the vertex-failure setting. Actually, the sparsifer is somewhat stronger because it preserves the connectivity of each hyperedge individually. Thus we expect that this tool can have applications on other vertex-failures problem. See the discussion at the very end of \Cref{sec:overview construct} for details.

\paragraph{Related work.}

In \cite{henzinger2016incremental}, Henzinger and Neumann studied the
same problem but in the \emph{fully dynamic} setting where the updates
can both turn-on or turn-off vertices (not just turning-off as in
our setting). Their result is a reduction to the deletion-only version
and so our new result implies an improved solution in their setting
too. When updates are given as an online sequence and fully dynamic,
this is exactly the well-studied \emph{dynamic subgraph connectivity
}problem\emph{ }\cite{frigioni2000dynamically,chan2006dynamic,chan2011dynamic,duan2010new,duan2017faster}
whose complexity was just settled for combinatorial algorithms by the recent conditional lower bounds by \cite{jin2022tight}. In
planar graphs, a near-optimal vertex-failure connectivity oracle was
given by Borradaile et al. \cite{borradaile2012connectivity}.

\paragraph{Organization.}

In \Cref{sect:TechnicalOverview} we give an overview of our techniques. We begin with some basic notations in \Cref{sect:Preliminaries}. In \Cref{sect:VertexExpander} we discuss the vertex expander decomposition. In \Cref{sect:LowDegreeHierarchy} we introduce the low degree hierarchy in \cite{DuanP20} and show an efficient construction algorithm. In \Cref{sect:InterCom} we augment the low degree hierarchy with some affiliated structures, and we show the update and query algorithm in \Cref{sect:UpdateAndQuery}. In \Cref{sect:LowerBound}, we show the conditional lower bounds which almost match the upper bound from our oracle, and we discuss some open problems in \Cref{sect:OpenProblems}. \Cref{Appendix:OmittedProof} includes some proofs omitted in \Cref{sect:VertexExpander} and \Cref{sect:InterCom}. In \Cref{sect:DetVertexFlow} we discuss a deterministic approximate vertex-capacitated max flow algorithm. In \Cref{sect:HyperDecomp} we extend the expander decomposition results in \Cref{sect:VertexExpander} to hypergraphs.

%% file: 2-overview.tex
\section{Technical Overview}
\label{sect:TechnicalOverview}
In this section, we give an overview of our main technical contribution,
which consists of two parts below.

\subsection{Low Degree Hierarchy via Vertex Expander Decomposition}

\label{sec:overview:hierarchy}

The \emph{low degree hierarchy} of Duan and Pettie \cite{DuanP20}
is the crucial underlying structure of their oracle and also ours. We discuss it in \Cref{sect:LowDegreeHierarchy} in details.
Although its definition is quite complicated (see \Cref{def:LowDegreeHierarchy}), its purpose
is easy to describe. Informally, this hierarchy allows us to assume
that the input graph is of the following format\footnote{This is informal because this claim is slightly simplistic; the graph
$G$ is a bit more complicated (see \Cref{def:AbstractGraph} of the abstract graph).}:
\begin{enumerate}
\item $G=(L,R,E)$ is a semi-bipartite graph (i.e. there is no edges between
$R$),
\item $L$ is spanned by a path $\tau$ and $\tau$ is given to us. 
\end{enumerate}
We will see why this graph can be useful in \Cref{sec:overview construct}.%

{} For now, we focus on constructing the hierarchy itself.

Duan and Pettie \cite{DuanP20} showed that constructing the hierarchy reduces to
$O(\log n)$ calls to a decomposition related to Steiner forests.
For any terminal set $T\subseteq V$ in a graph $G$, we say that
a forest $F\subseteq E$ is a \emph{$T$-Steiner forest} if, for every
terminal pair $u,v\in T$, $u$ and $v$ are connected in $F$ iff
there are connected in $G$. The decomposition says that, for any
terminal set $T$, one can delete $|T|/2$ vertices from the graph,
so that the resulting graph contains a low-degree Steiner forest for
the remaining terminals. This is formalized as follows:
\begin{lem}
[Low-degree Steiner Forest Decomposition \cite{DuanP20}]\label{lem:lowdeg forest decomp}Let
$G=(V,E)$ be a graph with terminal set $T\subseteq V$. There is
a $O(m|T|\log|T|)$-time algorithm that computes 
\begin{itemize}
\item a vertex set $X$ of size at most $|T|/2$, and
\item a $(T\setminus X)$-Steiner forest $F$ in $G\setminus X$ with maximum
degree $4$.%

\end{itemize}
\end{lem}

\Cref{lem:lowdeg forest decomp} is based on the $(+1)$-additive-approximation
algorithm for computing minimum degree spanning trees by Furer and
Raghavachari \cite{furer1994approximating}. Calling \Cref{lem:lowdeg forest decomp}
with $T=V$ leads to $\Omega(mn\log n)$ preprocessing time in \cite{DuanP20}. 

Our key contribution is showing that, when the maximum degree of $F$
is relaxed from $4$ to $n^{o(1)}$, the running time of \Cref{lem:lowdeg forest decomp}
can be improved from $O(m|T|\log|T|)$ to $\Ohat(m)$. We completely
bypass \cite{furer1994approximating} and exploit \emph{vertex expanders}
instead. 

Recall that a \emph{vertex cut} $(L,S,R)$ is a vertex partition such
that $L,R\neq\emptyset$ and there is no edge between $L$ and $R$.
We say that a terminal set $T$ is \emph{$\phi$-linked in $G$ }if,
for any vertex cut $(L,S,R)$ in $G$, we have that $|S|\ge\phi\min\{|(L\cup S)\cap T|,|(R\cup S)\cap T|\}$.
If the whole set $V$ is \emph{$\phi$-linked} in $G$, then we say
$G$ is a \emph{$\phi$-vertex-expander}. When $G$ is a $\frac{1}{n^{o(1)}}$-vertex-expander,
we usually just say that $G$ is a vertex expander. When $T$ is $\frac{1}{n^{o(1)}}$-linked
in $G$, then we say that $T$ is \emph{well-linked} $G$.

Our starting point is the observation that any vertex expander contains
a low degree spanning tree. More generally, if $T$ is well-linked
in $G$, then there is a low degree $T$-Steiner tree. Chekuri et
al. \cite{chekuri2005multicommodity} showed how to compute such tree
in polynomial time. As this is too slow for us, we show that how to
compute it in almost-linear time:
\begin{lem}
\label{lem:low deg tree}There is a deterministic algorithm that,
given any $\phi$-linked set $T$ in $G$, compute a $T$-Steiner
tree with maximum degree $O(\frac{\log^{2}n}{\phi})$ in $\Ohat(m)$
time.
\end{lem}

\begin{proof}
[Proof sketch]We run in $\Ohat(m)$ time the cut-matching game between
vertices in $T$ where the cut player is from of \cite{ChuzhoyGLNPS20}
and the matching player is the deterministic almost-linear time vertex
max flow derived from \cite{bernstein2022deterministic} (see \Cref{lemma:DetVertexFlow}
for the proof). Since $T$ is $\phi$-linked, the game must return
an embedding ${\cal P}$ of an expander $W$ where $V(W)=T$ and maximum
degree $O(\log n)$, where ${\cal P}$ has vertex congestion at most
$O(\frac{\log n}{\phi})$. The union of all paths in ${\cal P}$ will
span $T$ can has maximum degree $O(\frac{\log^{2}n}{\phi})$, which
spans our desired $T$-Steiner tree. 
\end{proof}
Of course, we cannot assume that $T$ itself is well-linked. However,
via a standard expander decomposition framework, 
the problem
can be reduced to the well-linked case modulo removing some vertices.
To formally carry out this approach, however, we will need a fast
vertex expander decomposition algorithm, and, indeed, we show how
to achieve that in almost-linear-time (see \Cref{thm:ComplexDecomp} and \Cref{Coro:WeightedDecomposition} for details):
\begin{lem}
[Vertex Expander Decomposition with Terminals]\label{lem:overview:decomp}There
is a deterministic algorithm that, given a graph $G=(V,E)$, a terminal
set $T\subseteq V$ and a parameter $\phi$, in $\Ohat(m)$ time computes
a separator set $X$ of size $|X|\le\phi|T|n^{o(1)}$ such that, for
each connected component $G_{i}=G[V_{i}]$ in $G\setminus X$, $T\cap V_{i}$
is $\phi$-linked in $G_{i}$. In particular, if $T=V$, then $G_{i}$
is a $\phi$-vertex expander.
\end{lem}

\paragraph{Speed up the Low-degree Steiner Forest Decomposition.}

Now, we are ready to speed up \Cref{lem:lowdeg forest decomp}. Given
a terminal set $T$, we call \Cref{lem:overview:decomp} with $\phi=1/n^{o(1)}$
and obtain the vertex set $X$ of size at most $|T|/2$ such that,
for each connected component $G_{i}=G[V_{i}]$ in $G\setminus X$,
we have that $T\cap V_{i}$ is well-linked in $G_{i}$. Then, we just
apply \Cref{lem:low deg tree} on each $G_{i}$ with terminal $T\cap V_{i}$.
The union of Steiner tree on each $G_{i}$ gives us the $(T\setminus X)$-Steiner
forest $F$ with maximum degree $n^{o(1)}$ as desired. 

To compute the low-degree hierarchy itself, we call the low-degree Steiner
forest decomposition $O(\log n)$ in the same way as in \cite{DuanP20}.
The idea is simple. Set $T_{1}\gets V$ as an initial terminal set.
For any $i>1$, we invoke \Cref{lem:lowdeg forest decomp} with terminal
set $T\gets T_{i-1}$ and obtain $T_{i}\gets X$. As $|T_{i}|\le|T_{i-1}|/2$,
we have $T_{O(\log n)}=\emptyset$. The low-degree hierarchy is naturally defined based on $T_{1},T_{2},\dots$. See \Cref{sect:HierarchyConstruction} for details. 

\subsection{Optimal Deterministic Oracles }

\label{sec:overview construct}

Next, we give an overview how to obtain a vertex-failure oracle with
near-optimal space and update time. (Achieving near optimal preprocessing
and query time is relatively simpler.) This part is discussed in \Cref{sect:InterCom} and \Cref{sect:UpdateAndQuery} in details. We will start by describing
the previous deterministic construction with $\Otil(md)$ space and
$\Otil(d^{3})$ update time of \cite{DuanP20}. Then,
we show how to improve either space or update time individually. Obtaining
both improvement simultaneously turn out to be challenging. This is
the most technically involved step of this part and we will sketch
the idea at the end of the overview.

Throughout this overview, we will assume that the input graph $G=(L,R,E)$
is a semi-bipartite graph (i.e.~there is no edge between $R$). Moreover,
$L$ is spanned by a path $\tau$ and $\tau$ is given to us. This
restricted structure of $G$ will allows us to explain the key ideas
more clearly. Actually, as mentioned in \Cref{sec:overview:hierarchy},
the low-degree hierarchy ``almost'' allows us to reduce the problem
to this case. So this assumption is almost without lost of generality.

To simplify notations, assume further that $|L|=|R|=n$ and the vertex
set $D$ to be deleted always contains $d$ vertices from $L$, denoted
by $D_{L}$, and $d$ vertices from $R$, denoted by $D_{R}$. 
In the updated graph $G\setminus D$, the path $\tau$ will be split
into $d+1$ connected intervals, denoted by $\I=I_{1},\dots,I_{d+1}$.
We call each interval $I\in\I$ a \emph{left interval } (or simply an \textit{interval}) and each vertex
$v\in R$ a \emph{right vertex}. Sometimes, we use $\gamma$ to denote
a right vertex, to make it consistent with the main body of the paper.
Let $A_{\gamma}\subseteq L$ denote the ordered set of vertices adjacent to $\gamma$ with order consistent with the path $\tau$.

To support connectivity queries in $O(d)$ time, given the set $D$,
it suffices (as shown in \cite{duan2010connectivity,DuanP20})
to compute connectivity between the left intervals $\I$ in $G\setminus D$,
i.e., whether $I$ and $I'$ are in the same connected components
of $G\setminus D$ for every $I,I'\in\I$. Suppose the query is $(u,v)$.
In the easy case when both $u,v\in L$, we just identify the intervals
$I$ containing $u$ and $I'$ containing $v$ in $O(1)$ time and
answer if $I$ and $I'$ are connected in $G\setminus D$ in $O(1)$.
Now, even if $u\in R$, we can identity a non-deleted neighbor $u'$
of $u$ in $O(d)$ time as there are only $O(d)$ deleted vertices.
Note that $u'\in L$ and, say, $u'\in I$. If $v\in R$, then we similarly
find a neighbor $v'\in L$ of $v$ in $O(d)$ time, and say $v'\in I'$.
Finally, we answer if $I$ and $I'$ are connected in $G\setminus D$.
The total query time is $O(d)$ time. 

From now, we discuss how to preprocess $G$ using small space so that,
given $D$, we can compute the connectivity between the left intervals
$\I$ in $G\setminus D$ fast. Below, for any set $L',L''\subset L$
on the left, we say that $L'$ and $L''$ are \emph{adjacent} in $G$
if there exists an edge $(u,v)\in L\times L''$ or there exists a
right vertex $\gamma$ adjacent to both vertices in $L'$ and $L''$.
We say that $L'$ and $L''$ are \emph{adjacent via $(u,v)$} in the
former case and \emph{adjacent via $\gamma$ }in the latter case.

The bounds discussed below will be slightly better than our final
bounds for general graphs. More precisely, we will pay an additional
$n^{o(1)}$ factor in update time and $\omega(1)$-factor in space
when apply the ideas together with the low-degree hierarchy.

\subsubsection{Warming-up: No right vertices.}

To motivate the idea, consider the extremely simple case when $R=\emptyset$.
We will show a data structure with $O(m)$ space and $\Otil(d^{2})$
update time. In this case, the whole edge set $E(G)$ can be represented
in a 2-dimentional table $\Table=[L]\times[L]$ where the entry $(u,v)$
is the number of (multi-)edges between $u$ and $v$. Hence, during
preprocessing we can build a range counting data structure on this
table, which takes $O(|E(G)|)=O(m)$ space. Given the deletion set
$D$ and so the left intervals $\I$ are defined, we can check if
$I$ and $I'$ are adjacent in $G\setminus D$ by querying the range
counting data structure in $\Otil(1)$ time if 
\[
|E(G)\cap(I\times I')|>0.\footnote{We note that in our complete algorithm, we will treat undirected edges in such counting problem as ordered pairs, so strictly speaking, we should check if $|E(G)\cap(I\times I')|+|E(G)\cap(I'\times I)|>0$. For simplicity, we will only consider a half of each counting problem throughout the overview, and the other half is almost symmetric.}
\]
By querying between all pairs $I,I'\in\I$, we deduce the connected
components of $\I$ in $\Otil(d^{2})$ time. 

\subsubsection{Deterministic Construction of \cite{DuanP20}: $\Otil(md)$
space and $\protect\Otil(d^{3})$ update time.}

We first describe the deterministic construction of \cite{DuanP20} that uses $\Otil(md)$ space and $\tilde{O}(d^{3})$ update time. We view the high-level idea of \cite{DuanP20} for handling
right vertices as follows: Construct a \emph{restricted} (multi-)graph
$\Gaug$ whose vertex set only contains $V(\Gaug)=L$, but $\Gaug$
should captures connectivity between left vertices of $G$ \emph{even
after vertex deletions}. Since $\Gaug$ contains only left vertices,
we can use the range counting idea on $\Gaug$. 

Here is a natural way to construct $\hat{G}$. Starting with $\hat{G}=G[L]$, for each right vertex $\gamma\in R$, add a $(d+1)$-vertex connected graph on its neighbors $A_{\gamma}\subseteq L$ (or a clique if $|A_{\gamma}|\leq d$), in which the edges are called the \textit{artificial edges} of $\gamma$, denoted by $\hat{E}_{\gamma}$. Specially, we define $\hat{E}_{\gamma}=A_{\gamma}\times B_{\gamma}$, where $B_{\gamma}\subseteq A_{\gamma}$ is an arbitrary subset of size $\min\{|A_{\gamma}|,d+1\}$.\footnote{We note that the deterministic construction in \cite{DuanP20} did not define $\Ehat_{\gamma}$ as a product set $A_{\gamma}\times B_{\gamma}$, but their $\hat{E}_{\gamma}$ still forms a $\min\{|A_{\gamma},d+1|\}$-connected graph on $A_{\gamma}$. Our product-set construction will be crucial for faster update time as will be explained in \Cref{sect:UpdateImprovement}.} Observe that $\Ehat_{\gamma}$ spans $A_{\gamma}$ with the following
\emph{fault-tolerant} property: for any $D_{L}\subseteq L$, $A_{\gamma}\setminus D_{L}$
(which are obliviously connected in $G\setminus D_{L}$ via $\gamma$)
are still connected in $\Gaug\setminus D_{L}$. It follows that, for
any deleted set $D=D_{L}\cup D_{R}$, the graph $\Gaug_{D}:=\Gaug\setminus(D_{L}\cup\bigcup_{\gamma\in D_{R}}\Ehat_{\gamma})$
preserves connectivity between all left vertices of $G\setminus D$
and, in particular, between all left intervals $\I$. So it suffices
to computed components of $\I$ in $\Gaug_{D}$. We can query if $I$ and $I'$ are adjacent in $\Gaug_{D}$ by checking
if 
\[
|E(\Gaug)\cap (I\times I')|-\sum_{\gamma\in D_{R}}|\Ehat_{\gamma}\cap(I\times I')|>0,
\]
which can be computed using $1+|D_{R}|=O(d)$ range queries. By querying
all pairs, the connected components of $\I$ in $\Gaug_{D}$ can be
deduced in $\Otil(d^{3})$ time as desired. The total space for range counting structures on $E(\hat{G})$ and all $\hat{E}_{\gamma}$ is $\Otil(|E(\Gaug)|+\sum_{\gamma}|\Ehat_{\gamma}|)=\Otil(|E(G)|+\sum_{\gamma}|A_{\gamma}|d)=\Otil(md)$.

\subsubsection{Update-Time Improvement: $\Otil(md)$ space and $\protect\Otil(d^{2})$
update time. }
\label{sect:UpdateImprovement}

Then we will show how to deterministically improve the update time to $\tilde{O}(d^{2})$. To improve the update time to $\Otil(d^{2})$, the high-level idea is that we want to avoid querying for all $O(d^{2})$
pairs of left intervals $I,I'\in\I$.
We will use a \Boruvka's based algorithm instead, which can reduce the number
of queries to $\Otil(d)$. However, the queries become more complicated. 

Recall that \Boruvka's algorithm uses the ``hook and merge'' approach. There are $O(\log|{\cal I}|)$ phases. In each phase there has already been a partition of intervals ${\cal I}$ into groups $Z_{1},...,Z_{z}$ such that intervals in a same group are in a same connected component. In this phase, these groups will be further merged by looking for an adjacent group for each $Z_{k}$\footnote{Namely, find another group $Z_{k'}$ such that some $I'\in Z_{k'}$ and some $I\in Z_{k}$ are adjacent.}, which can be reduced to some \textit{batched-adjacency queries}: given $Z_{k}$ and a batch of groups $Z_{l},Z_{l+1},...,Z_{r}$, is there any group in this batch adjacent to $Z_{k}$? This query can be answered by checking if 
\begin{equation}
    \sum_{I\in Z_{k}}\sum_{I'\in Z_{l,r}}|E(\hat{G})\cap(I\times I')|-\sum_{\gamma\in D_{R}}\sum_{I\in Z_{k}}\sum_{I'\in Z_{l,r}}|\hat{E}_{\gamma}\cap(I\times I')|>0,
    \label{eq:Boruvka}
\end{equation}%
where $Z_{l,r}=\bigcup_{l\leq k'\leq r}Z_{k'}$. By further exploiting the fact that each $\Ehat_{\gamma}$ is a product set $A_{\gamma}\times B_{\gamma}$, Inequality (\ref{eq:Boruvka}) is equivalent to
\[
\sum_{I\in Z_{k}}\sum_{I'\in Z_{l,r}}|E(\hat{G})\cap(I\times I')|-\sum_{\gamma\in D_{R}}\left(\sum_{I\in Z_{k}}|A_{\gamma}\cap I|\cdot\sum_{I'\in Z_{l,r}}|B_{\gamma}\cap I'|\right)>0.
\]

The key to fast update time is the following. We arrange intervals in an ordered list ${\cal I}_{t}$ ($t$ denotes the phase number) such that ${\cal I}_{t}$ is the concatenation of $Z_{1},...,Z_{z}$. 
Now, by preprocessing a 2D-counting table $\Table_t = \{|E(\hat{G})\cap (I\times I')|\}_{I\times I'\in {\cal I}_{t}\times{\cal I}_{t}}$ with size $|{\cal I}|^{2}$, the first term can be answered by a single 2D-range counting query on $\Table_t$.
Next, by preprocessing, for each $\gamma\in D_{R}$, a 1D counting array $\textsf{A-Array}_{\gamma,t} = \{|A_{\gamma}\cap I|\}_{I\in {\cal I}_{t}}$ and another 1D array $\textsf{B-Array}_{\gamma,t} = \{|B_{\gamma}\cap I|\}_{I\in {\cal I}_{t}}$, the second term can be answered using $2|D_{R}|$ queries to these arrays $\textsf{A-Array}_{\gamma,t}$ and $\textsf{B-Array}_{\gamma,t}$.
Therefore, the query time for each batched-adjacency query is $\Otil(d)$.
For each phase, we make $\Otil(d)$ such queries which take $\Otil(d^{2})$ total query time. 
Notice that since the ordered list ${\cal I}_{t}$ will be shuffled between phases,
$\Table_t, \textsf{A-Array}_{\gamma,t}$ and $\textsf{B-Array}_{\gamma,t}$ must be rebuilt at the beginning of each phase. But this takes $\Otil(d^{2})$ total time. Over all $O(\log |\cal{I}|) $ phases, the total update time is $\Otil(d^{2})$.

We emphasize that this approach highly relies on that each $\Ehat_{\gamma}$ is a product set $A_{\gamma}\times B_{\gamma}$, which enables us to consider two dimensions separately. If $\Ehat_{\gamma}$ is less structured, a 2D counting table is required for each $\gamma\in D_{R}$ in this approach, and we cannot afford to construct them in $\Otil(d^{2})$ update time. The space is still $O(md)$ as shown above. Note that this is an inherent limitation of this approach. As long as $\Ehat_{\gamma}$
forms $\min\{|A_{\gamma}|,d+1\}$-connected graph on $A_{\gamma}$, we will have space $\Omega(md)$ in the worst case.

\subsubsection{Space Improvement: $O(m)$ space and $\protect\Otil(d^{3})$ update
time. }

Now, we will improve the space to $O(m)$, but it will bring the update time back to $\Otil(d^{3})$. This step already gives the first deterministic oracle with
near-linear space. The high-level idea is to relax the fault-tolerant property in order to get sparser artificial edges. Concretely, it is acceptable that, for some right vertices $\gamma$, its artificial edges $\Ehat_{\gamma}$ do not satisfy the fault-tolerant property, but we want the number of such vertices to be small so we can afford to ``repair'' their $\Ehat_{\gamma}$ in an additional \textit{repairing phase} of the update algorithm. We note that the Monte Carlo construction in \cite{DuanP20} also takes advantage of a similar relaxation.

Towards this goal, let $\Ehat =\bigcup_{\gamma\in R}\Ehat_\gamma$ contains all artificial edges. 
We will redefine the artificial edges $\Ehat_{\gamma}$
of $\gamma$ so that the total number of artificial edges \textit{not counting multiplicity}, denoted by $\Vert\Ehat\Vert_{0}$, is linear to $\sum_{\gamma\in R}|A_{\gamma}|=O(m)$. The construction
works by scanning through right vertices $\gamma\in R$ one by one.
Initialize $\Ehat=\emptyset$. Then, for each $\gamma$, if $(u,v)\in \Ehat$
for all $(u,v)\in A_{\gamma}\times B_{\gamma}$ (as before $B_{\gamma}\subseteq A_{\gamma}$ is an arbitrary subset with size $\min\{|A_{\gamma}|,d+1\}$), then define $\Ehat_{\gamma}=A_{\gamma}\times B_{\gamma}$. Otherwise,
if there exists $(u_{\gamma},v_{\gamma})\notin \Ehat$ where $u_{\gamma},v_{\gamma}\in A_{\gamma}$,
then set $\Ehat_{\gamma}=A_{\gamma}\times\{u_{\gamma},v_{\gamma}\}$. Lastly, we insert $\Ehat_{\gamma}$ into $\Ehat$ and proceed to the next right vertex. 
Note that although $\sum_{\gamma\in R}|\Ehat_{\gamma}|$ can be large, $\Vert\Ehat\Vert_{0}$ is bounded by the number of new artificial edges added in each round. Since there is no new edge in the former case, it follows that  $\Vert\Ehat\Vert_{0}=\sum_{\gamma\in R}|A_{\gamma}|=O(m)$. 

We then show there are only a few $\Ehat_{\gamma}$ without the fault-tolerant property. Let $R_{\fault}$ contain all $\gamma$ where $D_{L}$ contains both $u_{\gamma}$ and $v_{\gamma}$. Observe that $A_{\gamma}\setminus D_{L}$
is not connected by $\Ehat_{\gamma}\setminus D_{L}$
only if $\gamma\in R_{\fault}$. As
$D_{L}$ contains at most $|D_{L}|^{2}$ different pairs of vertices,
so $|R_{\fault}|\le|D_{L}|^{2}=O(d^{2})$. Therefore, we have that
the graph $\Gaug_{D}:=\Gaug\setminus(D_{L}\cup\bigcup_{\gamma\in D_{R}}\Ehat_{\gamma})$
still preserves preserves connectivity between all left intervals
of $G\setminus D$, except that it might miss the connectivity information
from each $\gamma\in R_{\fault}$. 

To compute the connected components of $\I$ in $G\setminus D$, we
fist compute connected components of $\I$ in $\Gaug_{D}$. This can
be done in $\Otil(d^{2})$ time as in \Cref{sect:UpdateImprovement}. But then some connected components $\I$
in $\Gaug_{D}$ may be further connected via $\gamma\in R_{\fault}$.
Therefore, we add a \emph{reparing phase} in the update algorithm. For every $I\in\I$
and $\gamma\in R_{\fault}$, we simply check $I\cap A_{\gamma}\neq\emptyset$. If there is some $I,I'\in {\cal I}$ and $\gamma\in R_{\fault}$ such that $I\cap A_{\gamma}\neq\emptyset$
and $I'\cap A_{\gamma}\neq\emptyset$, then the connected components of
$I$ and $I'$ can be merged as $I$ and $I'$ are adjacent via $\gamma$. After these queries, we can finally deduce
connected components of $\I$ in $G\setminus D$. We can check
if $I\cap A_{\gamma}\neq\emptyset$ in $\Otil(1)$ time. 
However, there are $O(d^{3})$ pairs of $I,I'\in\I$ and $\gamma\in R_{\fault}$, the repairing phase needs $\Otil(d^{3})$ time.

Therefore, we obtain an oracle with $O(m)$ space and $\Otil(d^3)$ update time.

\subsubsection{Repairing-Time Improvement: $\Otil(d^{3})$ update time with $\Otil(d^{2})$ repairing time}
\label{sect:RepairImprovement}

As a prerequisite to our final construction in \Cref{sec:overview:final}, here we will improve the running time of the repairing phase from $\Otil(d^{3})$ to $\Otil(d^{2})$. However, the total update time would  still be $\Otil(d^{3})$.
because our artificial edges $\Ehat_{\gamma}$ in this step will not be well-structured and so we cannot apply the \Boruvka's based algorithm for interval connectivity as we did in \Cref{sect:UpdateImprovement}. We will sketch how to handle this issue in \Cref{sec:overview:final}.

The high-level idea is to redesign the artificial edges $\Ehat_{\gamma}$ for each $\gamma\in R$, such that for any failure set $D$, vertices in $A_{\gamma}\setminus D_{L}$ may not be connected by $\Ehat_{\gamma}\setminus D_{L}$ itself, but they will be connected by $(\Ehat_{\gamma}\cup\hat{F}_{\gamma})\setminus D_{L}$, where $\hat{F}_{\gamma}$ is a small set of \textit{$D$-repairing edges} generated after $D$ is deleted. In fact, in our following construction, the total number of $D$-repairing edges for all $\gamma$ is bounded by $O(d^{2})$. Therefore, the repairing phase only needs to add the connectivity provided by these repairing edges and it takes $\Otil(d^{2})$ time. 

We introduce a new structure called the \textit{segmentation hierarchy}. For each $\gamma\in R$, the segmentation hierarchy ${\cal S}_{\gamma}$ is simply a recursive division of $A_{\gamma}$ like a segment tree. The elements of ${\cal S}_{\gamma}$ are \textit{segments}, which are consecutive sublists of $A_{\gamma}$. ${\cal S}_{\gamma}$ is partitioned into $r=O(\log |A_{\gamma}|)$ levels ${\cal S}_{\gamma,1},...,{\cal S}_{\gamma,r}$ such that for each level $j$, segments in ${\cal S}_{\gamma,j}$ partition the list $A_{\gamma}$. In particular, the top level has the unique segment $A_{\gamma}$ and each segment at the bottom level is a singleton list of each vertex in $A_{\gamma}$. Furthermore, for each level $j\geq 2$, each segment $S\in {\cal S}_{\gamma,j}$ is the union of two \textit{child-segments} $S_{1},S_{2}\in {\cal S}_{\gamma,j-1}$. See \Cref{def:Segmentation} for the complete definition.

The construction of $\Ehat_{\gamma}$ is similar to the previous step. We process $\gamma\in R$ one by one and initialize $\Ehat=\emptyset$. 
Let $B_\gamma$ be an arbitrary subset of $A_\gamma$ of size $\min\{d+1,|A_\gamma|\}$.
For each $\gamma\in R$ if there are less than $|{\cal S}_{\gamma}|$ pairs in $(A_{\gamma}\times B_{\gamma})\setminus \Ehat$, we set $\Ehat_{\gamma}=A_{\gamma}\times B_{\gamma}$ as usual. Otherwise we assign a distinct pair $(u_{S},v_{S})\in (A_{\gamma}\times B_{\gamma})\setminus \Ehat$ to each segment $S\in {\cal S}_{\gamma}$ and let $\Ehat_{\gamma}=\sum_{S\in{\cal S}_{\gamma}}S\times\{u_{S},v_{S}\}$. Lastly, insert $\Ehat_{\gamma}$ to $\Ehat$ and proceed to the next $\gamma$. Now, we claim that the total number of artificial edges (without multiplicity) is $\| \Ehat \|_0 = \Otil(m)$ since the former case will bring at most $|{\cal S}_{\gamma}|$ new edges and in the latter case the number of new edges is at most $|\Ehat_{\gamma}|=O(\sum_{S\in{\cal S}_{\gamma}}|S|)=\Otil(|A_{\gamma}|)$. So $\| \Ehat \|_0 = \sum_{\gamma \in R} \Otil(|A_\gamma|) = \Otil(m)$ and so the total space is $O(\| \Ehat \|_0) = \Otil(m)$.

The repairing algorithm is based on the segmentation hierarchy. 
Let ${\cal S}_{\fault}$ contain all segments $S$ (in ${\cal S}_{\gamma}$ for \textit{all} $\gamma\in R$) such that $u_{S},v_{S}\in D_{L}$. 
By the same argument, $S\setminus D_{L}$ is not connected by $\Ehat_{\gamma}\setminus D_{L}$ only if $S\in{\cal S}_{\fault}$, and we have $|{\cal S}_{\fault}|=O(d^{2})$.  For each $\gamma\in R$ with nonempty ${\cal S}_{\fault}\cap{\cal S}_{\gamma}$, suppose that we are given a non-deleted vertex $v_{\gamma}\in A_{\gamma}\setminus D_{L}$.\footnote{In the complete algorithm, the $v_{\gamma}$ for each $\gamma\in R$ with ${\cal S}_{\fault}\cap{\cal S}_{\gamma}\neq\emptyset$ can be computed in totally $\Otil(d^2)$ time.} Then we construct $D$-repairing edges $\hat{F}_{\gamma}$ by, for each $S\in{\cal S}_{\fault}\cap {\cal S}_{\gamma}$, connecting its two child-segments $S_{1},S_{2}$ to $v_{\gamma}$. Concretely, we add to $\hat{F}_{\gamma}$ four $D$-repairing edges connecting $v_{\gamma}$ to each of $u_{S_{1}},v_{S_{1}},u_{S_{2}},v_{S_{2}}$ for each $S\in {\cal S}_{\fault}\cap {\cal S}_{\gamma}$. Thus there are totally $O(|{\cal S}_{\fault}|)=O(d^{2})$ $D$-repairing edges and the repairing time is $\Otil(d^{2})$.

To see the correctness, we say that a segment $S$ is \emph{maximal} if $S \in \cS_\gamma \setminus \cS_{\fault}$ but its parent-segment is in $\cS_\gamma \cap \cS_{\fault}$. Observe that maximal segments cover $A_\gamma \setminus D$, the whole non-deleted part of $A_\gamma$. Now, for each maximal segment $S$, we have that either $u_S \notin D$ or $v_S \notin D$. Say $u_S \notin D$,  so  $u_S$ is connected to $v_\gamma$ via $\hat{F}_\gamma$ and the whole $S$ is connected to $u_S$ via $\Ehat_\gamma$. Therefore, indeed we have that vertices in  $A_\gamma \setminus D$ will be connected via $(\Ehat_\gamma \cap \Fhat_\gamma)\setminus D$ as promised.

\subsubsection{Final Algorithm: $\Otil(m)$ space and $\Otil(d^{2})$ update
time. }
\label{sec:overview:final}

The last step to our final algorithm is to make $\Ehat_{\gamma}$ in \Cref{sect:RepairImprovement} well-structured. We cannot hope each $\Ehat_{\gamma}$ forms a perfect product set because of the segmentation hierarchy. Yet, we can still design it to have some good structural property so that it is compatible with the \Boruvka's based algorithm in \Cref{sect:UpdateImprovement}, which finally leads us to $\Otil(d^{2})$ update time and $\Otil(m)$ space simultaneously.

The construction algorithm of artificial edges is relatively simple. We only summarize the outputs here and see \Cref{sect:ArtificialEdges} for details. For each right vertex $\gamma\in R$, we still assign a pair $(u_{S},v_{S})\in A_{\gamma}\times A_{\gamma}$ to each segment $S\in{\cal S}_{\gamma}$, and the actual $\Ehat_{\gamma}$ will be
\[
\Ehat_{\gamma}=A_{\gamma}\times B_{\gamma}+\sum_{S\in{\cal S}_{\gamma}}S\times \{u_{S}\},\footnote{The add operations denote the multiset union operations.}
\]
where $B_{\gamma}$ is designed so that there will be $\| \Ehat \|_0 = \Otil(m)$ distinct artificial edges using similar argument as in \Cref{sect:RepairImprovement}. So the space bound is $\Otil(m)$

Now, we discuss how this structure enables $\Otil(d^2)$ update time.
For each $S\in {\cal S}_{\gamma}$, we call $u_{S}$ the \textit{witness} of $S$. Furthermore, let $f:{\cal S}_{\gamma}\to A_{\gamma}$ be the witness function defined by $f(S)=u_{S}$ for each $S\in{\cal S}_{\gamma}$. Our construction ensures that $f$ has the following \textit{monotonically increasing} property on each level ${\cal S}_{\gamma,j}$: for each $S,S'\in{\cal S}_{\gamma,j}$ such that $S$ is located before $S'$ on $A_{\gamma}$, $f(S)$ is located before $f(S')$ according the order of vertices in $A_{\gamma}$. 

Answering a batched-adjacency query in \Cref{sect:UpdateImprovement} given $Z_{k}$ and $Z_{l,r}$ now becomes more involved. We sketch the ideas below and see \Cref{sect:IntervalConnectivity} for details. The product term $A_{\gamma}\times B_{\gamma}$ is easy to deal with using the idea in \Cref{sect:UpdateImprovement}. To handle the term $\sum_{S\in{\cal S}_{\gamma}}S\times f(S)$, we first write down its contribution to the second term of Inequality (\ref{eq:Boruvka}) for fixed $\gamma\in D_{R}$, level $j$ and $I\in Z_{k}$ as
\[
Q(\gamma,j,I)=\left|\left(\sum_{S\in{\cal S}_{\gamma,j}}S\times f(S)\right)\cap \left(\bigcup_{I'\in Z_{l,r}}I'\times I \right)\right|.\footnote{For simplicity, in $Q(\gamma,j,I)$ we switch $I$ and $\bigcup_{i'\in Z_{l,r}}I'$ of the actual contribution $|(\sum_{S\in{\cal S}_{\gamma,j}}S\times f(S))\cap (I\times \bigcup_{I'\in Z_{l,r}}I')|$, which can be computed using the same idea of computing $Q(\gamma,j,I)$. Actually, in the complete algorithm, both contributions (before and after switching) will be computed as shown in \Cref{sect:ComputingDelta3}.}
\]
In fact, there will be totally at most $\Otil(d^{2})$ queries on $Q(\gamma,j,I)$, summing over different $\gamma,j,I$ in all \Boruvka's phases. In what follows, we discuss how to compute $Q(\gamma,j,I)$ for fixed $\gamma,j,I$ in $\Otil(1)$ time.

Recall that $I$ is an interval on the path $\tau$. 
Consider the set of segments $\{ S \in \cS_{\gamma,j} \mid f(S) \in I\}$, i.e., the set of level-$j$ segments whose witnesses are in $I$. 
By the monotonically increasing property of $f$, observe that the union of these segments form a consecutive sublist of $A_{\gamma}$, denoted by $A_{\gamma}(v_{L},v_{R})$. Then the contribution can be rewritten as \[Q(\gamma,j,I)=\left|A_{\gamma}(v_{L},v_{R})\cap\bigcup_{I'\in Z_{l,r}}I'\right|.\] 
This query can be further reduced to one 3D range counting query by exploiting that intervals $I'$ in $Z_{l,r}$ locate consecutively in the order ${\cal I}_{t}$ for the $t$-th \Boruvka's phase. Therefore, $Q(\gamma,j,I)$ can be indeed queried in $\Otil(1)$ time.
Moreover, the required data structures can be constructed in $\Otil(d^{2})$ time at the beginning of each \Boruvka's phase.

\paragraph{Perspective: Repairable Hypergraph-to-Graph Connectivity-Sparsifiers.}

Lastly we discuss our techniques in \Cref{sec:overview construct} in a different perspective. If we corresponds left vertices and right vertices in the semi-bipartite graph to vertices and hyperedges in a hypergraph respectively, the techniques on constructing artificial edges and $D$-repairing edges can be viewed as some kind of sparsifier which transform a hypergraph $H$ to a graph $G$.
Concretely, the transformation will keep the vertex set unchanged, and for each hyperedge $e\in E(H)$, add an edge set $\Ehat(e)$ to $E(G)$. Namely we will have $V(G)=V(H)$ and $E(G)=\bigcup_{e\in E(H)}\Ehat(e)$. Moreover, the sparsifier has the following properties.
\begin{itemize}
    \item (Vertex Restricted) The vertex set is unchanged.
    \item (Sparse) The number of edges in $G$ is nearly linear to the total size of edges in $H$. Namely $|E(G)|=\Otil(\sum_{e\in E(H)}|e|)$, where we use the same notation $e$ to denote the set of $e$'s endpoints. %
    \item (Fault-Tolerant/Repairable) Given any failure set $D\subseteq V(G)$ with $|D|=d$, by augmenting each $\hat{E}(e)$ with a set $\hat{F}(e)$ of repairing edges, it will be guaranteed that for each hyperedge $e\in E(H)$, vertices in $e$ are still connected by $\Ehat(e)\cup\hat{F}(e)$. Moreover, the total size of $\hat{F}(e)$ summing over all $e\in E(H)$ is $\Otil(d^{2})$, and all nonempty $\hat{F}(e)$ can be computed in totally $\Otil(d^{2})$ time.
\end{itemize}

We apply this tool to obtain an nearly optimal vertex-failure connectivity oracle. Furthermore, we note that this sparsifier is somewhat stronger than maintaining just pairwise connectivity because the connectivity of each hyperedge $e$ is perserved by $\hat{E}(e)\cup\hat{F}(e)$ individually. Thus we expect that this tool or its variant can have applications on other vertex-failure problems.

%% file: 3-body.tex
\section{Preliminaries}
\label{sect:Preliminaries}

Throughout the paper, we use the standard graph theoretic notation. For any graph, we use $V(\cdot)$ and $E(\cdot)$ to denote its vertex set and edge set respectively. If there is no other specification, $G$ will denote the graph on which we build the oracle and we set $n=|V(G)|$ and $m=|E(G)|$. Given a graph $G=(V(G),E(G))$, for any $S\subseteq V(G)$, we let $G[S]$ denote the subgraph of $G$ induced by vertices $S$, and let $N_{G}(S)=\{u\mid (u,v)\in E(G),u\notin S,v\in S\}$ denote the set of vertices adjacent to $S$. Also, for any $S\subseteq V(G)$, we use $G\setminus S$ to denote the graph after removing vertices in $S$ and edges incident to them. Similarly, for any $F\subseteq E(G)$, $G\setminus F$ denote the graph after removing edges in $F$.

In this paper, we will use \textit{ordered lists} frequently. An ordered list $L$ is a list of totally ordered elements, and it inherits properties and operations from sets. We say for each $1\leq k\leq |L|$ the $k$-th element (denoted by $L(k)$) has \textit{rank} $k$ on $L$ or has \textit{$L$-rank} $k$, and let $L(k_{1},k_{2})$ denote the consecutive sublist from rank $k_{1}$ to $k_{2}$. We define an operations on ordered list called \textit{restrictions} as follows. Assume we have an ordered list $L$ as the base. Let $L'$ be a sublist (maybe inconsecutive) of $L$ and let $\bar{L}$ be a consecutive sublist on $L$. The \textit{restriction of $\bar{L}$ on $L'$}, denoted by $\bar{L}\cap L'$, is a consecutive sublist on $L'$ that contains elements in $\bar{L}\cap L'$. When $L'$ is stored explicitly (with the $L$-rank of each element), the restriction $\bar{L}\cap L'$ will be stored as two integers in our algorithm, namely the $L'$-ranks of the leftmost and rightmost element of $\bar{L}\cap L'$, and they can be computed in $O(\log|L'|)$ time via the binary search. 

We also use the notion of \textit{multisets}. A multiset $M$ is a set but allows repeated elements. We let $|M|$ denote its size counting multiplicity, while we let $\Vert M\Vert_{0}$ denote its size without multiplicity. For two multisets $M_{1}$ and $M_{2}$, we call their multiset-union as their \textit{sum}, denoted by $M_{1}+M_{2}$, and we use $\sum_{i=1}^{k}M_{i}$ denote the multiset-union of several multisets $M_{1},...,M_{k}$. The intersection of a multiset $M$ and a set $S$ is still a multiset denoted by $M\cap S$. Concretely, for each element of $M$, it will be kept in $M\cap S$ if it is in $S$, otherwise it will be dropped.

\section{Weighted Vertex Expander Decomposition}
\label{sect:VertexExpander}

In this section, we introduce an almost-linear time algorithm to compute a \textit{vertex expander decomposition} of an undirected graph $G$ with non-negative weight function $w(v)$ and demand function $\d(v)$ on its vertices $v\in V(G)$. For simplicity, we denote the total weight for a set $S\subseteq V(G)$ by $w(S)=\sum_{v\in S}w(v)$, and similarly the total demand for $S$ is $\d(S)=\sum_{v\in S}\d(v)$\footnote{We use such notations for other weight functions throughout the paper.}. The term \textit{vertex expander} refers to some ``highly connected'' graph in the sense that there is no sparse vertex cut, or strictly speaking, it has high \textit{vertex expansion} (see \Cref{def:VertexExpansion}).

\begin{definition}[Weighted Vertex Expansion]
Let $G=(V(G),E(G),w,\d)$ be a graph with weight $w(v)\geq 0$ and demand $\d(v)\geq 0$ for vertex each $v\in V(G)$. A \textit{vertex cut} of $G$ is a partition $(L,S,R)$ of $V(G)$ such that $(u,v)\notin E(G)$ for all $u\in L,v\in R$. For any vertex cut $(L,S,R)$ of $G$, if $\d(L\cup S)>0$ and $\d(R\cup S)>0$, it has \textit{vertex expansion} 
\[
h_{G}(L,S,R)=\frac{w(S)}{\min\{\d(L\cup S),\d(R\cup S)\}}.
\]
The vertex expansion of $G$, denoted by $h(G)$, is the minimum vertex expansion of any vertex cut $(L,S,R)$ of $G$ with $\d(L\cup S),\d(R\cup S)>0$.
\label{def:VertexExpansion}
\end{definition}

In a vertex expander decomposition of $G$, $G$ will be separated by a small fraction vertices called the \textit{separator}, such that after removing the separator, each of the remaining connected components is a vertex expander. See \Cref{def:Decomp} for details.

\begin{definition}[Weighted Vertex Expander Decomposition]
Let $G=(V(G),E(G),w,\d)$ be a graph with weight $w(v)\geq 0$ and demand $\d(v)\geq 0$ for each vertex $v\in V(G)$. For parameters $\epsilon>\phi> 0$, an \textit{$(\epsilon,\phi)$-vertex expander decomposition} of $G$ is a collection ${\cal G}=(V_{1},...,V_{k})$ of disjoint subsets of $V(G)$. Let $X=V(G)\setminus (\bigcup_{i} V_{i}$) denote the \textit{separator}. Then ${\cal G}$ further satisfies the following.
\begin{itemize}
    \item %
    for each pair of distinct $V_{i},V_{i'}\in{\cal G}$, there is no edge connecting $V_{i}$ and $V_{i'}$,
    \item for all $V_{i}\in{\cal G}$, $h(G[V_{i}])\geq\phi$, and
    \item $w(X)\leq \epsilon\cdot\d(V(G))$.
\end{itemize}
\label{def:Decomp}
\end{definition}

This section is dedicated to show the fast deterministic vertex expander decomposition algorithm as stated in \Cref{thm:ComplexDecomp}. We note that there has been a long line of works on fast expander decomposition e.g.~\cite{ST04,NS17,SW19,ChuzhoyGLNPS20}. The first deterministic almost-linear time algorithm for expander decomposition in the edge cut setting has been shown in \cite{ChuzhoyGLNPS20} and later it was generalized to weighted graphs in \cite{LiS21}. We follow the approach in \cite{ChuzhoyGLNPS20} and generalize it to the vertex expander and weighted setting. Although the required techniques and tools have been developed, we remark that to the best of our knowledge, an efficient algorithm on vertex expander decomposition has never been shown explicitly in the literature before.

\begin{theorem}
Let $G=(V(G),E(G),w,\d)$ be an $n$-vertex $m$-edge graph with for each $v\in V(G)$, integral weight $1\leq w(v)\leq U$ and integral demand $0\leq \d(v)\leq w(v)$. Given parameters $0<\epsilon\leq 1$, $1\leq r\leq \lfloor\log_{20}\d(V(G))\rfloor$, there is an algorithm that computes an $(\epsilon,\phi)$-vertex expander decomposition of $G$, where $\phi=\epsilon/(\log^{O(r)}U\log^{O(r^{5})}(m\log U))$. The algorithm can be deterministic with running time $O(m^{1+o(1)+O(1/r)}\cdot U^{1+O(1/r)}\cdot\log^{O(r^{4})}(mU)/\phi)$, or randomized with running time $O(m^{1+o(1)}\cdot(mU)^{O(1/r)}\cdot\log^{O(r^{4})}(mU))$ and high correct probability.
\label{thm:ComplexDecomp}
\end{theorem}

\begin{corollary}
Let $G=(V(G),E(G),w,\d)$ be an $n$-vertex $m$-edge graph with non-negative real weight $w(v)\geq 0$ and demand $\d(v)\geq 0$ for each vertex $v\in V(G)$. Let $U$ be the ratio bound of both weight and demand such that $\frac{\max w(v)}{\min_{w(v)>0}w(v)},\frac{\max \d(v)}{\min_{\d(v)>0}\d(v)}\leq U$. Given parameter $\epsilon > 0$, there is an algorithm that computes an $(\epsilon,\phi)$-vertex expander decomposition of $G$, where $\phi=\epsilon/(m^{o(1)}U^{o(1)})$. The algorithm can be deterministic with running time $O(m^{1+o(1)}\cdot U^{1+o(1)}/\phi)$, or randomized with running time $O(m^{1+o(1)}\cdot U^{o(1)})$ and high correct probability.
\label{Coro:WeightedDecomposition}
\end{corollary}

\begin{remark}
In the following of \Cref{sect:VertexExpander}, we only discuss the less efficient deterministic algorithm in \Cref{thm:ComplexDecomp}. It suffices for our application in \Cref{sect:LowDegreeHierarchy} because we have $\phi=1/n^{o(1)}$ and $U=O(1)$ in that scenario. We emphasize that we can improve the running time at a cost of making the algorithm randomized, by simply substituting the approximate vertex-capacitaed max flow algorithm in \Cref{lemma:DetVertexFlow} with the faster randomized version in \Cref{lemma:RandVertexFlow}.
\label{remark:DetDecomp}
\end{remark}

\Cref{thm:ComplexDecomp} follows the natural divide-and-conquer paradigm which invokes an efficient algorithm in \Cref{lemma:CutOrCertify} as a subroutine. Roughly speaking, the subroutine will either return a balanced sparse vertex cut (a vertex cut is balanced if it has large demands on both sides), or exhibit a large vertex expander separated from the other part by a sparse vertex cut. In both cases, the graph size of the recursive calls will shrink by a constant factor. \Cref{Coro:WeightedDecomposition} is a corollary of \Cref{thm:ComplexDecomp} where the condition on weights and demands is relaxed and the formula of time complexity is simplified. The complete proof of \Cref{thm:ComplexDecomp} and \Cref{Coro:WeightedDecomposition} is shown in \Cref{proof:ComplexDecomp}.

\begin{lemma}
Let $G=(V(G),E(G),w,\d)$ be an $n$-vertex $m$-edge graph with, for each $v\in V(G)$, integral weight $1\leq w(v)\leq U$ and integral demand $0\leq \d(v)\leq w(v)$. Given parameters $0<\epsilon\leq 1$ and $1\leq r\leq \lfloor\log_{20}\d(V(G))\rfloor$, there is a deterministic algorithm that computes a vertex cut $(L,S,R)$ (possibly $L=S=\emptyset$) of $G$ with $w(S)\leq \epsilon\cdot\d(L\cup S)$ which further satisfies
\begin{itemize}
    \item either $\d(L\cup S),\d(R\cup S)\geq \d(V(G))/3$; or
    \item $\d(R)\geq \d(V(G))/2$ and $h(G[R])\geq \phi$ for some $\phi=\epsilon/(\log^{O(r)}U\log^{O(r^{5})}(m\log U))$.
\end{itemize}
The running time of this algorithm is $O(m^{1+o(1)+O(1/r)}\cdot U^{1+O(1/r)}\cdot\log^{O(r^{4})}(mU)/\phi)$.
\label{lemma:CutOrCertify}
\end{lemma}

\Cref{lemma:CutOrCertify} is an application of the cut-matching-game framework. We first introduce the cut-matching game in \Cref{sect:WeightedCutMatchingGame}. Then in \Cref{sect:BalancedSparseCut}, we exploit the cut-matching-game framework to solve the \textit{most-balanced vertex cut} problem (see \Cref{lemma:MostBalanceCut}), which is a critical step getting to \Cref{lemma:CutOrCertify}. The complete proof of \Cref{lemma:CutOrCertify} using \Cref{lemma:MostBalanceCut} is in \Cref{proof:BalCutPrune}.

\subsection{The Cut-Matching Game}
\label{sect:WeightedCutMatchingGame}

The cut-matching game is an interactive process between a \textit{cut player} and a \textit{matching player} to construct an expander. The goal of the cut player is to ensure the game will terminate with an expander in a small number of rounds, while the matching player can act arbitrarily or even adversarially to slow down this process. The game was first introduced by Khandekar, Rao and Variranzi \cite{KRV09} to design fast algorithm to compute sparse cuts and balanced cuts. 

Before we describe the cut-matching game, we emphasize that the game is designed under the edge cut setting, where whether an edge cut is sparse or not is measured by its \textit{sparsity} as defined in \Cref{def:sparsity}, and a graph is an \textit{expander} if there is no sparse edge cut. 

\begin{definition}[Sparsity]
Let $H=(V(H),E(H),w_{H},\d)$ be a graph with weight $w(e)\geq 0$ for each edge $e\in E(H)$ and demand $\d(v)\geq 0$ for each vertex $v\in V(H)$. For any edge cut $(S,V(H)\setminus S)$ with $0<\d(S)<\d(V(H))$, it has \textit{sparsity}
\[
\Psi_{H}(S,V(H)\setminus S)=\frac{w_{H}(E_{H}(S,V(H)\setminus S))}{\min\{\d(S),\d(V(H)\setminus S)\}},
\]
where $E_{H}(S,V(H)\setminus S)$ is the set of edges crossing the cut $(S,V(H)\setminus S)$ and $w_{H}(E')=\sum_{e\in E'}w(e)$ for any $E'\subseteq E(H)$. The sparsity of $H$, denoted by $\Psi(H)$, is the minimum sparsity of any edge cut $(S,V(H)\setminus S)$ with $0<\d(S)<\d(V(H))$.
\label{def:sparsity}
\end{definition}

The game we exploited is actually a variant by Khandekar et al. \cite{KKOV07}, and we will use it in a weighted setting. Specifically, the cut-matching game starts with a graph $H=(V(H),E(H),w_{H},\d)$, where each vertex $v\in V(G)$ has integral demands $0\leq\d(v)\leq U$ and the edge set $E(H)$ is empty. During the game, $E(H)$ will grow and each edge $e\in E(H)$ will have integral weight $1\leq w_{H}(e)\leq U$. When the game ends, $H$ will be guaranteed to be an expander with proper sparsity.

The cut-matching game has several rounds. In the $i$-th round, the cut player will choose a cut $(A_{i},B_{i})$ of $V(H)$ with $\d(B_{i})\geq \d(A_{i})\geq \d(V(H))/4$ and $w_{H}(E_{H}(A_{i},B_{i}))\leq \d(V(H))/100$. If no such cut exists, it can be proved that there must be a set $S\subseteq V(H)$ such that $\d(S)\geq \d(V(H))/2$ and $H[S]$ can be certificated as an expander. The cut player will choose the cut $(A_{i},B_{i})=(V(H)\setminus S,S)$, and the game will terminate right after this round. Given the current graph $H$ and the cut $(A_{i},B_{i})$ chosen by the cut player, the matching player should choose an arbitrary $\d$-matching $M_{i}$ with value $M_{i}(A_{i})=\d(A_{i})$ (see \Cref{def:dMatching}) between $A_{i}$ and $B_{i}$, and then this matching $M_{i}$ will be added to $E(H)$.

\begin{definition}[$\d$-matching]
For two sets $A,B$ of vertices with demand $\d(v)\geq 0$ for each $v\in A\cup B$ (assuming $\d(A)\leq \d(B)$), a $\d$-matching between $A$ and $B$ is an edge-weighted bipartite graph of $A$ and $B$ such that for each $v\in A\cup B$, the weighted degree (denoted by $M(v)$) is at most $\d(v)$. The value of $M$ is the total weight of edges, which is equal to $M(A)=\sum_{v\in A}M(v)$. We say this matching is \textit{perfect} if $M(A)=\d(A)$.
\label{def:dMatching}
\end{definition}

\Cref{lemma:CMGRound} follows the analysis in \cite{KKOV07}. Strictly speaking, \cite{KKOV07} only showed that the cut-matching game on an $n$-vertex \textit{unweighted} graph has at most $O(\log n)$ rounds. However, observe that the game on a weighted graph $H$ can also run on an unweighted graph $H'$ in which $w(v)$ copies are created for each original vertex $v\in V(H)$, so \Cref{lemma:CMGRound} follows.

\begin{lemma}
The number of rounds of the cut-matching game will be at most $O(\log(nU))$.
\label{lemma:CMGRound}
\end{lemma}

\

\noindent\textbf{The Cut Player.} We will use the algorithm in \Cref{thm:DetCutPlayer} from \cite{LiS21} to implement the cut player.%

\begin{lemma}[\textit{Theorem 2.4 in \cite{LiS21}}]
Let $H=(V(H),E(H),w_{H},\d)$ be an $n$-vertex $m$-edge graph with integral edge weight $1\leq w_{H}(e)\leq U$ for all $e\in E(H)$ and integral demand $0\leq d(v)\leq U$ for all $v\in V(H)$. Given a parameter $r\geq 1$, there is a deterministic algorithm that returns either
\begin{itemize}
    \item a cut $(A,B)$ in $H$ with $\d(A),\d(B)\geq \d(V(H))/4$ and $w_{H}(E_{H}(A,B))\leq \d(V(H))/100$; or
    \item a set $S\subseteq V(H)$ with $\d(S)\geq \d(V(H))/2$ such that $\Psi(H[S])\geq 1/\log^{O(r^{4})}m$.
\end{itemize}
The running time of this algorithm is $O(m\cdot(mU)^{O(1/r)}\log^{O(r^{2})}(mU))$.
\label{thm:DetCutPlayer}
\end{lemma}

\Cref{prop:CMGExpansion} concludes that the game ends with an expander, which is an easy observation from the fact that in the last round of the game, a large subgraph $H[S]$ has been an expander by \Cref{thm:DetCutPlayer} and the matching player connects vertices in $V(H)\setminus S$ to $S$ by a $\d$-matching with value $|V(H)\setminus S|$.

\begin{proposition}
If the cut player is implemented by the algorithm in \Cref{thm:DetCutPlayer} with parameter $r\geq 1$, the final graph $H$ has sparsity $\Psi(H)\geq 1/\log^{O(r^{4})}m$, where $m$ is the final number of edges.
\label{prop:CMGExpansion}
\end{proposition}

\

\noindent\textbf{The Matching Player.} As shown in \Cref{sect:BalancedSparseCut}, in order to solve the most-balanced sparse vertex cut problem, we will try to \textit{embed} (see \Cref{def:embedding}) an expander generated by the cut-matching game into the graph $G$ we want to decompose, and we hope the embedding has low \textit{vertex congestion}. Since the expander is the union of a small number of matchings decided by the matching player, the strategy of the matching player is to somehow compute a matching and its low-vertex-congestion embedding simultaneously\footnote{To be precise, the matching player in \Cref{thm:RandMatchingPlayer} will compute either a matching which can be almost embedded into $G$, or some balanced sparse vertex vertex cut. This suffices to show the result in \Cref{sect:BalancedSparseCut}.} in each round. Because the number of rounds is small, we can simply takes the union of embeddings of matchings as the desired embedding of the expander. 

Analogous to \cite{NS17,ChuzhoyGLNPS20}, we implement the matching player (see \Cref{thm:RandMatchingPlayer}) using approximate vertex-capacitated max flow algorithm as a subroutine.
We discuss the approximate vertex-capacitated max flow algorithm in \Cref{sect:DetVertexFlow} and leave the complete proof of \Cref{thm:RandMatchingPlayer} in \Cref{proof:DetMatchingPlayer}.

\begin{definition}[Embedding]
Let $G=(V(G),E(G),w)$ be a graph with positive integral vertex weights $w$, and let $H=(V(H),E(H),w_{H})$ be an graph with $V(H)=V(G)$ and positive integral edge weights $w_{H}$. An \textit{embedding of $H$ into $G$} is a collection ${\cal P}$ of unweighted paths in $G$, where for each $e=(u,v)\in E(H)$, there are exactly $w_{H}(e)$ paths in ${\cal P}$ connecting $u$ and $v$. If there exists a value $\eta$ such that each vertex $v\in V(G)$ appears in at most $\eta\cdot w(v)$ paths in ${\cal P}$, we say that the embedding has \textit{vertex-congestion} $\eta$. We use $E({\cal P})$ to denote edges in $E(G)$ involved in ${\cal P}$.
\label{def:embedding}
\end{definition}

\begin{lemma}
Let $G=(V(G),E(G),w,\d)$ be an $n$-vertex $m$-edge graph with, for each $v\in V(G)$, integral weights $1\leq w(v)\leq U$ and integral demands $0\leq \d(v)\leq w(v)$. Let $(A,B)$ be an edge cut of $G$ with $\d(A)\leq \d(B)$. Given $z\geq 0,0<\phi\leq 1$, there is a deterministic algorithm that computes
\begin{itemize}
    \item either a $\d$-matching $M$ between $A$ and $B$ with value at least $\d(A)-z$ such that there is an embedding ${\cal P}$ of $M$ into $G$ with vertex-congestion $O(\log (nU)/\phi)$; or
    \item a vertex cut $(L,S,R)$ in $G$ with $\d(L),\d(R)> z$ and $h_{G}(L,S,R)\leq \phi$. 
\end{itemize}
The running time of this algorithm is $O(m^{1+o(1)}U\log U/\phi)$. If the output is the $\d$-matching $M$, the algorithm can further output the set $E({\cal P})$ of ${\cal P}$.
\label{thm:RandMatchingPlayer}
\end{lemma}

\subsection{Most-Balanced Sparse Vertex Cut}
\label{sect:BalancedSparseCut}

In a most-balanced sparse vertex cut problem, the goal is to compute a sparse vertex cut $(L,S,R)$ with vertex expansion at most a given parameter $\phi$ such that it is the most-balanced in the sense that the demand of the smaller side, i.e. $\min\{\d(L),\d(R)\}$, is maximized. This problem in the edge cut setting has been studied in \cite{ST04,KRV09,NS17,Wul17}. %
As shown in these works, there exists an efficient approximation algorithm for the bi-criteria decision version of this problem. By applying the same techniques to the vertex cut setting, we can obtain \Cref{lemma:MostBalanceCut}.

\begin{lemma}
Let $G=(V(G),E(G),w,\d)$ be an $n$-vertex $m$-edge graph with, for each $v\in V(G)$, integral weight $1\leq w(v)\leq U$ and integral demand $0\leq \d(v)\leq w(v)$. Given parameters $0<\phi\leq 1,z\geq0$ and $r\geq 1$, there is a deterministic algorithm that either
\begin{itemize}
    \item returns a vertex cut $(L,S,R)$ of $G$ with $\d(R)\geq\d(L)> z$ and $h_{G}(L,S,R)\leq\phi$; or
    \item certifies that every vertex cut $(L,S,R)$ with $\d(L\cup S),\d(R\cup S)\geq \alpha z\log U\log^{\alpha r^{4}}(m\log U)$ has $h_{G}(L,S,R)\geq \phi/(\alpha\log^{2} U\log^{\alpha r^{4}}(m\log U))$, where $\alpha>0$ is a universal constant.
\end{itemize}
The running time of this algorithm is $O(m^{1+o(1)+O(1/r)}\cdot U^{1+O(1/r)}\cdot\log^{O(r^{2})}(mU)/\phi)$.
\label{lemma:MostBalanceCut}
\end{lemma}

The key idea is trying to embed an expander into the graph $G$ with low vertex congestion via the cut-matching game. If we fail to get the embedding because the matching player throws out a balanced sparse cut in some round, then this cut is an acceptable output of \Cref{lemma:MostBalanceCut}. Otherwise, we hope to get such an embedding, so that the algorithm can also terminate because intuitively, a low-vertex-congestion embedding of an expander into $G$ can certify that $G$ is a vertex expander. 

However, the matching player in \Cref{thm:RandMatchingPlayer} is somewhat weaker. In the latter case, it only computes in each round a non-perfect (but large) matching $M_{i}$ embeddable into $G$. Equivalently, if we augment $M_{i}$ into a perfect matching $M_{i}\cup F_{i}$ by adding some \textit{fake edges} $F_{i}$, finally we can only embed a ``nearly expander'' $\hat{H}=H-F$ into $G$ where $H=\bigcup_{i}M_{i}\cup F_{i}$ is an expander and $F=\bigcup_{i} F_{i}$ the set of a small number of fake edges. Fortunately, this suffices to prove \Cref{lemma:MostBalanceCut} because intuitively, all balanced cuts in the ``nearly expander'' $\hat{H}$ are not sparse, and the low-vertex-congestion embedding of $\hat{H}$ into $G$ can certify that all balanced vertex cuts in $G$ have high vertex expansion as desired. 

\Cref{coro:EmbeddingOfExpander} is a corollary of \Cref{lemma:MostBalanceCut} where we set the balanced parameter $z=0$. It will be used in \Cref{sect:LowDegreeHierarchy} as a subroutine. The complete proof of \Cref{lemma:MostBalanceCut} and \Cref{coro:EmbeddingOfExpander} is in \Cref{proof:MostBalanceCut}.

\begin{corollary}
Let $G=(V(G),E(G),w,\d)$ be an $n$-vertex $m$-edge graph with, for each $v\in V(G)$, integral weight $1\leq w(v)\leq U$ and integral $0\leq \d(v)\leq w(v)$. Given parameters $0<\phi\leq 1$, there is a deterministic algorithm that either
\begin{itemize}
    \item returns a vertex cut $(L,S,R)$ of $G$ with $\d(R)\geq\d(L)>0$ and $h_{G}(L,S,R)\leq\phi$; or
    \item certifies that there is an embedding ${\cal P}$ of an expander $H$ with the same demand $\d$ and $\Psi(H)\geq 1/(m\log U)^{o(1)}$ into $G$ with vertex-congestion at most $O(\frac{\log^{2}(nU)}{\phi})$ and returns the set $E({\cal P})$ of edges involved in this embedding.
\end{itemize}
The running time of this algorithm is $O(m^{1+o(1)}U^{1+o(1)}/\phi)$.
\label{coro:EmbeddingOfExpander}
\end{corollary}

The remaining proof getting to \Cref{lemma:CutOrCertify} follow the parameter-adjusting technique in \cite{NS17}, which also appeared in \cite{ChuzhoyGLNPS20}. The algorithm keep prunning the graph $G$ using \Cref{lemma:MostBalanceCut} in several phases, and the parameters $\phi$ and $z$ passed to \Cref{lemma:MostBalanceCut} will change between phases. When the algorithm terminates, the vertex cut separated the pruned part and the remaining part is sparse since the graph is pruned by a sparse vertex cut each time. If the pruned part is large, then we obtain a balanced sparse vertex cut, otherwise the remaining part is large and it will be certified as an expander. The complete proof is shown in \Cref{proof:BalCutPrune}.

\section{The Low Degree Hierarchy}
\label{sect:LowDegreeHierarchy}

The \textit{low degree hierarchy}, which was first introduced in \cite{DuanP20}, is an important ingredient of our oracle. In this section, we will show an alternative construction algorithm of the low degree hierarchy, which exploits vertex expander decomposition. Compared to the construction in \cite{DuanP20}, our algorithm is more time-efficient, at a cost of giving a slightly low-quality hierarchy.

In short, the low degree hierarchy includes a laminar set ${\cal C}$ of components and a set ${\cal T}$ of low-degree Steiner trees. Every component $\gamma\in{\cal C}$ is corresponding to a tree $\tau\in{\cal T}$ which will span all (but may not only) \textit{terminals} of $\gamma$, where terminals of $\gamma$ are vertices in $V(\gamma)$ that are not in any of the child components. The main result of this section is an almost-linear time construction algorithm for the low degree hierarchy as shown in \Cref{thm:LowDegreeHierarchy}. In \Cref{sect:RepeatDecomp}, we will show how to use vertex expander decomposition repeatedly to construct a prototype of the low degree hierarchy. In \Cref{sect:HierarchyConstruction}, we will refine the prototype to formally construct the hierarchy and complete the proof of \Cref{thm:LowDegreeHierarchy}.

\begin{definition}[Low Degree Hierarchy]
Let $G$ be a connected undirected graph. A $(p,\Delta)$-low degree hierarchy with height $p$ and degree parameter $\Delta$ on $G$ is a pair $({\cal C},{\cal T})$ of sets, where ${\cal C}$ is a set of vertex-induced connected subgraphs called \textit{components}, and ${\cal T}$ is a set of Steiner trees with maximum vertex degree at most $\Delta$.

The set ${\cal C}$ of components is a laminar set. Specifically, ${\cal C}$ has the following properties.
\begin{itemize}
    \item[(1)] Components in ${\cal C}$ belong to $p$ levels and we denote by ${\cal C}_{i}$ the set of components at level $i$. In particular, at the top level $p$, ${\cal C}_{p}=\{G\}$ is a singleton set with the whole $G$ as the unique component. Furthermore, for each level $i\in[1,p]$, components in ${\cal C}_{i}$ are vertex-disjoint and there is no edge in $E(G)$ connecting two components in ${\cal C}_{i}$.
    \item[(2)] For each level $i\in [1,p-1]$ and each component $\gamma\in{\cal C}_{i}$, there is a unique component $\gamma'\in{\cal C}_{i+1}$ such that $V(\gamma)\subseteq V(\gamma')$, where we say that $\gamma'$ is the \textit{parent-component} of $\gamma$ and that $\gamma$ is a \textit{child-component} of $\gamma'$.
    \item[(3)] For each component $\gamma\in{\cal C}$, the \textit{terminals of $\gamma$}, denoted by $U(\gamma)$, are vertices in $\gamma$ but not in any of $\gamma$'s child-components. Note that $U(\gamma)$ can be empty. In particular, for each $\gamma\in{\cal C}_{1}$, $U(\gamma)=V(\gamma)$.
\end{itemize}

Generally, for each level $i\in[1,p]$, we define the \textit{terminals at level $i$} be terminals in all components in ${\cal C}_{i}$, denoted by $U_{i}=\bigcup_{\gamma\in{\cal C}_{i}}U(\gamma)$.

The set ${\cal T}$ of low-degree Steiner trees has the following properties
\begin{itemize}
    \item[(4)] ${\cal T}$ can also be partitioned into subsets ${\cal T}_{1},...,{\cal T}_{p}$, where ${\cal T}_{i}$ denote trees at level $i$ and trees in ${\cal T}_{i}$ are vertex-disjoint.
    \item[(5)] For each level $i\in[1,p]$ and tree $\tau\in{\cal T}_{i}$, the \textit{terminals of $\tau$} is defined by $U(\tau)=U_{i}\cap V(\tau)$.
    \item[(6)] For each level $i\in[1,p]$ and each component $\gamma\in{\cal C}_{i}$ with $U(\gamma)\neq\emptyset$, there is a tree $\tau\in{\cal T}_{i}$ such that $U(\gamma)\subseteq U(\tau)$, denoted by $\tau(\gamma)$. We emphasize that two different components $\gamma$ and $\gamma' \in {\cal C}_i$ may correspond to the same tree $\tau \in {\cal T}_i$. 
\end{itemize}
\label{def:LowDegreeHierarchy}
\end{definition}

\begin{proposition}
In \Cref{def:LowDegreeHierarchy}, terminal sets of components $\{U(\gamma)\mid \gamma\in{\cal C}\}$ forms a partition of $V(G)$. Similarly, regarding terminal sets of Steiner trees and levels, $\{U(\tau)\mid\tau\in{\cal T}\}$ and $\{U_{i}\mid 1\leq i\leq p\}$ are also partitions of $V(G)$.
\label{prop:TerminalPartition}
\end{proposition}

The main result of this section is \Cref{thm:LowDegreeHierarchy}. As a comparison, we state the result of \cite{DuanP20} in \Cref{thm:DPLowDegreeHierarchy}. Combining new algorithms in the following sections, we can get some tradeoffs on complexity measures of our final oracles by choosing the construction algorithm of the low degree hierarchy.

\begin{theorem}
Given an $n$-vertex $m$-edge graph $G$, there is a deterministic algorithm to construct the low degree hierarchy $({\cal C},{\cal T})$ with $p=O(\beta(n))$ levels and degree parameter $\Delta=n^{o(1)}$, where $\beta(n)$ is any slowly growing function e.g. $\beta(n)=\log^{*}n$. The running time of this algorithm is $m^{1+o(1)}$ and the space to store the hierarchy explicitly is $O(pm)=O(m\cdot\beta(n))$.
\label{thm:LowDegreeHierarchy}
\end{theorem}

\begin{theorem}[\cite{DuanP20}]
Given an $n$-vertex $m$-edge graph $G$, there is a deterministic algorithm to construct the low degree hierarchy $({\cal C},{\cal T})$ with $p=O(\log n)$ levels and degree parameter $\Delta=4$. The running time of this algorithm is $O(mn\log n)$ and the space to store the hierarchy explicitly is $O(m\log n)$.
\label{thm:DPLowDegreeHierarchy}
\end{theorem}

\subsection{Repeated Vertex Expander Decomposition}
\label{sect:RepeatDecomp}

In this subsection, we will design an algorithm which takes the graph $G$ as the input and output a \textit{repeated vertex expander decomposition} of $G$, which serves as a prototype of the hierarchy. In our algorithm, we will view \Cref{Coro:WeightedDecomposition} as a subroutine which computes a ``terminal version'' of unweighted vertex expander decomposition, in which the vertex expansion is with respect to a given terminal set $T$ (see \Cref{def:TerminalVertexExpansion}).

\begin{definition}[Vertex Expansion with respect to Terminals]
Let $G$ be an unweighted graph with a terminal set $T\subseteq V(G)$. For each vertex cut $(L,S,R)$ of $G$ with $|(L\cup S)\cap T|,|(R\cup S)\cap T|>0$, its \textit{vertex expansion with respect to $T$} is $h_{G,T}(L,S,R)=|S|/\min\{|(L\cup S)\cap T|,|(R\cup S)\cap T|\}$. The vertex expansion of $G$ with respect to $T$, denoted by $h(G,T)$, is the minimum $h_{G,T}(L,S,R)$ of any vertex cut $(L,S,R)$ with $|(L\cup S)\cap T|,|(R\cup S)\cap T|>0$.
\label{def:TerminalVertexExpansion}
\end{definition}

Note that \Cref{def:TerminalVertexExpansion} is equivalent to the notion of vertex expansion in \Cref{def:VertexExpansion} if we extend $G$ to a weighted graph with weight function $w$ and demand function $\d_{T}$ as follows. For each vertex $v\in V(G)$, it has weight $w(v)=1$. For a given $T$, each vertex in $T$ has demand $\d_{T}=1$ and each vertex in $V(G)\setminus T$ has demand $\d_{T}=0$.

Let $\beta(n)\leq \log n$ be any slowly growing function. The algorithm is iterative. We let $T_{i}$ denote the terminals of the $i$-th iteration with $T_{1}=V(G)$ initially. In the $i$-th iteration, we apply \Cref{Coro:WeightedDecomposition} (the deterministic version) on graph $G$ with weight function $w$, demand function $\d_{T_{i}}$ and parameter $\epsilon=n^{1/\beta(n)}=n^{o(1)}$, which will return a decomposition ${\cal G}_{i}$. The terminal set $T_{i+1}$ for the next iteration is the vertices dropped in this decomposition, i.e. $T_{i+1}=V(G)\setminus\left(\bigcup_{V_{j}\in{\cal G}_{i}}V_{j}\right)$. We proceed to the next iteration unless $T_{i+1}=\emptyset$. Let $p$ denote the number of iterations. The outputs are $T_{1},...,T_{p}$ and ${\cal G}_{1},...,{\cal G}_{p}$ with properties summarized in \Cref{lemma:RepeatDecompProperty}. 

\begin{lemma}
Given an $n$-vertex $m$-edge graph $G$, there is a deterministic algorithm that computes $T_{1},...,T_{p}$ and ${\cal G}_{1},...,{\cal G}_{p}$ with $p=O(\beta(n))$ and the following properties.
\begin{itemize}
\item[(1)] For each $1\leq i\leq p-1$, $T_{i+1}=V(G)\setminus\left(\bigcup_{V_{j}\in{\cal G}_{i}}V_{j}\right)$. In particular, $T_{1}=V(G)$.
\item[(2)] For each $1\leq i\leq p$ and any two distinct $V_{j},V_{j'}\in {\cal G}_{i}$, $V_{j}$ and $V_{j'}$ are disjoint and there is no edge in $E(G)$ connecting a vertex in $V_{j}$ and another vertex in $V_{j'}$.
\item[(3)]
For each $1\leq i\leq p$ and $V_{j}\in{\cal G}_{i}$, the vertex expansion of $G[V_{j}]$ with respect to $V_{j}\cap T_{i}$ is at least $1/n^{o(1)}$.
\end{itemize}
The running time of this algorithm is $m^{1+o(1)}$.

\label{lemma:RepeatDecompProperty}
\end{lemma}
\begin{proof}
First we claim that $p=O(\beta(n))$ because in every iteration $i$, we have $|T_{i+1}|=w(V(G)\setminus(\bigcup_{V\in{\cal G}_{i}}V))\leq \epsilon\cdot\d_{T_{i}}(V(G))=|T_{i}|/n^{1/\beta(n)}$ by \Cref{Coro:WeightedDecomposition}. The property (1) is directly from the algorithm, and properties (2) and (3) are from \Cref{Coro:WeightedDecomposition}. Furthermore, each iteration takes $m^{1+o(1)}$ time to apply \Cref{Coro:WeightedDecomposition} once, so the total running time is $m^{1+o(1)}$.

\end{proof}

\subsection{Construction}
\label{sect:HierarchyConstruction}

We now formally construct the low degree hierarchy. The key relation between the repeated vertex expander decomposition and the low degree hierarchy is shown in \Cref{lemma:ExpanderToLowDegree}. Roughly speaking, a graph with high vertex expansion with respect to a terminal set contains a low degree Steiner tree spanning all terminals.

\begin{lemma}
Let $G$ be an $n$-vertex $m$-edge graph with a terminal set $T\subseteq V(G)$. If $h(G,T)$ is at least $\phi$, there is a deterministic algorithm that computes a Steiner tree $\tau$ in $G$ spanning $T$ with maximal vertex degree at most $O(\frac{\log^{2}n}{\phi})$. The running time of this algorithm is $m^{1+o(1)}/\phi$.
\label{lemma:ExpanderToLowDegree}
\end{lemma}
\begin{proof}
We apply \Cref{coro:EmbeddingOfExpander} to the graph $G=(V(G),E(G),w_{1},\d_{T})$ with parameter $\phi/2$. Because the weighted vertex expansion $h(G)$ is exactly $h(G,T)\geq \phi$, \Cref{coro:EmbeddingOfExpander} must certify that we can embed an expander $H$ with the same demand $\d_{T}$ into $G$ with vertex-congestion $\eta=O(\log^{2}n/\phi)$ and return $E({\cal P})$ of this embedding ${\cal P}$. Let $\tau'$ be the graph with $V(\tau')=V(G)$ and $E(\tau')=E({\cal P})$. Because $H$ is an expander where each $v\in T$ has demands $\d_{T}=1$, by the definition of embedding, vertices in $T$ must be connected in $\tau'$ (although $\tau'$ itself may not be a connected graph). Moreover, the maximum vertex degree in $\tau'$ is at most $\eta=O(\log^{2}n/\phi)$. Therefore, in the spanning forest of $\tau'$, there must be a tree $\tau$ spanning all vertices in $T$. The running time is $m^{1+o(1)}/\phi$ taken by \Cref{coro:EmbeddingOfExpander} and the other parts take nearly linear time.
\end{proof}

Given the outputs $T_{1},...,T_{p}$ and ${\cal G}_{1},...,{\cal G}_{p}$ of \Cref{lemma:RepeatDecompProperty}, the construction algorithm is as follows.
\begin{itemize}
    \item[(A)] The hierarchy has $p$ levels. For each level $i$, the terminals at this level are $U_{i}=T_{i}\setminus\bigcup_{i'=i+1}^{p}T_{i'}$.
    \item[(B)] For each level $i\in[1,p]$, the set ${\cal T}_{i}$ of low-degree Steiner trees is constructed as follows. For each $V_{j}\in{\cal G}_{i}$, by applying \Cref{lemma:ExpanderToLowDegree} on $G[V_{j}]$ with terminal set $V_{j}\cap T_{i}$, we can obtain a Steiner tree $\tau_{i,j}$ in $G[V_{j}]$ spanning $V_{j}\cap T_{i}$ with maximal vertex degree $n^{o(1)}$ since $h(G[V_{j}],V_{j}\cap T_{i})\geq 1/n^{o(1)}$. The set ${\cal T}_{i}$ will collect all such $\tau_{j}$, i.e. ${\cal T}_{i}=\{\tau_{i,j}|V_{j}\in {\cal G}_{i}\}$, where each $\tau_{i,j}$ has terminals $U(\tau_{i,j})=V_{j}\cap U_{i}$. The whole set ${\cal T}=\bigcup_{i=1}^{p}{\cal T}_{i}$.
    \item[(C)] For each level $i\in[1,p]$, the set ${\cal C}_{i}$ of components collect all connected components of the graph $G\setminus \bigcup_{i'=i+1}^{p}U_{i'}$. For each $\gamma \in{\cal C}_{i}$, its terminals are $U(\gamma)=U_{i}\cap V(\gamma)$. The whole set ${\cal C}=\bigcup_{i=1}^{p}{\cal C}_{i}$.
\end{itemize}

\begin{lemma}
The $(\cal C,\cal T)$ constructed above is a low-degree hierarchy with $p=O(\beta(n))$ levels and degree parameter $r=n^{o(1)}$.
\label{lemma:LowDegreeHierarchyCorrectness}
\end{lemma}
\begin{proof}
Properties (1) and (2) in \Cref{def:LowDegreeHierarchy} are directly from the step (C). Property (3) follows step (C) and that $\{U_{1},\cdots,U_{p}\}$ constructed above is a partition of $V(G)$, where the latter is from $T_{1}=V(G)$ and step (A). Property (4) holds because of step (B) and that subgraphs in ${\cal G}_{i}$ are vertex-disjoint for each level $i$. Regarding property (5), notice that each $\tau_{i,j}\in{\cal T}_{i}$ corresponding to some $V_{j}\in{\cal G}_{i}$ will span $V_{j}\cap T_{i}$, so $U(\tau_{i,j})=V_{j}\cap U_{i}\subseteq V_{j}\cap T_{i}\subseteq V(\tau_{i,j})$, which implies $U(\tau_{i,j})=V(\tau_{i,j})\cap U_{i}$. It remains to show property (6). 

Consider some level $i$. Each $\gamma\in{\cal C}_{i}$ must be totally contained by some $V_{j}\in{\cal G}_{i}$, because $G\setminus(U_{i+1}\cup\cdots\cup U_{p})\subseteq V(G)\setminus T_{i+1}=\bigcup_{V_{j}\in{\cal G}_{i}}V_{j}$ by (1) in \Cref{lemma:RepeatDecompProperty} and there is no edge connecting two distinct $V_{j},V_{j'}\in{\cal G}_{i}$ by (2) in \Cref{lemma:RepeatDecompProperty}. Therefore, $U(\gamma)=V(\gamma)\cap U_{i}\subseteq V_{j}\cap U_{i}=U(\tau_{i,j})$ as desired.
\end{proof}

\begin{proof}[Proof of \Cref{thm:LowDegreeHierarchy}]
The correctness is from \Cref{lemma:LowDegreeHierarchyCorrectness}, and we analyse the running time below. We first invoke the algorithm in \Cref{lemma:RepeatDecompProperty}, which takes $m^{1+o(1)}$ time. Consider the algorithm in this subsection. Observe that steps (A) and (C) only take $\tilde{O}(m)$ time. About step (B), applying \Cref{lemma:ExpanderToLowDegree} on $G[V_{j}]$ for all $V_{j}\in{\cal G}_{i}$ takes $m^{1+o(1)}$ time for each level $i$. Summing over $p=O(\beta(n))\leq O(\log n)$ levels, step (B) takes totally $m^{1+o(1)}$ time. Thus the total running time is $m^{1+o(1)}$.
\end{proof}

\section{The Affiliated Structures}
\label{sect:InterCom}
In this section, we will construct some affiliated structures to augment the low degree hierarchy.
Before introducing the affiliated structures, we first discuss in \Cref{sect:AbstractGraph} how we react against a set $D$ of failed vertices by exploiting a $(p,\Delta)$-low degree hierarchy. Basically, given a failure set $D$, we will reduce the original graph to a semi-bipartite graph called the \textit{abstract graph}. The overview of affiliated structures is shown in \Cref{sect:AffStrOverview}.

\subsection{The Abstract Graph $H$ for a Failure Set $D$}
\label{sect:AbstractGraph}

Let $D$ be a set of failed vertices with $|D|=d\leq d_{\star}$. After failures, some \textit{objects} in the low degree hierarchy still provide valid information on connectivity. 

We first focus on components. We call some $\gamma\in{\cal C}$ an \textit{unaffected component} if $V(\gamma)\cap D=\emptyset$, otherwise we call $\gamma$ an \textit{affected component}. because an unaffected component $\gamma$ does not %
involve any failed vertex, it certifies that vertices in $V(\gamma)$ are still connected in graph $G\setminus D$. By the laminar property of ${\cal C}$, it is enough to consider \textit{maximal unaffected components}, which is defined as an unaffected component with an affected component as the parent-component. We let all maximal unaffected components be the \textit{first-type objects}. When we are solving the connectivity of $G\setminus D$, we can view vertices in each first-type object as a whole.

The remaining vertices which cannot be grouped by first-type objects are exactly the non-failed terminals of affected components, because the terminal sets of components partition $V(G)$ and for each first-type object $\gamma$, $V(\gamma)$ is exactly the union of terminal sets of its descendant components. Recall that for each affected component $\gamma\in{\cal C}$, its terminals $U(\gamma)$ are contained in a single tree $\tau(\gamma)\in{\cal T}$ if $U(\gamma)\neq\emptyset$ (see (6) in \Cref{def:LowDegreeHierarchy} for details) and we call this tree an \textit{affected tree}.
Although an affected tree may be split by failed vertices into several subtrees (called \textit{affected subtrees}) in $G\setminus D$, the total number of affected subtrees will not be large by the low-degree property of trees in ${\cal T}$. We let all affected subtrees be the \textit{second-type objects}, and for each affected subtree $\hat{\tau}$ of an affected tree $\tau$, we define the terminal set of $\hat{\tau}$ as $U(\hat{\tau})=U(\tau)\cap V(\hat{\tau})$.
Similarly, for each affected subtree, we can group as a whole its terminals not in any of first-type objects\footnote{Some vertex in $U(\hat{\tau})$ can be inside some first-type object because $U(\hat{\tau})$ may contain terminals of several components (see (6) of \Cref{def:LowDegreeHierarchy}). Indeed we can group all vertices inside an affected subtree as a whole since this tree is connected in $G\setminus D$. Here we only group a part of vertices because for simplicity, we want all groups form a partition of $V(G)\setminus D$.}.

\begin{definition}
We define the \textit{abstract graph} $H$ for $D$ as follows. 
\begin{itemize}
    \item Let $\hat{\cal C}$ denote the set of first-type objects and $\hat{\cal T}$ denote the set of second-type objects. Then $V(H)=\hat{\cal C}\cup\hat{\cal T}$, which means vertices in $V(H)$ one-one correspond to objects. Without ambiguity, we use the same notation $\gamma\in\hat{\cal C}$ (resp. $\hat{\tau}\in\hat{\cal T}$) to denote both the first-type object (resp. the second-type object) and its corresponding vertex in $V(H)$.
    \item Each vertex in $V(H)$ is an abstraction of a subset of $V(G)\setminus D$, where all subsets form a partition of $V(G)\setminus D$. Specifically, for each vertex $\gamma\in\hat{\cal C}$, it represents $g^{-1}(\gamma)=V(\gamma)$, while for each vertex $\hat{\tau}\in\hat{\cal T}$, it represents $g^{-1}(\hat{\tau})=U(\hat{\tau})\setminus \bigcup_{\gamma\in\hat{\cal C}}V(\gamma)$. We assume that each $\hat{\tau}\in\hat{\cal T}$ has nonempty $g^{-1}(\hat{\tau})$ by ignoring those $\hat{\tau}$ with empty $g^{-1}(\hat{\tau})$.
    
    We have $\{g^{-1}\}_{v_{H}\in V(H)}$ is a partition of $V(G)$ by the discussion above, so we can define a mapping $g:V(G)\to V(H)$ by defining $g(v)$ be the unique $v_{H}$ such that $v\in g^{-1}(v_{H})$ for each $v\in V(G)$.
    \item The edge set $E(H)=\{(g(u),g(v))\mid (u,v)\in E(G)\}$ is a multiset of edges one-one correspond to edges in $E(G\setminus D)$, where each edge $(u,v)\in E(G\setminus D)$ is mapped to an edge $(g(u),g(v))\in E(H)$. 
\end{itemize} 

\label{def:AbstractGraph}
\end{definition}

\begin{lemma}
Given a failure set $D$ with $|D|=d$, the abstract graph $H$ has the following properties.
\begin{itemize}
\item[(1)] The number of second-type objects is $O(pd\Delta)$. \label{lemma:AffectedSubtreesNumber}
\item[(2)] There is no edge connecting two first-type objects.
\label{lemma:PseudoBipartiteH}
\item[(3)] For any two second-type objects $\hat{\tau}_{1}$ and $\hat{\tau}_{2}$, $\hat{\tau}_{1}$ and $\hat{\tau}_{2}$ is connected in $H$ if some $v_{1}\in V(\hat{\tau}_{1})$ and $v_{2}\in V(\hat{\tau}_{2})$ is connected in $G\setminus D$, and any $v_{1}\in V(\hat{\tau}_{1})$ and $v_{2}\in V(\hat{\tau}_{2})$ is connected in $G\setminus D$ if $\hat{\tau}_{1}$ and $\hat{\tau}_{2}$ is connected in $H$.
\label{lemma:ConnEqOrigianlH}
\end{itemize}

\end{lemma}
\begin{proof}
Because components in the same level are disjoint by (1) in \Cref{def:LowDegreeHierarchy}, there are at most $O(d)$ affected trees in each level, so the total number of affected trees is $O(pd)$. Because affected trees at the same level are disjoint by (4) in \Cref{def:LowDegreeHierarchy}, there are totally $O(pd)$ failed vertices on affected trees summing over $p$ levels. Furthermore, each affected tree has maximum vertex degree $\Delta$, so the number of affected subtrees (namely second-type objects) is $O(pd\Delta)$.

Let $\gamma_{1},\gamma_{2}$ be two distinct first-type objects and assume $\gamma_{1}$ is at a higher level (or the same level) than $\gamma_{2}$. By (2) in \Cref{def:LowDegreeHierarchy}, $\gamma_{2}$ has a unique ancestor $\gamma'_{2}$ at the same level with $\gamma_{1}$, such that $V(\gamma_{2})\subseteq V(\gamma'_{2})$ and there is no edge connecting $\gamma_{1}$ and $\gamma'_{2}$. Thus property (2) follows.

Note that each second-type object $\hat{\tau}$ has nonempty $g^{-1}(\hat{\tau})$. Then property (3) follows the facts that each affected subtree is a connected subgraph of $G\setminus D$ and the transformation from $G\setminus D$ to $H$ preserves connectivity.
\end{proof}

The task of our update algorithm is to compute the connectivity of the abstract graph $H$. Precisely, since the number of first-type objects may be large but the number of second-type objects is quite small by \Cref{lemma:AffectedSubtreesNumber}, we will see that, our update algorithm only reestablish connectivity among second-type objects in $H$. 
The data structures constructed in this section aim for supporting fast update algorithm.

\subsection{Overview}
\label{sect:AffStrOverview}

There are two parts of the affiliated structures. The first part is basically a 2D range counting structure in \Cref{sect:2DRangeStructure}. The intuition behind is that we can view affected subtrees, i.e. the second-type objects, as intervals on some Euler tour order by \Cref{lemma:TreeToInterval}, so querying whether two affected subtrees are connected directly by some edges is equivalent to some 2D range counting queries. 

An issue is that two second-type objects in $H$ can be connected by, instead of an edge, a path through a first type object, but such paths are not captured by a 2D range counting structure only storing original edges. Following an idea from \cite{DuanP20}, we will add \textit{artificial edges} into the counting structure to capture the connectivity provided by them. Our construction of artificial edges in \Cref{sect:ArtificialEdges} is deterministic and compact, which leads to faster deterministic update algorithm and lower space requirement.

Since the objects of $H$ are ``dynamic'' (depending on $D$), we should be able to efficiently ``subtract'' from the 2D range counting structure artificial edges whose underlying paths go through affected components, because these edges are not guaranteed to provide valid connectivity. Therefore, we also construct some navigation structures around artificial edges in \Cref{sect:ArtificialEdges}. When some components becomes affected, these structures can support efficient ``subtraction'' of invalid artifical edges.

\subsection{The Euler Tour Order and the 2D Range Counting Structure}
\label{sect:2DRangeStructure}

Our goal of this subsection is to construct a 2D range counting structure (see \Cref{lemma:2DRangeCounting}) to support connectivity queries of affected subtrees under a multiset of edges $E=E(G)+\hat{E}$, where $\hat{E}$ is the multiset of artificial edges which will be constructed in \Cref{sect:ArtificialEdges}. We emphasize that although we conceptually define $E$ and $\hat{E}$ be multisets of undirected edges, actually our algorithm will construct and view them as multisets of ordered pairs, where each pair relates to an undirected edge. The original edge set $E(G)$ can also be viewed as a set of ordered pair by giving an arbitrary direction to each undirected edge.

We will exploit the Euler tour order as in \cite{DuanP20} to reduce queries on connectivity of affected subtrees to 2D range counting queries. Our algorithm will process each $\tau\in{\cal T}$ individually. For an unrooted tree $\tau$, by picking a root arbitrarily, we define the Euler tour order $\ET(\tau)$ as an ordered list of the vertices ordered by the time each vertex first appears during an Euler tour of $\tau$ (treating each undirected edge as two directed edges). The Euler tour order can be interpreted as a linearization of the tree in the sense that after deleting some vertices on $\tau$, the remaining subtrees are corresponding to intervals in $\ET(\tau)$ as shown in \Cref{lemma:TreeToInterval}. Proving \Cref{lemma:TreeToInterval} is an easy exercise and we leave the complete proof in \Cref{proof:TreeToInterval}.

\begin{lemma}
Let $\tau$ be an undirected tree with maximum vertex degree $\Delta$. A removal of $d$ failed vertices from $\tau$ will split $\tau$ into at most $O(\Delta d)$ subtrees $\hat{\tau}_{1},\hat{\tau}_{2},\cdots,\hat{\tau}_{O(\Delta d)}$, and there exists a set $I_{\tau}$ of at most $O(\Delta d)$ disjoint intervals on $\ET(\tau)$, such that each interval is owned by a unique subtree and for each subtree $\tau_{i}$, $V(\tau_{i})$ is equal to the union of intervals it owns. 

Furthermore, by preprocessing $\tau$ in $O(|V(\tau)|)$ time, given any set of $d$ failed vertices, after $O(\Delta d\log (\Delta d))$ update time, these intervals (with the owner of each of them) can be computed and any query on which subtree contains a given vertex can be answered in $O(\log d)$ query time. 
\label{lemma:TreeToInterval}
\end{lemma}

As mentioned in \Cref{sect:AbstractGraph}, each affected subtree $\hat{\tau}\in V(H)$ serves as an abstraction of vertices in $g^{-1}(\hat{\tau})$ which is a subsets of its terminal set $U(\hat{\tau})$. Therefore, we define the global Euler tour order by focusing on terminals of trees in ${\cal T}$, and take it as indices of the 2D range counting structure in \Cref{lemma:2DRangeCounting}. We note that the global Euler tour order is a permutation of $V(G)$ since $\{U(\tau)\mid \tau\in{\cal T}\}$ is a partition of $V(G)$ as shown in \Cref{prop:TerminalPartition}.

\begin{definition}[Global Euler Tour Order]
For each level $i$ and each $\tau\in{\cal T}_{i}$, we treat its terminal set $U(\tau)$ as an ordered list ordered by $\ET(\tau)$. The \textit{global Euler tour order} $\ET_{\glo}$ is the concatenation of $U(\tau)$ for all $\tau\in{\cal T}$ in an arbitrary order. For each vertex $v\in V(G)$, we let $\Id(v)$ be an integer in $[1,n]$ which denotes the rank of $v$ in $\ET_{\glo}$.
\label{def:GlobalET}
\end{definition}

\begin{lemma}[2D Range Counting Structure]
Let $E$ be a multiset of ordered pairs in $V(G)\times V(G)$. There is a data structure which given any disjoint intervals $I_{1},I_{2}$ on $\ET_{\glo}$, can answer in $O(\log n)$ time the number of pairs in $E$ with first entries in $I_{1}$ and second entries in $I_{2}$, namely $|E\cap(I_{1}\times I_{2})|$. The data structure takes space is $O(\Vert E\Vert_{0}\log n)$ and the construction time is $O(|E|\log n)$, where $\Vert E \Vert_{0}$ denote the number of distinct elements in $E$.
\label{lemma:2DRangeCounting}
\end{lemma}
\begin{proof}
We can generate from $E$ a weighted 2D point set $Q$ with $|Q|=\Vert E\Vert_{0}$ by putting for each $(u,v)\in E$, a point $(\Id(u),\Id(v))$ to $Q$ weighted by the number of $(u,v)$ in $E$. This step takes $O(|E|\log n)$ time via sorting. Then the problem is reduced to a standard weighted 2D range counting problem. It can be solved by textbook algorithms e.g. the range tree \cite{Ben80} or the persistent segment tree \cite{DSST89,BW80}. The space is $O(|Q|\log |Q|)=O(\Vert E\Vert_{0}\log n)$ and the query time is $O(\log n)$.
\end{proof}

\begin{remark}
There are linear-space data structures for weighted 2D range counting problem. For example, for a set $Q$ of weighted 2D points, the semigroup range searching data structure\footnote{It even works for a stronger setting that the weights are in a semigroup.} in \cite{C88} takes $O(|Q|\log|Q|)$ preprocessing time and $O(|Q|)$ space and achieves $O(\log^{2+\epsilon}|Q|)$ query time for any constant $\epsilon>0$. Using such data structures for \Cref{lemma:2DRangeCounting} will improve the space of our oracle by an $O(\log n)$ factor, but will increase the update time by a factor $O(\log^{1+\epsilon}n)$.
\label{remark:LinearSpace2DCounting}
\end{remark}

\subsection{Artificial Edges}
\label{sect:ArtificialEdges}

In this subsection, we will construct the set $\hat{E}$ of artificial edges. Actually, our algorithm will construct a set $\hat{E}_{\gamma}$ of artificial edges for each $\gamma\in{\cal C}$, in which each artificial edge will connect two vertices adjacent to $\gamma$, i.e. the set $\hat{E}_{\gamma}$ of ordered pairs is a subset of $A_{\gamma}\times A_{\gamma}$, where $A_{\gamma}$ is $\gamma$'s adjacency list (see \Cref{def:AdjacencyList}). The whole set $\hat{E}$ is a multiset such that $\hat{E}=\sum_{\gamma\in{\cal C}}\Ehat_{\gamma}$.%

\begin{definition}[Adjacency List]
For each $\gamma\in{\cal C}$, the \textit{adjacency list} of $\gamma$ is an ordered list
\[
A_{\gamma}=\{v \mid v\notin V(\gamma)\text{ s.t. }\exists u\in V(\gamma), (u,v)\in E(G)\},
\]
whose order is consistent with $\ET_{\glo}$.%
\label{def:AdjacencyList}
\end{definition}

\begin{proposition}
The total size of adjacency lists of all components is $\sum_{\gamma\in{\cal C}}|A_{\gamma}|=O(pm)$.
\label{lemma:TotalSizeOfA}
\end{proposition}
\begin{proof}
Consider a level $i$ in the low degree hierarchy. For each $\gamma\in {\cal C}_{i}$ and $v\in A_{\gamma}$, there is an edge connecting $\gamma$ and $v$. Also, each edge will correspond to at most one such pair of $\gamma$ and $v$ since property (1) in \Cref{def:LowDegreeHierarchy}. Thus $\sum_{\gamma\in{\cal C}_{i}}|A_{\gamma}|\leq m$ and $\sum_{\gamma\in{\cal C}}|A_{\gamma}|=O(pm)$.
\end{proof}

Our update algorithm will only solve the connectivity of second-type objects in $H$ for efficiency. Recall that there is no edge in $H$ connecting two different first-type objects (see \Cref{lemma:PseudoBipartiteH}). 
Therefore, from the perspective of second-type objects, each first-type object $\gamma$ with its adjacent edges can be viewed as a hyperedge, which will also contribute to the connectivity of second-type objects in addition to original edges. The artificial edges $\hat{E}_{\gamma}$ of $\gamma$ are served as a simulation of this hyperedge. Therefore, we hope each $\hat{E}_{\gamma}$ has the following \textit{fault-tolerant} property: for any failure set $D$ with $|D|\leq d_{\star}$, there will be either $A_{\gamma}\subseteq D$ or vertices in $A_{\gamma}\setminus D$ are connected by $\Ehat_{\gamma}$.
On the other hand, we hope the size of $\hat{E}_{\gamma}$ is small to save the space taken by the 2D range counting structure. A natural idea is to take an arbitrary subset $B_{\gamma}\subseteq A_{\gamma}$ with size $|B_{\gamma}|=d_{\star}+1$ ($B_{\gamma}=A_{\gamma}$ if $|A_{\gamma}|\leq d_{\star}$), and let $\hat{E}_{\gamma}=A_{\gamma}\times B_{\gamma}$ with $|\hat{E}_{\gamma}|= O(d_{\star}|A_{\gamma}|)$. This simple construction can indeed lead to an oracle that takes space $\tilde{O}(d_{\star}m)$. Actually, this is the optimal construction of $\hat{E}_{\gamma}$ under this setting by the fact that any $d_{\star}$-vertex-connected graph on $A_{\gamma}$ has at least $\Omega(d_{\star}|A_{\gamma}|)$ edges.

In our construction, we get around this barrier by slightly relax the requirement on the fault-tolerant property, inspired by the Monte Carlo construction algorithm in \cite{DuanP20}. For each $\gamma\in{\cal C}$, instead of ensuring the fault-tolerant property on all $\hat{E}_{\gamma}$ in the preprocessing phase, we are allowed to add some \textit{$D$-repairing edges} $\hat{F}_{\gamma}$ temporarily in the update phase, such that $\hat{E}_{\gamma}\cup\hat{F}_{\gamma}$ is fault-tolerant. Under this setting, we are able to deterministically construct artificial edges $\hat{E}_{\gamma}$ for each $\gamma\in{\cal C}$ such that the total number of distinct artificial edges is bounded by $\tilde{O}(m)$, which is exactly nearly linear to the space of the 2D range counting structure. Moreover, we will show that the construction of $D$-repairing edges only takes $\tilde{O}(d^{2})$ time and the total number of $D$-repairing edges is bounded by $\tilde{O}(d^{2})$, which can be afforded by the update algorithm. Lastly, our construction of $\hat{E}_{\gamma}$ is based on a simple recursive division on the list $A_{\gamma}$ called \textit{segmentation hierarchy}, so the artificial edges are well-structured and can be efficiently manipulated by the update algorithm.

\begin{definition}[Segmentation Hierarchy]
For each component $\gamma\in{\cal C}$, we define a \textit{segmentation hierarchy} ${\cal S}_{\gamma}$ on its adjacency list $A_{\gamma}$. ${\cal S}_{\gamma}$ is an ordered list of \textit{segments} on $A_{\gamma}$, where a segment on $A_{\gamma}$ is a consecutive sublist of $A_{\gamma}$. 

There are $r=\lceil\log |A_{\gamma}|\rceil +1$ levels in the hierarchy and we denote by ${\cal S}_{\gamma,j}$ the ordered list of segments at level $j$ for each $1\leq j\leq r$. The whole list ${\cal S}_{\gamma}$ is the concatenation of ${\cal S}_{\gamma,1},...,{\cal S}_{\gamma,r}$. The list ${\cal S}_{\gamma,j}$ is obtained by partitioning $A_{\gamma}$ into segments using step length $2^{j-1}$. Concretely, Let $A_{\gamma}[k_{1},k_{2}]$ denote the sublist of $A_{\gamma}$ from rank $k_{1}$ to $k_{2}$.
\begin{itemize}
    \item For $1\leq i\leq \lfloor|A_{\gamma}|/2^{j-1}\rfloor$, the $i$-th segment $S_{j,i}$ of ${\cal S}_{\gamma,j}$ is $A_{\gamma}[2^{j-1}(i-1)+1,2^{j-1}i]$. 
    \item If $|A_{\gamma}|$ is not divisible by $2^{j-1}$, we let the tail $A_{\gamma}[2^{j-1}(i-1)+1,|A_{\gamma}|]$ be the last segment $S_{j,i}$ where $i=\lfloor|A_{\gamma}|/2^{j-1}\rfloor+1$.
\end{itemize}
Observe that the total number of segments is $|{\cal S}_{\gamma}|=O(|A_{\gamma}|)$, and the total length of segments is $O(|A_{\gamma}|\log|A_{\gamma}|)$.

For each segment $S_{j,i}$ $(j\geq 2)$ in the hierarchy, we let $S_{j-1,2i-1}$ and $S_{j-1,2i}$ be its \textit{child-segments} (the latter may not exist), and we say $S_{j,i}$ is the \textit{parent-segment} of $S_{j-1,2i-1}$ and $S_{j-1,2i}$. Observe that $S_{j,i}$ is exactly the union of its child-segments. In particular, the topmost level $S_{\gamma,r}$ only contains $A_{\gamma}$ as its unique segment, called the \textit{root segment} of this hierarchy.
\label{def:Segmentation}
\end{definition}

\begin{remark}
For simplicity, the segmentation hierarchy defined in \Cref{def:Segmentation} is a \textit{two-branched} recursive division with $O(\log n)$ levels. It is natural to define \textit{multiple-branched} segmentation hierarchy to obtain a tradeoff between the number of levels and the number of branches. The number of levels will affect the space of the oracle, while the number of branches will affect the number of $D$-repairing edges and the update time. Therefore, when we allow a update time of $d^{2}n^{o(1)}$, we can use a $n^{o(1)}$-branched segmentation hierarchy with $\beta(n)$ levels, where $\beta(n)$ is any slowly growing function e.g. $\beta(n)=\log^{*}n$, which can substitute a $O(\log n)$ factor in the space complexity by $O(\beta(n))$.
\label{remark:MultipleBranchedSeg}
\end{remark}

\

\noindent\textbf{Construction of Artificial Edges.} The construction algorithm of all $\Ehat_{\gamma}$ will maintain a 2-dimensional \textit{labeling table} on $V(G)\times V(G)$, in which all entries are unlabeled initially. %
In fact, the labeling table are maintained implicitly by storing all labeled entries in a binary search tree in order to save the space. We will process components one by one in an arbitrary order. Let $\gamma$ be the current component.

\begin{itemize}
    \item[(1)] Let $v_{1},...,v_{|A_{\gamma}|}$ be vertices in $A_{\gamma}$ in order. We let $b_{\gamma}\leq \min\{|A_{\gamma}|,d_{\star}+1\}$ be the smallest integer such that there are at least $|{\cal S}_{\gamma}|$ unlabeled entries in range $A_{\gamma}\times \{v_{1},...,v_{b_{\gamma}}\}$ of the labeling table. If no such $b_{\gamma}$ exists, we let $b_{\gamma}=\min\{|A_{\gamma}|,d_{\star}+1\}$. We further let $B_{\gamma}$, called the \textit{connector} %
    of $\gamma$, be an ordered list with elements $v_{1},...,v_{b_{\gamma}}$. If $b_{\gamma}=\min\{|A_{\gamma}|,d_{\star}+1\}$, we finish processing $\gamma$ with $\hat{E}_{\gamma}=A_{\gamma} \times B_{\gamma}$.
    
    \item[(2)] We define $\Tag_{\gamma}$ be an ordered list of pairs $(u,v)$ (called \textit{tags}) such that $(u,v)\in A_{\gamma}\times B_{\gamma}$ and the entry $(u,v)$ is unlabeled. $\Tag_{\gamma}$ is in an order which ensures that for each $(u,v),(u',v')\in \Tag_{\gamma}$, $(u,v)$ is ordered before $(u',v')$ if $\Id(u)<\Id(u')$. For each $(u,v)\in\Tag_{\gamma}$ at rank $k$, we label the entry $(u,v)$ of the labeling table with a pair $(\gamma,k)$. 
    
    \begin{proposition}
    For each $\gamma$ that arrives at step (2), $|\Tag_{\gamma}|=\Theta(|A_{\gamma}|)$.
    \label{prop:TagSize}
    \end{proposition}
    \begin{proof}
    First, by the choice of $b_{\gamma}$, the number of unlabeled pairs in $A_{\gamma}\times (B_{\gamma}\setminus v_{b_{\gamma}})$ is less than $|{\cal S}_{\gamma}|=O(|A_{\gamma}|)$, so the number of unlabeled pairs in $A_{\gamma}\times B_{\gamma}$, i.e. $|\Tag_{\gamma}|$ is at most $|{\cal S}_{\gamma}|+|A_{\gamma}|=O(|A_{\gamma}|)$. Combining the number of unlabeled pairs in $A_{\gamma}\times B_{\gamma}$ is at least $|{\cal S}_{\gamma}|=\Omega(|A_{\gamma}|)$, we conclude $|\Tag_{\gamma}|=\Theta(|A_{\gamma}|)$.
    \end{proof} %
    
    \item[(3)] For each segment $S\in{\cal S}_{\gamma}$ with rank $k$, we assign it the $k$-ranked $(u,v)$ in $\Tag_{\gamma}$ and set $f(S)=u$ (we also view $f(S)$ as a singleton set $\{u\}$). At last, we set $\hat{E}_{\gamma}=A_{\gamma}\times B_{\gamma}+\sum_{S\in{\cal S}_{\gamma}}S\times f(S)$. For each $S\in{\cal S}_{\gamma}$, we call $f(S)$ its \textit{witness} and we call $f$ the witness function.
\end{itemize}

For convenience, if a component $\gamma$ has $\hat{E}_{\gamma}=A_{\gamma}\times B_{\gamma}$, we augment $\hat{E}_{\gamma}$ to $A_{\gamma}\times B_{\gamma}+\sum_{S\in{\cal S}_{\gamma}} S\times f(S)$, where $f(S)=u$ for all $S\in{\cal S_{\gamma}}$ where $u\in A_{\gamma}$ is an arbitrary fixed vertex. Therefore, we can write $\hat{E}$ as
\[
\hat{E}=\sum_{\gamma\in{\cal C}}\hat{E}_{\gamma}=\sum_{\gamma\in{\cal C}} \left(A_{\gamma}\times B_{\gamma} + \sum_{S\in{\cal S}_{\gamma}}S\times f(S)\right).
\]
The structural property of $\hat{E}_{\gamma}$ shown in \Cref{lemma:MonotonicallyIncreasing} is useful in the update algorithm. %

\begin{proposition}
For each $\gamma\in{\cal C}$, $\Id(f(S))$ is monotonically increasing with respect to $S\in {\cal S}_{\gamma}$, namely for each $S,S'\in{\cal S}_{\gamma}$ such that $S$ is ordered before $S'$, $\Id(f(S))\leq \Id(f(S'))$. In particular, this property holds for all ${\cal S}_{\gamma,j}$ since each ${\cal S}_{\gamma,j}$ is a consecutive sublist of ${\cal S}_{\gamma}$.
\label{lemma:MonotonicallyIncreasing}
\end{proposition}

We store some navigation structures supporting efficient navigation around artificial edges. These structures include the labeling table, the witness function $f$ and for each $\gamma$, the segmentation hierarchy ${\cal S}_{\gamma}$, the adjacency list $A_{\gamma}$ and the connector $B_{\gamma}$. Besides, we store the following witness lists $C_{\gamma}$ and $C_{\gamma,j}$ for each $\gamma$ and each level $j$. For each $\gamma$, we let $C_{\gamma}=\{f(S)\mid S\in{\cal S}_{\gamma}\}$ (without multiplicity) be an ordered list collecting all witnesses in ${\cal S}_{\gamma}$, ordered by $\ET_{\glo}$. Analogously we define an ordered list $C_{\gamma,j}=\{f(S)\mid S\in{\cal S}_{\gamma,j}\}$ for each level $j$ but $C_{\gamma,j}$ is weighted. Concretely, for each vertex $u\in C_{\gamma,j}$, we let $u$ have $C_{\gamma,j}$-weight $|f^{-1}(u,{\cal S}_{\gamma,j})|$, where $f^{-1}(u,{\cal S}_{\gamma,j})$ is the union of segments $S\in {\cal S}_{\gamma,j}$ such that $f(S)=u$. We proprocess the prefix sum for each $C_{\gamma,j}$ so that the total $C_{\gamma,j}$-weights of a given range can be answered in $O(1)$ time. We note that each $C_{\gamma,j}$ is a consecutive sublist of $C_{\gamma}$ by \Cref{lemma:MonotonicallyIncreasing}.

\begin{lemma}
The above construction algorithm for $\hat{E}$ takes $O(d_{\star}pm\log n)$ time and extra $O(pm)$ space to store the navigation structures. The size of $\hat{E}$ is $O(d_{\star}pm\log n)$ and the number of distinct ordered pairs in $\hat{E}$ is $O(pm\log n)$. 
\label{lemma:SizeOfEhat}
\end{lemma}
\begin{proof}
First we show that the construction algorithm takes $O(d_{\star}pm\log n)$ time. For each component $\gamma$, we need to scan the region $A_{\gamma}\times B_{\gamma}$ of the labeling table, where $|A_{\gamma}\times B_{\gamma}|=O(d_{\star}|A_{\gamma}|)$. Summing over all $\gamma$, step (1) takes totally $O(d_{\star}pm\log n)$ time by \Cref{lemma:TotalSizeOfA} where the $O(\log n)$ factor is from accessing the binary search tree of the table. %
Step (2) and (3) takes $O(|A_{\gamma}|)$ time for each $\gamma$, so the total time is $O(pm)$ summing over all $\gamma$. 

Regarding the space taken by the navigation structures, the labeling table takes space linear to the number of labeled pairs, which is exactly $\sum_{\gamma}|\Tag_{\gamma}|=\sum_{\gamma}O(|A_{\gamma}|)=O(pm)$. Other navigation structures take totally $O(pm)$ space following the fact that for each $\gamma$, $|B_{\gamma}|\leq|A_{\gamma}|$ and both $\sum_{j}|C_{\gamma,j}|$ and $|C_{\gamma}|$ are at most $|{\cal S}_{\gamma}|=O(|A_{\gamma}|)$.  Therefore, the extra space is $O(pm)$.

We show that the number of distinct pairs in $\hat{E}$ is $O(pm\log n)$ by bounding the size of $\bigcup_{\gamma}A_{\gamma}\times B_{\gamma}$ and the size of $\bigcup_{\gamma,S\in{\cal S}_{\gamma}}S\times f(S)$, not counting mulplicity. The latter is bounded by $O(pm\log n)$ because $\sum_{\gamma}\sum_{S\in{\cal S}_{\gamma}}|S\times f(S)|=\sum_{\gamma}O(|A_{\gamma}|\log n)=O(pm\log n)$ where the first equation follows that the total size of segments in ${\cal S}_{\gamma}$ is $O(|A_{\gamma}|\log n)$ for each $\gamma$. It remains to bound the size of $\bigcup_{\gamma}A_{\gamma}\times B_{\gamma}$. Observe that all such pairs will be labeled in the labeling table, where the number of labeled pairs is $O(pm)$ as shown above, so $|\bigcup_{\gamma}A_{\gamma}\times B_{\gamma}|=O(pm)$. Similarly, the size of $\Ehat$ counting multiplicity is $|\hat{E}|=\sum_{\gamma}|A_{\gamma}\times B_{\gamma}|+\sum_{\gamma}\sum_{S\in{\cal S}_{\gamma}}|S\times f(S)|=O(d_{\star}pm+pm\log n)=O(d_{\star}pm\log n)$.

\end{proof}

We have finished the description of affiliated structures and we conclude the construction time and space in the lemma below.

\begin{lemma}
Given an $n$-vertex $m$-edge graph $G$ with the low degree hierarchy, there is a deterministic algorithm construct the affiliated structures in time $O(d_{\star}pm\log^2 n)$. Furthermore, the affiliated structures take space $O(pm\log^2 n)$.
\label{lemma:Space}
\end{lemma}
\begin{proof}
The 2D range counting structure in \Cref{sect:2DRangeStructure} needs $O(d_{\star}pm\log^2 n)$ construction time and $O(pm\log^2 n)$ space by \Cref{lemma:2DRangeCounting}, because the size of $E=E(G)+\hat{E}$ is $O(d_{\star}pm\log n)$ and there are only $O(pm\log n)$ distinct pairs in $E$ by \Cref{lemma:SizeOfEhat}. The construction of artificial edges takes $O(d_{\star}pm\log n)$ time and extra space $O(pm)$ by \Cref{lemma:SizeOfEhat}.

\end{proof}

\

\noindent\textbf{Construction of $D$-repairing Edges.} In what follows, we will show the algorithm for computing $D$-repairing edges $\hat{F}_{\gamma}$ for each $\gamma\in{\cal C}$ given any $D$ with $|D|=d\leq d_{\star}$, which is a subroutine invoked by the update algorithm. We begin by defining \textit{unreliable components} and only these components will have a nonempty $\hat{F}_{\gamma}$. First, each component $\gamma$ with $|B_{\gamma}|=\min\{|A_{\gamma}|,d_{\star}+1\}$ is reliable because $\Ehat_{\gamma}$ has been fault-tolerant.
For each remaining component $\gamma$, we say a segment $S\in{\cal S}_{\gamma}$ is an \textit{unreliable segment} if its tag $(u,v)$ has both $u,v\in D$, and we say $\gamma$ is unreliable if there is at least an unreliable segment in ${\cal S}_{\gamma}$. 

\begin{itemize}
    \item[(1)] We can find all unreliable components with the unreliable segments by scanning the region $D\times D$ of the labeling table.
    Concretely, for a label $(\gamma,k)$ of some entry $(u,v)\in D\times D$, if $k\leq |{\cal S}_{\gamma}|$, the $k$-ranked segment in ${\cal S}_{\gamma}$ is an unreliable segment of $\gamma$ and $\gamma$ is unreliable. 
    
    \item[(2)] For each unreliable component $\gamma$, we find a vertex $v_{\gamma}\in A_{\gamma}\setminus D$ as follows. First we scan the list $\Tag_{\gamma}$ until we find a tag $(u,v)$ with $u\notin D$ or $v\notin D$, which means we obtain $v_{\gamma}$ since all $(u,v)$ in $\Tag_{\gamma}$ are in $A_{\gamma}\times A_{\gamma}$. If we fail to find a $v_{\gamma}$ by scanning $\Tag_{\gamma}$, we directly scan the list $A_{\gamma}$ to get $v_{\gamma}$. If we still cannot find $v_{\gamma}$, it means $A_{\gamma}\subseteq D$.
    
    \item[(3)] For each unreliable component $\gamma$ with a $v_{\gamma}\in A_{\gamma}\setminus D$ found in step (2), we do the following for each of its unreliable segments $S$ to generate $\hat{F}_{\gamma}$ (initially $\hat{F}_{\gamma}$ is empty). If $S$ is a singleton set (i.e. $S\in {\cal S}_{1}$) containing a vertex $v$, we add an edge $(v_{\gamma},v)$ to $\hat{F}_{\gamma}$. Otherwise, $S$ has at most two child-segments $S_{1}$ and $S_{2}$ in the segmentation hierarchy. Let $(u_{1},v_{1})$ and $(u_{2},v_{2})$ be tags of $S_{1}$ and $S_{2}$. For each of $u_{1},v_{1},u_{2},v_{2}$, add an edge connecting it with $v_{\gamma}$ to $\hat{F}_{\gamma}$. 
\end{itemize}

\begin{lemma}
Given a set $D$ where $|D|=d\leq d_{\star}$, there is a deterministic algorithm computing a set of $D$-repairing edges $\hat{F}_{\gamma}$ for each $\gamma\in{\cal C}$, such that vertices in $A_{\gamma}\setminus D$ are connected by edges $\hat{E}_{\gamma}\cup\hat{F}_{\gamma}$ (or possibly $A_{\gamma}\subseteq D$). The size of $\hat{F}=\bigcup_{\gamma\in{\cal C}}\hat{F}_{\gamma}$ is $O(d^{2})$ and the running time of this algorithm is $O(d^{2}\log n)$.
\label{lemma:SensitiveEdges}
\end{lemma}
\begin{proof}
First we show that for each reliable component $\gamma$, $\hat{E}_{\gamma}$ can connect all vertices in $A_{\gamma}\setminus D$. By the definition, a reliable $\gamma$ will either have $|B_{\gamma}|=\min\{|A_{\gamma}|,d_{\star}+1\}$ or all segments in ${\cal S}_{\gamma}$ are reliable. In the former case, we have either $A_{\gamma}\subseteq D$ or some $v\in B_{\gamma}\setminus D\subseteq A_{\gamma}\setminus D$ connects all vertices in $A_{\gamma}$, so vertices in $A_{\gamma}\setminus D$ are connected to $v$ by $\Ehat_{\gamma}$. In the latter case, consider the root segment $S=A_{\gamma}$ of ${\cal S}_{\gamma}$. The tag $(u,v)$ of $S$ has either $u=f(S), u\notin D$ or $v\in B_{\gamma}, v\notin D$. Because $A_{\gamma}\times\{u\}= S\times f(S)\subseteq \hat{E}_{\gamma}$ and $A_{\gamma}\times\{v\}\subseteq A_{\gamma}\times B_{\gamma}\subseteq \hat{E}_{\gamma}$, vertices in $A_{\gamma}\setminus D$ are connected by $\hat{E}_{\gamma}$.

For each unreliable component $\gamma$, we will show that vertices in $A_{\gamma}\setminus D$ are connected by $\hat{E}_{\gamma}\cup \hat{F}_{\gamma}$. Assume that the root segment $S=A_{\gamma}$ is unreliable, otherwise vertices in $A_{\gamma}\setminus D$ are connected by $\hat{E_{\gamma}}$ by the same argument. Also assume that in step (2) we found the vertex $v_{\gamma}\in A_{\gamma}\setminus D$, otherwise we have $A_{\gamma}\subseteq D$. We claim that each $v\in A_{\gamma}\setminus D$ is connected to $v_{\gamma}$ by $\hat{E}_{\gamma}\cup\hat{F}_{\gamma}$. For each $v\in A_{\gamma}\setminus D$, let $S=\{v\}$ be the singleton segment of $v$ in ${\cal S}_{\gamma}$. If $S$ is unreliable, then an edge $(v_{\gamma},v)$ added to $\hat{F}_{\gamma}$ in step (3) will connect $v$ to $v_{\gamma}$. Otherwise, $S$ is reliable and because the root segment is unreliable, there is a reliable segment $S'$ (an ancestor of $S$ in ${\cal S}_{\gamma}$) such that $v\in S'$ and the parent segment of $S'$ is unreliable. Let $(u',v')$ be the tag of $S'$. In step (3), we have either $u'\notin D,(u',v_{\gamma})\in\hat{F}_{\gamma}$ or $v'\notin D,(v',v_{\gamma})\in\hat{F}_{\gamma}$. Combining the fact that $v\in S'$ and $S'\times\{u',v'\}\subseteq (S'\times f(S'))\cup(A_{\gamma}\times B_{\gamma}) \subseteq \hat{E}_{\gamma}$, we have $v$ and $v_{\gamma}$ are connected by $\hat{E}_{\gamma}\cup\hat{F}_{\gamma}$.

The bound $|\hat{F}|=O(d^{2})$ is obtained as follows. From the algorithm, the number of $D$-repairing edges we added is linear to the number of unreliable segments, which is equal to the number of labeled entries in region $D\times D$ of the labeling table, which is at most $O(d^{2})$. 

The running time of this algorithm is $O(d^{2}\log n)$. Observe that in step (1) we take $O(d^{2}\log n)$ time to scan the region $D\times D$ of the labeling table (the $\log n$ factor is from the underlying binary search tree of the table) and step (3) takes time linear to the number of unreliable segments, which is $O(d^{2})$. To bound the running time of step (2), we use the fact that tags in all $\Tag_{\gamma}$ are distinct and charge the running time on tags in $D\times D$ and unreliable components. Concretely, if we found $v_{\gamma}$ when scanning $\Tag_{\gamma}$, the running time is linear to one plus the number of tags in $\Tag_{\gamma}\cap(D\times D)$ we scanned. Then we charge $O(1)$ costs on the unreliable $\gamma$ and each such tag. If $v_{\gamma}$ is found by taking $O(|A_{\gamma}|)$ time to directly scan the whole $A_{\gamma}$, we charge $O(1)$ costs on each tag in $\Tag_{\gamma}\cap(D\times D)$, since $|\Tag_{\gamma}\cap(D\times D)|=|\Tag_{\gamma}|=\Theta(|A_{\gamma}|)$ by \Cref{prop:TagSize}. Therefore, the total running time of step (2) is proportional to the number of unreliable components plus the number of label pairs in $D\times D$, which is $O(d^{2})$.

\end{proof}

\section{Update and Query Algorithms}
\label{sect:UpdateAndQuery}

In this section, we will show the update and query algorithms. The update algorithm is introduced in \Cref{sect:UpdateAlgo} and \Cref{sect:IntervalConnectivity}, and we describe the query algorithm in \Cref{sect:QueryAlgo}. We conclude the main result of $d$-vertex-failure connectivity oracles in \Cref{thm:MainDetailed}..

\begin{theorem}
Given a $n$-vertex $m$-edge undirected graph $G$ with an upper bound $d_{\star}$ of the size of failure sets, there exists deterministic $d_{\star}$-vertex-failure connectivity oracles that use either
\begin{itemize}
    \item $O(m\cdot \beta(n))$ space, $\hat{O}(m)+\tilde{O}(md_{\star})$ preprocessing time, $\hat{O}(d^{2})$ update time and $O(d)$ query time, where $\beta(n)$ can be any slowly growing function e.g. $\beta(n)=\log^{*}n$, or
    \item $O(m\log^{3}n)$ space\footnote{The space can be improved to $O(m\log^{2}n)$ as discussed in \Cref{remark:LinearSpace2DCounting} at a cost of increasing the update time by an $O(\log^{1+\epsilon}n)$ factor for any $\epsilon>0$.}, $O(mn\log n)$ preprocessing time, $\bar{O}(d^{2}\log^{3}n\log^{4}d)$ update time and $O(d)$ query time.
\end{itemize}
All the $m$ factors can be replaced by $\bar{m}=\min\{m,n(d_{\star}+1)\}$ by the standard sparsification algorithm in \cite{nagamochi1992linear}, at a cost of $O(m)$ additional preprocessing time.
\label{thm:MainDetailed}
\end{theorem}

The update time is from \Cref{thm:UpdateAlgorithm}. The query time is from \Cref{sect:QueryAlgo}. The preprocessing time is from \Cref{thm:LowDegreeHierarchy} (or \Cref{thm:DPLowDegreeHierarchy}) on the low degree hierarchy and \Cref{lemma:Space} on the affiliated structures. See \Cref{remark:LinearSpace2DCounting} and \Cref{remark:MultipleBranchedSeg} for some improvements on the space.

\subsection{The Update Algorithm}
\label{sect:UpdateAlgo}

Given a failure set $D$ with $|D|=d\leq d_{\star}$, the update algorithm will compute the connectivity of second-type objects in the abstract graph $H$. Let $\hat{{\cal T}}$ denote all second-type objects. Let ${\cal C}_{D}$ denotes the set of affected components and we have $|{\cal C}_{D}|=O(pd)$ since there are at most $d$ affected components each level.

\begin{theorem}
Let $G$ be a graph with a $(p,\Delta)$-low degree hierarchy and affiliated structures prepared. Given a set of $D$ failed vertices with $|D|=d\leq d_{\star}$, there is a deterministic algorithm that computes a partition $\hat{\cal R}$ of $\hat{{\cal T}}$ such that for each $\hat{\tau},\hat{\tau}'\in\hat{{\cal T}}$, they are connected in $H$ if and only if they are in the same group of $\hat{{\cal R}}$. The running time of this algorithm is $O(\Delta^{2}p^{2}d^{2}\log n+\Delta p^{2}d^{2}\log^{4}(\Delta pd)\log n)$.
\label{thm:UpdateAlgorithm}
\end{theorem}

Instead of constructing the abstract graph $H$ and working directly on $H$, our algorithm will actually consider some disjoint intervals on $\ET_{\glo}$, obtained by ``linearlizing'' the affected subtrees. Strictly speaking, for each affected tree, we apply \Cref{lemma:TreeToInterval} on it with failed vertices $D\cap V(\tau)$, and get a set of disjoint intervals on $\ET(\tau)$, each of which is owned by an affected subtree. Recall that $U(\tau)$ is a consecutive sublist of $\ET_{\glo}$. By taking the restrictions of these intervals on $U(\tau)$,%
we can obtain some disjoint intervals on $\ET_{\glo}$. After processing all affected trees, let ${\cal I}$ denote the whole set of intervals on $\ET_{\glo}$ we obatined. 

\begin{proposition}
The number of intervals is $|{\cal I}|=O(\Delta pd)$ and the time to construct ${\cal I}$ is $O(\Delta pd\log n)$. For each affected subtrees $\hat{\tau}\in\hat{\cal \tau}$, the union of intervals it owns is exactly its non-failed terminals $U(\hat{\tau})\setminus D$.
\label{prop:Iproperty}
\end{proposition}
\begin{proof}
Recall that in ${\cal T}$, the Steiner trees in the same level are disjoint. Thus there are at most $O(pd)$ failed vertices appearing on all trees in ${\cal T}$. The trees with failed vertices can be easily identified in $O(pd)$ time since the low degree hierarchy is stored explicitly. Then this proposition directly follows \Cref{lemma:TreeToInterval}. 
\end{proof}

Our algorithm will connect intervals in ${\cal I}$ using a multiset $E_{D}$ of ordered pairs (equivalent to undirected edges) defined by
\[
E_{D}=E-\bar{E}_{D}=E-\sum_{\gamma\in{\cal C}_{D}}\hat{E}_{\gamma}=E-\sum_{\gamma\in{\cal C}_{D}}\left(A_{\gamma}\times B_{\gamma}+\sum_{S\in{\cal S}_{\gamma}}S\times f(S)\right),
\]
where $E$ are edges we stored in the 2D range counting structure and $\bar{E}_{D}=\sum_{\gamma\in{\cal C}_{D}}\Ehat_{\gamma}$ represents ``invalid'' artificial edges of affected components.
Observe that two intervals $I_{1},I_{2}\in{\cal I}$ are connected by an edge in $E_{D}$ whenever there is an ordered pair in $E_{D}$ with one entry in $I_{1}$ and the other one in $I_{2}$, namely $|E_{D}\cap(I_{1}\times I_{2})|+|E_{D}\cap(I_{2}\times I_{1})|>0$. We say $I_{1}$ and $I_{2}$ are \textit{adjacent} in this case. %

\begin{lemma}
With access to the low-degree hierarchy and affiliated structures, there is a deterministic algorithm that computes a partition ${\cal R}$ of ${\cal I}$ such that for any $I,I'\in{\cal I}$, they are connected by a path only with edges in $E_{D}$ %
if and only if they are in the same group of ${\cal R}$. The running time of this algorithm is $O(|{\cal I}|^{2}\log n+|{\cal I}|\cdot|{\cal C}_{D}|\log^{4}|{\cal I}|\log n)$.
\label{lemma:IntervalConnectivity}
\end{lemma}

In \Cref{sect:IntervalConnectivity}, we will show the algorithm in \Cref{lemma:IntervalConnectivity}, which computes the connectivity of intervals over $E_{D}$, but before that, we complete our update algorithm using \Cref{lemma:IntervalConnectivity} as a subroutine. 

The last step of our update algorithm is to compute the desired partition $\hat{\cal R}$ of affected subtrees by further merging groups of intervals through the connectivity from affected subtrees and $D$-repairing edges.
Let ${\cal R}$ be the partition of ${\cal I}$ computed by \Cref{lemma:IntervalConnectivity}.  %
Let $\hat{F}_{D}=\sum_{\gamma\in{\cal C}\setminus{\cal C}_{D}}\hat{F}_{\gamma}$ by collecting $D$-repairing edges of unaffected components. Then we use edges in $\hat{F}_{D}$ to merge groups in ${\cal R}$. Concretely, for each $(u,v)\in\hat{F}_{D}$, if $u$ and $v$ belong to some $I_{u},I_{v}\in{\cal I}$ respectively, we merge the groups containing $I_{u}$ and $I_{v}$ into a new group of ${\cal R}$. Lastly, for each affected subtree $\hat{\tau}$, we merge the groups with at least one interval owned by $\hat{\tau}$ into a new group. The partition $\hat{\cal R}$ of $\hat{\cal T}$ is constructed by, for each group $R\in{\cal R}$, putting the owners of intervals in $R$ into a group $\hat{R}\in\hat{\cal R}$. 

Now we show the correctness of the update algorithm assuming \Cref{lemma:IntervalConnectivity}.

\begin{proof}[Proof of \Cref{thm:UpdateAlgorithm}]

For each second-type object $\hat{\tau}\in\hat{\cal T}$, let $\hat{\cal R}(\hat{\tau})$ denote the group containing $\hat{\tau}$, and let ${\cal I}(\hat{\tau})$ denote intervals owned by $\hat{\tau}$. 

First, we will show that if two second-type object $\hat{\tau},\hat{\tau}'\in\hat{\cal T}$ are connected in $H$, then $\hat{\cal R}(\hat{\tau})$ and $\hat{\cal R}(\hat{\tau}')$ are the same. Let $P$ be a simple path connecting $\hat{\tau}$ and $\hat{\tau}'$ in $H$. Let $\hat{\tau}_{1}$ and $\hat{\tau}_{2}$ be any two second-type objects on $P$ such that there is no other second-type object between them on $P$. We will show that $\hat{\tau}_{1}$ and $\hat{\tau}_{2}$ belong to the same group of $\hat{\cal R}$, so $\hat{\tau}$ and $\hat{\tau}'$ are in the same group. To show the claim, there are two cases.

\

\noindent\textbf{Case 1.}%
$\hat{\tau}_{1}$ and $\hat{\tau}_{2}$ are adjacent via an edge on $P$. From the definition of $H$ (\Cref{def:AbstractGraph}), this edge corresponds to an original edge in $E(G)$ which connects one vertex $v_{1}\in U(\hat{\tau}_{1})$ and another vertex $v_{2} \in U(\hat{\tau}_{2})$. Therefore, $v_{1}$ and $v_{2}$ belong to some $I_{1}\in {\cal I}(\hat{\tau_{1}})$ and $I_{2}\in {\cal I}(\hat{\tau_{2}})$ respectively, %
which implies $I_{1}$ and $I_{2}$ are in the same group of ${\cal R}$ and, $\hat{\tau}_{1}$ and $\hat{\tau}_{2}$ are in the same group of $\hat{\cal R}$.

\

\noindent\textbf{Case 2.} $\hat{\tau}_{1}$ and $\hat{\tau}_{2}$ are connected by a subpath through some first-type objects. By \Cref{lemma:PseudoBipartiteH}, the subpath will only go through a unique first-type object $\gamma$. Two edges of this subpath correspond to two edges $(v_{1},u_{1})$ and $(u_{2},v_{2})$ in $E(G)$, where $v_{1}\in U(\hat{\tau}_{1})$, $v_{2}\in U(\hat{\tau}_{2})$ and $u_{1},u_{2}\in V(\gamma)$, which implies $v_{1},v_{2}\in A_{\gamma}\setminus D$. Assume for a moment that each $v\in A_{\gamma}\setminus D$ belongs to some interval $I_{v}\in{\cal I}$. Also by \Cref{lemma:SensitiveEdges}, vertices in $A_{\gamma}\setminus D$ will be connected by edges in $\hat{E}_{\gamma}\cup\hat{F}_{\gamma}\subseteq E_{D}\cup\hat{F}_{D}$. Therefore, these intervals $I_{v}$ for all $v\in A_{\gamma}\setminus D$ are in the same group of ${\cal R}$, which implies $\hat{\tau}_{1}$ and $\hat{\tau}_{2}$ are in the same group of $\hat{\cal R}$. To see why each $v\in A_{\gamma}\setminus D$ belong to some $I_{v}\in{\cal I}$, %
observe that there is an original edge connecting $v$ and $\gamma$, so $v$ cannot belong to any other first-type object by \Cref{lemma:PseudoBipartiteH}, which means $v$ must belong to some $U(\hat{\tau})$ and $v$ is in some interval $I$ owned by $\hat{\tau}$. 

\

On the other direction, we show that if $\hat{\tau}$ and $\hat{\tau}'$ are in the same group of $\hat{\cal R}$, then $\hat{\tau}$ and $\hat{\tau}'$ are connected in $H$. 
Let $I$ and $I'$ be intervals owned by $\hat{\tau}$ and $\hat{\tau}'$ respectively. Then they are in the same group of ${\cal R}$ by the algorithm. We claim that the vertices of $I$ are connected to those of $I'$ in $G\setminus D$. Therefore, $\hat{\tau}$ and $\hat{\tau}'$ are connected in $H$ by \Cref{lemma:ConnEqOrigianlH}. The claim indeed holds because we group intervals using valid connectivity shown as follows. Because there are no failed vertices in intervals, the algorithm will automatically connect intervals using original edges, artificial edges and $D$-repairing edges in $E_{D}\cup\hat{F}_{D}$ not incident to failed vertices. Each of such artificial edges and repairing edges can be substituted by a valid path in $G\setminus D$ through some unaffected component. Merging groups with intervals owned by the same affected subtree is also valid since each affected subtree is a connected subgraph of $G\setminus D$.

Regarding the running time, first constructing ${\cal I}$ takes $O(\Delta pd\log n)$ time by \Cref{prop:Iproperty}.
Solving connectivity of intervals by \Cref{lemma:IntervalConnectivity} takes $O(|{\cal I}|^{2}\log n+|{\cal I}|\cdot|{\cal C}_{D}|\log^{4}|{\cal I}|\log n)$ time. By \Cref{lemma:SensitiveEdges}, computing $D$-repairing edges $\hat{F}$ takes $O(d^{2}\log |{\cal I}|)$ time and connecting intervals using edges in $\hat{F}_{D}$ takes $O(|\hat{F}_{D}|\log |{\cal I}|)=O(d^{2}\log|{\cal I}|)$ time, where the $O(\log |{\cal I}|)$ factor is from querying intervals containing the endpoints of each edge. Lastly, it spends extra $O(|{\cal I}|)$ time on connecting intervals owned by the same affected subtree. In total, the running time is $O(\Delta^{2}p^{2}d^{2}\log n+\Delta p^{2}d^{2}\log^{4}(\Delta pd)\log n)$ by plugging in $|{\cal I}|=O(\Delta pd)$ and $|{\cal C}_{D}|=O(pd)$.
\end{proof}

\subsection{Solving Connectivity of Intervals: Proof of \Cref{lemma:IntervalConnectivity}}
\label{sect:IntervalConnectivity}

In this section we will show the algorithm in \Cref{lemma:IntervalConnectivity}. We will connect the intervals using the ``hook and contract'' framework as in \Boruvka's MST algorithm. Starting with singleton sets of intervals, the framework will keep merging the sets until all of them are maximal sets of connected intervals.

Let $Z_{1,k}=\{I_{k}\}$ be the initial singleton set for all $I_{k}\in{\cal I}$, and we let ${\cal Z}_{1}$ collect all of them. During the algorithm, we will maintain a set ${\cal R}$ (initially ${\cal R}=\emptyset$), which collects maximal sets of connected intervals. At the $t$-th \Boruvka's phase, we will find for each $Z_{t,k}\in{\cal Z}_{t}$ a neighbor, which is another $Z_{t,k'}\in{\cal Z}_{t}$ such that there exists $I\in Z_{t,k}$ and $I'\in Z_{t,k'}$ adjacent. If there is no neighbor for some $Z_{t,k}$, it means that $Z_{t,k}$ has been a maximal set of connected intervals, and we add it into ${\cal R}$. Finally, we merge the sets with neighbors according to the neighbor-relations, which generates the collection ${\cal Z}_{t+1}$ for the next phase. If ${\cal Z}_{t+1}$ is empty, we return ${\cal R}$ as the desired partition. One cay easily verify that $|{\cal Z}_{t+1}|\leq |{\cal Z}_{t}|/2$ for each phase $t$, so the number of phases is bounded by $\log(|{\cal Z}_{1}|)=\log|{\cal I}|$. %

\subsubsection{The \textsc{FindNeighbor} Function}
\label{sect:FindNeighbors}

The ``find-neighbor'' operations are delegated to the function $\FindNeighbors(t,Z_{t,k})$. The inputs are the phase number $t$ and some $Z_{t,k}\in {\cal Z}_{t}$, and the required output is another $Z_{t,k'}\in{\cal Z}_{t}$ adjacent to $Z_{t,k}$ (or claiming such $Z_{t,k'}$ does not exist), where we say $Z_{t,k}$ and $Z_{t,k'}$ are \textit{adjacent} if there exists two adjacent intervals $I\in Z_{t,k}$ and $I'\in Z_{t,k'}$.

First, we let ${\cal I}_{t}$ be an ordered list of intervals in ${\cal Z}_{k}$ by viewing each $Z_{t,k}$ a list of intervals with an arbitrary order and taking the concatenation from $Z_{t,1}$ to $Z_{t,|{\cal Z}_{k}|}$. %
By the definition, each $Z_{t,k}$ is a set of intervals locating consecutively in the list ${\cal I}_{t}$. Let $\BatchConnCheck(t,k,l,r)$ be a function which returns a boolean indicating the existence of $Z_{t,k'}$ such that $l\leq k'\leq r$ and $Z_{t,k'}$ is adjacent to $Z_{t,k}$. With access to the function $\BatchConnCheck(t,k,l,r)$, the function $\FindNeighbors(t,Z_{t,k})$ can use a binary search scheme to find the minimum $k_{r}$ such that $k_{r}>k$ and $Z_{t,k}$ and $Z_{t,k_{r}}$ are adjacent. If there is no such $k_{r}$ exists, it will perform a symmetric binary search on the other side.

Implementing the indicator $\BatchConnCheck(t,k,l,r)$ is equivalent to solving a counting problem. By the definition of ${\cal I}_{t}$, intervals in $\bigcup_{k'=l}^{r}Z_{t,k'}$ are located consecutively on ${\cal I}_{t}$ and we let them be ${\cal I}_{t}(x,y)$. The indicator responds affirmatively if and only if the following value $\delta(t,k,x,y)$ is greater than zero.
\[
\delta(t,k,x,y)=\sum_{I\in Z_{t,k}}\sum_{I'\in{\cal I}_{t}(x,y)}|E_{D}\cap (I\times I')|+|E_{D}\cap(I'\times I)|.
\]
By the fact that $E_{D}=E-\sum_{\gamma\in{\cal C}_{D}}(A_{\gamma}\times B_{\gamma}+\sum_{S\in{\cal S}_{\gamma}}S\times f(S))$, we can divide the whole counting problem into several subproblems.
\[
\delta(t,k,x,y)=\delta_{1}(t,k,l,r)-\delta_{2}(t,k,l,r)-\delta_{3}(t,k,l,r),
\]
where
\begin{align*}
\delta_{1}(t,k,x,y)&=\sum_{I\in Z_{t,k}}\sum_{I'\in{\cal I}_{t}(x,y)}|E\cap(I\times I')|+|E\cap(I'\times I)|,\\
\delta_{2}(t,k,x,y)&=\sum_{\gamma\in{\cal C}_{D}}\sum_{I\in Z_{t,k}}\sum_{I'\in{\cal I}_{t}(x,y)}|(A_{\gamma}\times B_{\gamma})\cap(I\times I')|+|(A_{\gamma}\times B_{\gamma})\cap(I'\times I)|,\\
\delta_{3}(t,k,x,y)&=\sum_{\gamma\in{\cal C}_{D}}\sum_{S\in{\cal S}_{\gamma}}\sum_{I\in Z_{t,k}}\sum_{I'\in{\cal I}_{t}(x,y)}|(S\times f(S))\cap(I\times I')|+|(S\times f(S))\cap(I'\times I)|.
\end{align*}

In what follows, we will construct some additional structures and use them to compute $\delta_{1}(t,k,x,y)$ and $\delta_{2}(t,k,x,y)$ in \Cref{sect:AdditionalStructure}. In section \Cref{sect:ComputingDelta3}, we will show the algorithm for computing $\delta_{3}(t,k,x,y)$.

\subsubsection{The Additional Structures}
\label{sect:AdditionalStructure}
We preprocess the following before the \Boruvka's algorithm starts. First we query the 2D range counting structure for each $I,I'\in{\cal I}$ and construct a counting table on ${\cal I}\times{\cal I}$ where we store explicitly the answer $|E\cap(I\times I')|$ in the entry $(I,I')$. For each interval $I\in{\cal I}$, compute and store its restriction on each list $A_{\gamma},B_{\gamma},C_{\gamma,j}$ for each affected component $\gamma$ and level $j$ of ${\cal S}_{\gamma}$. 

We then construct the following data structures at the beginning of the $t$-th \Boruvka's phase.
\begin{itemize}
    \item[(A)] We shuffle the initial counting table to get a new counting table $\{|E\cap(I\times I')|\}_{I\times I'\in{\cal I}_{t}\times{\cal I}_{t}}$ which indices are consistent with the order of ${\cal I}_{t}$. By preparing the 2D-prefix sum, for any 2D-range sum query given $1\leq a_{1}\leq a_{2}\leq |I|,1\leq b_{1}\leq b_{2}\leq |I|$, the sum $\sum_{a_{1}\leq a\leq a_{2}}\sum_{b_{1}\leq b\leq b_{2}}|E\cap({\cal I}_{t}(a)\times{\cal I}_{t}(b))|$ can be answered in $O(1)$ time.%
    \item[(B)] For each affected component $\gamma$, we construct a counting array $\hat{A}_{\gamma}=\{|I\cap A_{\gamma}|\}_{I\in{\cal I}_{t}}$, indexed by ${\cal I}_{t}$. It supports $O(1)$-time range sum query by calculating the prefix sum in advanced. Namely, given any $1\leq a_{1}\leq a_{2}\leq |I|$, the sum $\sum_{a_{1}\leq a\leq a_{2}}|{\cal I}_{t}(a)\cap A_{\gamma}|$ can be returned. Analogously, we construct such an array $\hat{B}_{\gamma}$ for $B_{\gamma}$ of each affected component $\gamma$.
    
    \item[(C)] For each affected component $\gamma$, we construct for $A_{\gamma}$ a set $M_{\gamma}$ of weighted triples. For each $1\leq i\leq |{\cal I}_{t}|$, there is a triple $(i,l,r)$ where $l$ and $r$ are the $A_{\gamma}$-ranks of endpoints of the restriction ${\cal I}_{t}(i)\cap A_{\gamma}$, and this triple has weight $|{\cal I}_{t}(i)\cap A_{\gamma}|$. We then construct the data structure in \Cref{lemma:3DRangeCounting} on $M_{\gamma}$.
    
    For each affected component $\gamma$ and each level $j$ of ${\cal S}_{\gamma}$, we construct for $C_{\gamma,j}$ a set $N_{\gamma,j}$ with the data structure in \Cref{lemma:3DRangeCounting} analogously, except that the weight of each triple $(i,l,r)$ is the total $C_{\gamma,j}$-weight of vertices in ${\cal I}_{t}(i)\cap C_{\gamma,j}$.
\end{itemize}

\begin{lemma}
Given a set $Q$ of weighted triples, there is a data structure supporting the following queries in $O(\log^2 |Q|)$ time. Given $1\leq x\leq y\leq |{\cal I}|,1\leq L\leq R\leq n$, the interface $\threeDRangeSum(x,y,L,R)$ will return the total weight of triples $(i,l,r)\in Q$ such that $x\leq i\leq y$ and $L\leq l\leq r\leq R$. The data structure occupies $O(|Q|\log^{2}|Q|)$ space and it can be constructed in $O(|Q|\log^{2}|Q|)$ time.
\label{lemma:3DRangeCounting}
\end{lemma}

To prove \Cref{lemma:3DRangeCounting}, it suffices to show a weighted 3D range counting structure. Note that since the data structure in \Cref{lemma:3DRangeCounting} will be constructed in the update phase, there is no need to reduce the space, and a textbook data structure suffices, e.g. the range tree \cite{Ben80} or a persistent segment tree \cite{BW80,DSST89}.

\begin{lemma}
Constructing the additional data structures takes $O(|{\cal I}|^{2}\log n+|{\cal I}|\cdot|{\cal C}_{D}|\log^{3}|{\cal I}|\log n)$ time.
\label{lemma:AddStructConstruct}
\end{lemma}
\begin{proof}
First we analyse the preprocessing step and show it takes $O(|{\cal I}|^{2}\log n + |{\cal I}|\cdot|{\cal C}_{D}|\log n)$ time. To construct the initial counting table, we need $|{\cal I}|^{2}$ queries on the 2D range counting structure, so it takes totally $O(|{\cal I}|^{2}\log n)$ by \Cref{lemma:2DRangeCounting}. Querying the restriction of some interval on some list $A_{\gamma}$ or $B_{\gamma}$ takes $O(\log n)$ time, so all restrictions on lists $A_{\gamma}$ and $B_{\gamma}$ can be computed in $O(|I|\cdot|{\cal C}_{D}|\log n)$ time summing over $|{\cal I}|$ intervals and $|{\cal C}_{D}|$ affected components. Regarding the restrictions on $C_{\gamma,j}$, note that $C_{\gamma,j}$ is a consecutive sublist of $C_{\gamma}$, so after querying the restriction $I\cap C_{\gamma}$ in $O(\log n)$ time, we can get $I\cap C_{\gamma,j}$ for each level $j$ of ${\cal S}_{j}$ in $O(1)$ time. There are $O(\log n)$ levels in each ${\cal S}_{\gamma}$, so the total time is $O(|I|\cdot|{\cal C}_{D}|\log n)$. 

The construction of additional structures for each \Boruvka's phase takes $O(|{\cal I}|^{2}+|{\cal I}|\cdot|{\cal C}_{D}|\log^{2}|{\cal I}|\log n)$ time as shown below, so the total time is $O(|{\cal I}|^{2}\log |{\cal I}|+|{\cal I}|\cdot|{\cal C}_{D}|\log^{3}|{\cal I}|\log n)$ summing over $\log|{\cal I}|$ phases. Consider each phase. Part (A) trivially takes $O(|{\cal I}|^{2})$ time and part (B) takes $O(|{\cal I}|\cdot|{\cal C}_{D}|)$ time when restrictions on $A_{\gamma}$ and $B_{\gamma}$ have prepared. Regarding part (C), each set $M_{\gamma}$ and $N_{\gamma,j}$ has size $O(|{\cal I}|)$, so constructing the data structure in \Cref{lemma:3DRangeCounting} takes $O(|{\cal I}|\log^{2}|{\cal I}|)$ time. Summing over each $\gamma\in{\cal C}_{D}$ and each level $j$ of ${\cal S}_{\gamma}$, the total construction time for (C) is $O(|{\cal I}|\cdot|{\cal C}_{D}|\log^{2}|{\cal I}|\log n)$.
\end{proof}

We can immediately compute $\delta_{1}(t,k,x,y)$ and $\delta_{2}(t,k,x,y)$ using the additional data structures. The subproblem $\delta_{1}(t,k,x,y)$ can be computed in $O(1)$ time using part (A). Recall that intervals in $Z_{t,k}$ are located consecutively in ${\cal I}_{t}$, so $Z_{t,k}$ can be written as ${\cal I}_{t}(k_{1},k_{2})$. 
Then \[
\delta_{1}(t,k,x,y)=\sum_{i=k_{1}}^{k_{2}}\sum_{i'=x}^{y}|E\cap ({\cal I}_{t}(i)\times{\cal I}_{t}(i'))|+\sum_{i=k_{1}}^{k_{2}}\sum_{i'=x}^{y}|E\cap ({\cal I}_{t}(i')\times{\cal I}_{t}(i))|\]
is exactly two 2D-range sums on (A). 

We now computes $\delta_{2}(t,k,x,y)$ using part (B). We can rewrite $\delta_{2}(t,k,x,y)$ in the following form.
\[
\delta_{2}(t,k,x,y)=\sum_{\gamma\in{\cal C}_{D}}\left[\left(\sum_{i=k_{1}}^{k_{2}}|A_{\gamma}\cap {\cal I}_{t}(i)|\right)\cdot\left(\sum_{i=x}^{y}|B_{\gamma}\cap {\cal I}_{t}(i)|\right)+\left(\sum_{i=x}^{y}|A_{\gamma}\cap {\cal I}_{t}(i)|\right)\cdot\left(\sum_{i=k_x}^{k_y}|B_{\gamma}\cap {\cal I}_{t}(i)|\right)\right].
\]
By enumerating $\gamma\in {\cal C}_{D}$, the calculation only requires totally $O(|{\cal C}_{D}|)$ queries on parts (A) and (B).

\subsubsection{Computing $\delta_{3}(t,k,x,y)$}
\label{sect:ComputingDelta3}

We first rearrange the expression of $\delta_{3}(t,k,l,r)$ as follows.
\[
\delta_{3}(t,k,l,r)=\sum_{\gamma\in{\cal C}_{D}}\sum_{j=1}^{\lceil\log|A_{\gamma}|\rceil+1}\sum_{I\in Z_{t,k}}Q_{1}(t,\gamma,j,I,x,y)+Q_{2}(t,\gamma,j,I,x,y),
\]
where
\begin{align*}
Q_{1}(t,\gamma,j,I,x,y)&=\left|\left(\sum_{S\in{\cal S}_{\gamma,j}}S\times f(S)\right)\cap \left(I\times \bigcup_{I'\in {\cal I}_{t}(x,y)}I'\right)\right|,\\
Q_{2}(t,\gamma,j,I,x,y)&=\left|\left(\sum_{S\in{\cal S}_{\gamma,j}}S\times f(S)\right)\cap \left(\bigcup_{I'\in {\cal I}_{t}(x,y)}I'\times I\right)\right|.
\end{align*}

The algorithm we show below can compute $Q_{1}(t,\gamma,j,I,x,y)$ and $Q_{2}(t,\gamma,j,I,x,y)$ in nearly constant time. The key property for the efficient computation is the monotonically increasing property of the witness function $f$ over ${\cal S}_{\gamma,j}$ (\Cref{lemma:MonotonicallyIncreasing}). Take the function $Q_{1}$ as an example. Roughly speaking, segments in ${\cal S}_{\gamma,j}$ that intersects $I$ have their witnesses locating in a consecutive range on $C_{\gamma,j}$ by \Cref{lemma:MonotonicallyIncreasing}. In the other direction, ignoring the boundary case, segments with witnesses in this range will totally included by $I$. Therefore, if we view $I\times \bigcup_{I'\in {\cal I}_{t}(x,y)}I'$ as ``conditions'' of the counting problem $Q_{1}$, we can reduce the condition $I$ on the first dimension into a new condition on the second dimension, i.e. the range on $C_{\gamma,j}$. The task basically becomes a counting problem on the intersection of the range on ${\cal C}_{\gamma,j}$ and intervals in ${\cal I}_{t}(x,y)$, which can be answered by part (D) of the additional data structures.

\begin{figure}
    \centering
    \includegraphics[scale=0.5]{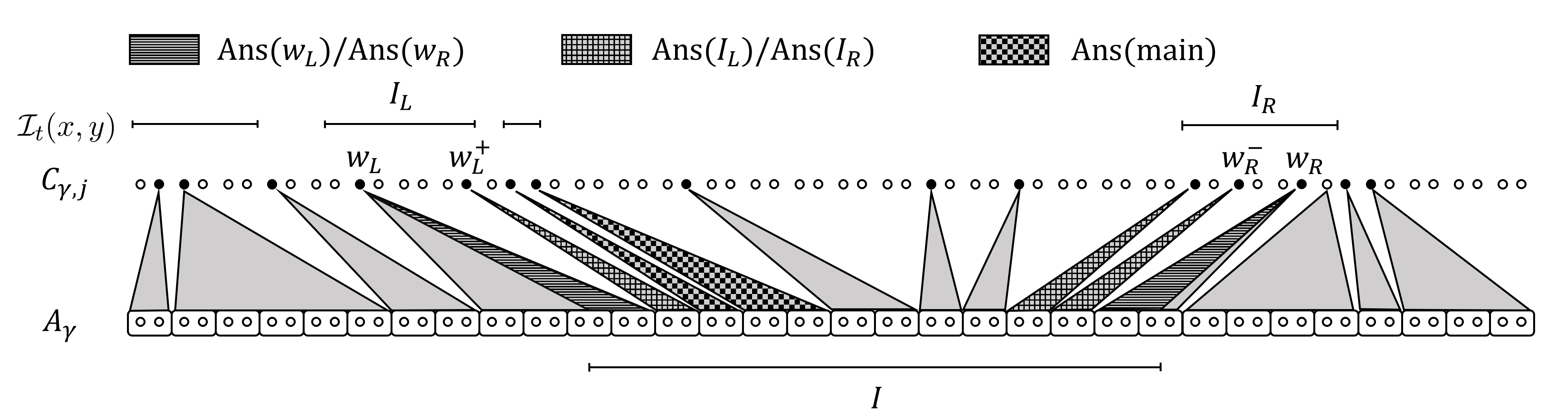}
    \caption{An illustration for computing $Q_{1}(t,\gamma,j,I,x,y)$. The list $A_{\gamma}$ (vertices below) are partitioned by segments in ${\cal S}_{\gamma,j}$. The witness list $C_{\gamma,j}$ (black vertices above) is a sublist of $A_{\gamma}$. The intervals in this figure have been restricted on $A_{\gamma}$.}
    \label{fig:delta3Q1}
\end{figure}

We now formally describe the algorithm for computing $Q_{1}(t,\gamma,j,I,x,y)$. See \Cref{fig:delta3Q1} for an illustration. 
\begin{itemize}
    \item First, we take the restriction of $I$ on $A_{\gamma}$, and let $v_{L}$ and $v_{R}$ denote the leftmost vertex and the rightmost vertex in $I\cap A_{\gamma}$. If $I\cap A_{\gamma}$ is empty, the algorithm terminates with zero as the output. Let $S_{L}$ (resp. $S_{R}$) be the unique segment in ${\cal S}_{\gamma,j}$ containing $v_{L}$ (resp. $v_{R}$). We further denote their witnesses $f(S_{L})$ and $f(S_{R})$ by $w_{L}$ and $w_{R}$ respectively.
    
    \item In this step we will count the part contributed by the witnesses $w_{L}$ and $w_{R}$. Precisely, the part $w_{L}$ contributes to the answer is
    \[
        \Ans(w_{L})=\left|\left(f^{-1}(w_{L},{\cal S}_{\gamma,j})\times w_{L}\right)\cap \left(I\times \bigcup_{I'\in {\cal I}_{t}(x,y)}I'\right)\right|,
    \]
    where $f^{-1}(w_{L},{\cal S}_{\gamma,j})$ denotes the union of segments in ${\cal S}_{\gamma,j}$ with $w_{L}$ as their witness. We first check if $w_{L}$ belongs to an interval in ${\cal I}_{t}(x,y)$. If so, $\Ans(w_{L})=|f^{-1}(w_{L},{\cal S}_{\gamma,j})\cap I|$, otherwise $\Ans(w_{L})=0$. In the former case, we first take the restriction $I\cap A_{\gamma}$ and then let $\Ans(w_{L})$ be the size of intersection of $I\cap A_{\gamma}$ and $f^{-1}(w_{L},{\cal S}_{\gamma,j})$ (note that vertices in $f^{-1}(w_{L},{\cal S}_{\gamma,j})$ are located consecutively on $A_{\gamma}$ by \Cref{lemma:MonotonicallyIncreasing}). If $w_{L}$ and $w_{R}$ are two distinct witnesses, we compute $\Ans(w_{R})=\left|\left(f^{-1}(w_{R},{\cal S}_{\gamma,j})\times w_{R}\right)\cap \left(I\times \bigcup_{I'\in {\cal I}_{t}(x,y)}I'\right)\right|$ similarly, otherwise the algorithm terminates with the answer $\Ans(w_{L})$.
    
    \item In the following steps we will compute the part contributed by witnesses between $w_{L}$ and $w_{R}$ on $C_{\gamma,j}$ (not including $w_{L},w_{R}$). Note that witnesses located before $w_{L}$ or after $w_{R}$ will have no contribution by \Cref{lemma:MonotonicallyIncreasing}. We let $w_{L}^{+}$ be the successor of $w_{L}$ and let $w_{R}^{-}$ be the predecessor of $w_{R}$ on $C_{\gamma,j}$. We denote witnesses from $w_{L}^{+}$ to $w_{R}^{-}$ on $C_{\gamma,j}$ by a sublist $C_{\gamma,j}(w_{L}^{+},w_{R}^{-})$. If the list $C_{\gamma,j}(w_{L}^{+},w_{R}^{-})$ is empty, i.e. $w_{R}^{-}$ is before $w_{L}^{+}$ on $C_{\gamma,j}$, the algorithm terminates with the current answer $\Ans(w_{L})+\Ans(w_{R})$. Otherwise, we compute
    \[
    \Ans(w_{L}^{+},w_{R}^{-})=\left|\sum_{w\in C_{\gamma,j}(w_{L}^{+},w_{R}^{-})}\left(f^{-1}(w,{\cal S}_{\gamma,j})\times w\right)\cap\left(I\times\bigcup_{I'\in{\cal I}_t(x,y)}I'\right)\right|.
    \]
    By \Cref{lemma:MonotonicallyIncreasing}, for each $w\in C_{\gamma,j}(w_{L}^{+},w_{R}^{-})$, we have $f^{-1}(w,{\cal S}_{\gamma,j})\subseteq I$. Therefore, we can rewrite $\Ans(w_{L}^{+},w_{R}^{-})$ as follows.
    \[
    \Ans(w_{L}^{+},w_{R}^{-})=\sum_{w\in W}|f^{-1}(w,{\cal S}_{\gamma,j})|,\text{ where }W=C_{\gamma,j}(w_{L}^{+},w_{R}^{-})\cap \bigcup_{I'\in{\cal I}_t(x,y)}I'.
    \]
    Recall that for each $w$ on the list $C_{\gamma,j}$, $|f^{-1}(w,{\cal S}_{\gamma,j})|$ is exactly the $C_{\gamma,j}$-weight of $w$.
    
    \item Because $C_{\gamma,j}(w_{L}^{+},w_{R}^{-})$ is a consecutive sublist on $C_{\gamma,j}$, there will be at most two \textit{special} intervals $I'\in{\cal I}_{t}(x,y)$ such that the restriction $I'\cap C_{\gamma,j}$ crosses the boundary of $C_{\gamma,j}(w_{L}^{+},w_{R}^{-})$ (namely $(I'\cap C_{\gamma,j})\cap C_{\gamma,j}(w_{L}^{+},w_{R}^{-})\neq\emptyset$ and $(I'\cap C_{\gamma,j})\setminus C_{\gamma,j}(w_{L}^{+},w_{R}^{-})\neq\emptyset$). In this step we compute the contribution of the special intervals to $\Ans(w_{L}^{+},w_{R}^{-})$.
    
    Observe that each special interval must contain either $w_{L}^{+}$ or $w_{R}^{-}$, so only $I_{L}$ and $I_{R}$ can be special, where $I_{L}$ (resp. $I_{R}$) is the interval in ${\cal I}$ such that $w_{L}^{+}\in I_{L}$ (resp. $w_{R}^{-}\in I_{R}$). W.l.o.g. we assume that $I_{L}$ and $I_{R}$ exist. $I_{L}$ is a special interval if and only if $I_{L}\in {\cal I}_{t}(x,y)$ and the leftmost vertex in $I_{L}\cap C_{\gamma,j}$ is before $w_{L}^{+}$, which can be verified easily. If $I_{L}$ is indeed special, its contribution is
    \[
        \Ans(I_{L})=\sum_{w\in C_{\gamma,j}(w_{L}^{+},w_{R}^{-})\cap I_{L}} |f^{-1}(w,{\cal S}_{\gamma,j})|,
    \]
    which can be answered by a range sum query on $C_{\gamma,j}$. Similarly, if $I_{R}$ is not the same interval as $I_{L}$ and $I_{R}$ is special, its contribution is $\Ans(I_{R})=\sum_{w\in C_{\gamma,j}(w_{L}^{+},w_{R}^{-})\cap I_{R}}|f^{-1}(w,{\cal S}_{\gamma,j})|$.
    \item The remaining part of $\Ans(w_{L}^{+},w_{R}^{-})$ is contributed by the non-special intervals in ${\cal I}_{t}(x,y)$. For each of the non-special interval $I$, the restriction $I\cap C_{\gamma,j}$ is either disjoint with $C_{\gamma,j}(w_{L}^{+},w_{R}^{-})$ or totally inside $C_{\gamma,j}(w_{L}^{+},w_{R}^{-})$. Therefore, their contribution is
    \[
    \Ans(\text{main})=\sum_{I\in{\cal I}_{t}(x,y),I\cap C_{\gamma,j}\subseteq C_{\gamma,j}(w_{L}^{+},w_{R}^{-})}\sum_{w\in I\cap C_{\gamma,j}} |f^{-1}(w,{\cal S}_{\gamma,j})|,
    \]
    which is exactly the return of querying $\threeDRangeSum(x,y,w_{L}^{+},w_{R}^{-})$ on the additional data structures of $N_{\gamma,j}$ in part (C). Note that here we use the same notation $w_{L}^{+}$ and $w_{R}^{-}$ to denote the $C_{\gamma,j}$-ranks of $w_{L}^{+}$ and $w_{R}^{-}$ respectively.
    \item Finally, we return
    \begin{align*}
    Q_{1}(t,\gamma,j,I,x,y)&=\Ans(w_{L})+\Ans(w_{R})+\Ans(w_{L}^{+},w_{R}^{-})\\
    &=\Ans(w_{L})+\Ans(w_{R})+\Ans(I_{L})+\Ans(I_{R})+\Ans(\text{main}).
    \end{align*}
\end{itemize}

The algorithm for $Q_{2}(t,\gamma,j,I,x,y)$ is similar and we sketch it as follows.
\begin{itemize}
    \item Because $I\cap C_{\gamma,j}$ is consecutive on $C_{\gamma,j}$, by \Cref{lemma:MonotonicallyIncreasing}, the segments with witnesses in $I\cap C_{\gamma,j}$ are located consecutive on ${\cal S}_{\gamma,j}$, which implies their union is a consecutive sublist of $A_{\gamma}$. We denote this sublist by $A_{\gamma}(v_{L},v_{R})$, where $v_{L}$ and $v_{R}$ are the leftmost vertex and the rightmost vertex of this sublist. Therefore, $Q_{2}$ can be written as
    $Q_{2}=\left|A_{\gamma}(v_{L},v_{R})\cap \bigcup_{I'\in{\cal I}_t(x,y)}I'\right|$.
    \item Let $I_{L}\in {\cal I}$ (resp. $I_{R}\in{\cal I}$) be the interval containing $v_{L}$ (resp. $v_{R}$). If $I_{L}$ exists, $I_{L}$ is in ${\cal I}_{t}(x,y)$ and $I_{L}\cap A_{\gamma}$ is not totally included by $A_{\gamma}(v_{L},v_{R})$, we compute the contribution of $I_{L}$, which is $\Ans(I_{L})=|A_{\gamma}(v_{L},v_{R})\cap I_{L}|$. Similarly, if $I_{R}$ exists, $I_{R}$ is not the same as $I_{L}$, $I_{R}$ is in ${\cal I}_{t}(x,y)$ and $I_{R}\cap A_{\gamma}$ is not a subset of $A_{\gamma}(v_{L},v_{R})$, we compute $\Ans(I_{R})=|A_{\gamma}(v_{L},v_{R})\cap I_{R}|$.
    \item The remaining intervals in ${\cal I}_{t}(x,y)$ will have their restrictions on $A_{\gamma}$ either disjoint with $A_{\gamma}(v_{L},v_{R})$ or totally inside $A_{\gamma}(v_{L},v_{R})$. Their contribution $\Ans(\text{main})$ can be answered by the query $\threeDRangeSum(x,y,v_{L},v_{R})$ on the structure of $M_{\gamma}$ in part (C).
    \item The return is $Q_{2}(t,\gamma,j,I,x,y)=\Ans(I_{L})+\Ans(I_{R})+\Ans(\text{main})$.
\end{itemize}

\begin{proof}[Proof of \Cref{lemma:IntervalConnectivity}]

The correctness of the algorithm directly follows the description. We will focus on analysing the running time.

As mentioned, the \Boruvka's  algorithm has $\log|{\cal I}|$ phases. For each phase $t$, we invoke the function $\FindNeighbors(t,Z_{t,k})$ for each $Z_{t,k}\in {\cal Z}_{t}$. Solving each $\FindNeighbors(t,Z_{t,k})$ via binary search needs to query the indicator $\BatchConnCheck(t,k,l,r)$ with fix parameters $t$ and $k$ at most $O(\log|{\cal Z}_{k}|)$ times. As shown below, with additional structures in \Cref{sect:AdditionalStructure}, solving $\BatchConnCheck(t,k,l,r)$ takes $O(|{\cal C}_{D}|\cdot|Z_{t,k}|\log^{2}|{\cal I}|\log n)$ time. Therefore, the total running time is 
\[
O(|{\cal I}|^{2}\log n+|{\cal I}|\cdot|{\cal C}_{D}|\log^{3}|{\cal I}|\log n)+\sum_{t\leq \log|{\cal I}|}\sum_{Z_{t,k}\in{\cal Z}_{t}} O(\log|{\cal Z}_{k}|)\cdot O(|{\cal C}_{D}|\cdot|Z_{t,k}|\log^{2}|{\cal I}|\log n),
\]
where the first term is the construction time of the additional structures from \Cref{lemma:AddStructConstruct}.
It is exactly $O(|{\cal I}|^{2}\log n+|{\cal I}|\cdot|{\cal C}_{D}|\log^{4}|{\cal I}|\log n)$ 
because $\sum_{Z_{t,k}\in{\cal Z}_{t}}|Z_{t,k}|\leq |{\cal I}_{t}|\leq|{\cal I}|$ for any $t$.

It remains to show that computing $\BatchConnCheck(t,k,l,r)$ takes $O(|{\cal C}_{D}|\cdot|Z_{t,k}|\log^{2}|{\cal I}|\log n)$ time. We analyse the three subproblems one by one. Computing $\delta_{1}(t,k,x,y)$ requires only two 2D-range sum query on part (A), which takes $O(1)$ time. Computing $\delta_{2}(t,k,x,y)$ needs $O(|{\cal C}_{D}|)$ queries on part (B) each of which takes constant time, so the total time is $O(|{\cal C}_{D}|)$. Computing $\delta_{3}(t,k,x,y)$ takes $O(|{\cal C}_{D}|\cdot|Z_{t,k}|\log^{2}|{\cal I}|\log n)$ because it can be decomposed into $O(|{\cal C}_{D}|\cdot|Z_{t,k}|\log n)$ computations of $Q_{1}(t,\gamma,j,I,x,y)$ and $Q_{2}(t,\gamma,j,I,x,y)$. Each computation of $Q_{1}$ or $Q_{2}$ takes $O(\log^{2}|{\cal I}|)$ time because the bottleneck is querying part (C).

\end{proof}

\subsection{The Query Algorithm}
\label{sect:QueryAlgo}

The query algorithm is identical to that in Section 5.2 of \cite{DuanP20}, since we have implemented the required interfaces. We describe the algorithm as follows for completeness. Let $u,v\notin D$ be the given vertices. The goal of the query algorithm is to find two affected subtrees connected to $u$ and $v$ in $G\setminus D$, so the connectivity of $u$ and $v$ is equivalent to the connectivity of these two subtrees, which can be answered by checking if they belong to the same group of $\hat{\cal R}$.

First we find the component $\gamma(u)$ with $u$ in its terminal set $U(\gamma(u))$. If $\gamma(u)$ is an unaffected component, let $\hat{\gamma}(u)$ be the most ancestral unaffected component of $\gamma(u)$, which can be identified in $O(\log d)$ time as follows. Fix some preorder of components in ${\cal C}$ (recall that they form a tree structure). Let $\gamma_{-}(u)$ be the predecessor of $\gamma(u)$ among components $\gamma$ with $U(\gamma)\cap D\neq \emptyset$
(note that it is not the same with the definition of affected components). Then $\hat{\gamma}(u)$ is exactly the $\LCA(\gamma(u),\gamma_{-}(u))$'s child such that it is an ancestor of $\gamma(u)$. Specially, if $\gamma_{-}(u)$ does not exist, then $\hat{\gamma}(u)$ is the root. Analogously, we find $\gamma(v)$ and $\hat{\gamma}(v)$. If $\hat{\gamma}(u)=\hat{\gamma}(v)$, we claim $u$ and $v$ are connected. The bottleneck is the predecessor search for $\gamma_{-}(u)$, which takes $O(\log d)$ time.

If $\gamma(u)$ itself is affected, the affected tree $\tau(\gamma(u))$ has an affected subtree $\hat{\tau}$ with $u$ as a terminal, which can be found in $O(\log d)$ time by \Cref{lemma:TreeToInterval}. Otherwise, we scan the list $A_{\hat{\gamma}(u)}$ until we find a vertex $u'\notin D$, which takes $O(d)$ time. If there is no such $u'$, we claim that $u$ and $v$ is not connected since $u\in V(\hat{\gamma}(u))$, $v\notin V(\hat{\gamma}(u))$ and $\hat{\gamma}(u)$ has no edge connecting outside after failures. If we find such $u'$, then $u$ and $u'$ are connected in $G\setminus D$ and $u'$ is a terminal of some affected component, which implies we can find an affected subtree containing it in $O(\log d)$ time by \Cref{lemma:TreeToInterval}. We do the same thing for $v$.

Obviously, the query time is $O(d)$ and the correctness follows \Cref{thm:UpdateAlgorithm} and \Cref{lemma:ConnEqOrigianlH}.

\section{Conditional Lower Bounds}
\label{sect:LowerBound}

In this section, we will show conditional lower bounds on major complexity measures of the connectivity oracle with batched vertex failures, including update time and construction time. The lower bounds on query time and space have been shown in \cite{HKNS15,DuanP20}. Combining all the lower bounds, our oracle with construction time $O(m)+\tilde{O}(\mbar d_{\star})+\mbar^{1+o(1)}$, space $\tilde{O}(\mbar)$, update time $d^{2}n^{o(1)}$ and query time $O(d)$, where $\mbar=\min\{m,nd_{\star}\}$, is nearly optimal up to $n^{o(1)}$ factors. Namely, improvement on any measure without extra costs on other measures needs significant breakthrough.

In this section, we use the same notation as in \cite{HKNS15}.

\begin{definition}[$\littleo$ Notation] For any parameters $n_{1},n_{2},n_{3}$, we say that a function $f(n_{1},n_{2},n_{3})=\littleo(n_{1}^{c_{1}}n_{2}^{c_{2}}n_{3}^{c_{3}})$ iff there exists some constant $\epsilon>0$ such that $f=O(n_{1}^{c_{1}-\epsilon}n_{2}^{c_{2}}n_{3}^{c_{3}} + n_{1}^{c_{1}}n_{2}^{c_{2}-\epsilon}n_{3}^{c_{3}} + n_{1}^{c_{1}}n_{2}^{c_{2}}n_{3}^{c_{3}-\epsilon})$. We use the analogous definition for functions with one or two parameters.

\end{definition}

\paragraph{Lower Bound for Construction Time.} We will show a lower bound on construction time conditioning on \Cref{conj:BMM}, by reducing Boolean matrix multiplication (BMM) problems to $d$-vertex-failure connectivity problems.

\begin{conjecture}[No truly subcubic combinatorial BMM, \cite{AW14}]
In the Word RAM model with words of $O(\log n)$ bits, any combinatorial algorithm requires $n^{3-o(1)}$ time in expectation to compute the Boolean product of two $n\times n$ matrices.
\label{conj:BMM}
\end{conjecture}

\begin{theorem}
For any constant $\delta\in(0,1]$, unless \Cref{conj:BMM} fails, there is no combinatorial algorithm for $d$-vertex-failure connectivity for $n$-vertices $m$-edge graphs, where $m=\Theta(n^{1+\delta})$ and $d=\Theta(n^{\delta})$, with construction time $c(n,d)=\littleo(nd^{2})$, update time $u(d)=\littleo(d^{1/\delta})$ and query time $q(d)=\littleo(d^{1/\delta})$ in expectation.
\label{thm:ConstructionLowerBound}
\end{theorem}
\begin{proof}
Suppose there exists such an efficient algorithm ${\cal A}$ of $d$-vertex-failure connectivity for some $\delta\in(0,1]$. We first show that, for any $\hat{n}$ and $\hat{d}=\Theta(\hat{n}^{\delta})$, we can compute the Boolean product $C$ of arbitrary $\hat{d}\times \hat{n}$ matrix $A'$ and $\hat{n}\times \hat{d}$ matrix $B'$ in $\littleo(\hat{n}\hat{d}^{2})$ time. We construct a graph $G$ with $V(G)=\{x_{i}\mid 1\leq i\leq \hat{d}\}\cup\{y_{j}\mid 1\leq j\leq \hat{n}\}\cup \{z_{k}\mid 1\leq k\leq \hat{d}\}$ and $E(G)=\{(x_{i},y_{j})\mid A'_{ij}=1\}\cup\{(y_{j},z_{k})\mid B'_{jk}=1\}$. Then we construct an $d$-vertex-failure connectivity oracle of $G$ using ${\cal A}$ with parameters $n=\hat{n}+\hat{2d}$ and $d=2\hat{d}=\Theta(n^{\delta})$ in $c(n,d)$ time. For each $1\leq i,k\leq \hat{d}$, $C'_{ik}=\bigvee_{j} A'_{ij}\wedge B'_{jk}$ is exactly the connectivity of vertices $x_{i}$ and $z_{k}$ when the failure set $D$ is set to be $\{x_{i'}\mid i'\neq i\}\cup\{z_{k'}\mid k'\neq k\}$, which can be computed in $u(d)+q(d)$ time. We conclude that $C'=A'B'$ can be computed in $c(n,d)+d^{2}(u(d)+q(d))=\littleo(\hat{n}\hat{d}^{2})$ time using ${\cal A}$. 

Next, to solve BMM for arbitrary $\hat{n}\times \hat{n}$ matrices $A$ and $B$. We can divide $A$ into $\hat{d}\times \hat{n}$ matrices $A^{1},...,A^{\hat{n}/\hat{d}}$ and divide $B$ into $\hat{n}\times \hat{d}$ matrices $B^{1},...,B^{\hat{n}/\hat{d}}$. The Boolean product $C=AB$ can be computed by querying for each $1\leq x,y\leq \hat{n}/\hat{d}$, the product of $A^{x}$ and $B^{y}$, which means ${\cal A}$ can lead to an $\littleo(\hat{n}^{3})$ algorithm for BMM and \Cref{conj:BMM} fails.
\end{proof}

For simplicity, we let $d_{\star}=d$. Assuming \Cref{conj:BMM}, \Cref{thm:ConstructionLowerBound} shows that our construction time $O(m)+\tilde{O}(\mbar d)+\mbar^{1+o(1)}=O(m)+n^{1+o(1)}d^{2}$ is somewhat tight because we cannot improve the construction time by some factors polynomial in $n$ and $d$, as long as we want combinatorial algorithm with update time and query time polynomial in $d$ up to $n^{o(1)}$ factors.

\paragraph{Lower Bound for Update Time.} The lower bound of the update time is shown by exploiting the relationship between online Boolean vector-matrix-vector
multiplication ($\OuMv$) problems and $d$-vertex-failure connectivity problems. Formally speaking, in a $\OuMv$ instance with parameters $n_{1},n_{2},n_{3}$, an $n_{1}\times n_{2}$ matrix $M$ is given in advanced and it can be preprocessed. What follows is $n_{3}$ online queries, where the $k$-th query will give two vectors $u^{k}$ and $v^{k}$ as input and the required output is $(u^{k})^{T}Mv^{k}$. As shown in \cite{HKNS15}, a lower bound on the $\OuMv$ problem can be obtained conditioning on the well-known $\OMv$ conjecture (\Cref{conj:omv}). %

\begin{conjecture}[$\OMv$ Conjecture, \cite{HKNS15}]
For any constant $\epsilon >0$, there is no $O(n^{3-\epsilon})$-time algorithm that solves $\OMv$ with an error probability of at most $1/3$.
\label{conj:omv}
\end{conjecture}

\begin{lemma}[Conditional lower bounds on $\OuMv$, \cite{HKNS15}]
\Cref{conj:omv} implies that there is no algorithm for $\OuMv$ with parameters $n_{1},n_{2},n_{3}$ using preprocessing time $p(n_{1},n_{2})=poly(n_{1},n_{2})$ and computation time $c(n_{1},n_{2},n_{3})=\littleo(n_{1}n_{2}n_{3})$ that has an error probability of at most $1/3$.
\end{lemma}

\begin{theorem}
For any constant $\delta\in(0,1]$, unless \Cref{conj:omv} fails, there is no algorithm for $d$-vertex-failure connectivity for $n$-vertices $m$-edge graphs, where $d=\Theta(n^{\delta})$, with construction time $c(n,d)=poly(n,d)$, update time $u(d)=\littleo(d^{2})$ and query time $q(d)=\littleo(d^{2})$ that has an error probability $1/3$.
\label{thm:UpdateLowerBound}
\end{theorem}
\begin{proof}
Suppose that such an algorithm ${\cal A}$ for $d$-vertex-failure connectivity exists for $\delta=1$. Then we can design an $\littleo(\hat{n}^{3})$-time algorithm for $\OuMv$ with parameters $n_{1},n_{2},n_{3}=\hat{n}$. First we construct the graph $G$ by making $V(G)=\{x_{0},y_{0}\}\cup\{x_{i}\mid 1\leq i\leq n_{1}\}\cup\{y_{j}\mid 1\leq j\leq n_{2}\}$ and $E(G)=\{(x_{0},x_{i})\mid i\}\cup\{(y_{j},y_{0})\mid j\}\cup\{(x_{i},y_{j})\mid M_{ij}=1\}$, and we use ${\cal A}$ to construct a $d$-vertex-failure connectivity oracle on $G$ with parameters $n=|V(G)|=n_{1}+n_{2}+2$ and $d=n_{1}+n_{2}$ in $c(n,d)$ time. For each query $k$ ($1\leq k\leq n_{3}$) of the $\OuMv$ instance, let the failure set be $D=\{x_{i}\mid u^{k}_{i}=0\}\cup\{y_{j}\mid v^{k}_{j}=0\}$, so the answer of this query is exactly the connectivity of vertices $x_{0}$ and $y_{0}$ in $G$, which can be answered by ${\cal A}$ in $u(d)+p(d)$ time. Therefore, we conclude that any $\OuMv$ instance with $n_{1},n_{2},n_{3}=\hat{n}$ can be solved with preprocessing time $p(n_{1},n_{2})=c(n,d)=poly(\hat{n})$ and computation time $c(n_{1},n_{2},n_{3})=n_{3}(u(d)+q(d))=\littleo(\hat{n}^{3})$ by using ${\cal A}$, and \Cref{conj:omv} fails. To make this argument work for any $\delta\in(0,1]$, we can just add some redundant vertices to $G$ to ensure $d=\Theta(n^{\delta})$. 
\end{proof}

\paragraph{Lower Bound for the $d$-Vertex-Failure Global Connectivity Problem.} We also study a related problem of \emph{vertex-failure
global connectivity oracles.} This oracle is similar to a vertex-failure
connectivity oracle, except that in the second phase, after we are
given $D\subset V$ in the update phase, we only need to return if $G\setminus D$ is
connected and there is no individual query phase. We show that there is no oracle with reasonable update time for this problem by considering the (unbalanced) orthogonal vectors (OV) problem. Given parameters $d,n,m$, an instance of the OV problem is two set $A,B$ of vectors in $\{0,1\}^{d}$ such that $|A|=n$ and $|B|=m$, and the required answer is to decide whether there exists vectors $a^{i}\in A$ and $b^{j}\in B$ such that they are orthogonal.

\begin{conjecture}[Strong Exponential Time Hypothesis (SETH), \cite{AW14}]
For every $\epsilon > 0$, there exists a $k$ such that SAT on $k$-CNF formulas on $n$ variables cannot be solved in $\Otil(2^{(1-\epsilon)n})$
\label{conj:SETH}
\end{conjecture}

\begin{lemma}[Unbalanced Orthogonal Vectors Hypothesis (UOVH), \cite{BK15}]
For any constant $t\in(0,1]$, unless \Cref{conj:SETH} fails, there is no an algorithm for OV, restricted to $m=\Theta(n^{t})$ and $d= n^{o(1)}$, that runs in time $\littleo(nm)$.
\end{lemma}

\begin{theorem}
For any constant $t>0$, unless \Cref{conj:SETH} fails, there is no algorithm for $d$-vertex-failure global connectivity on $n$-vertex $m$-edge graphs even with $d= n^{o(1)}$, with construction time $c(n)=\littleo(n^{1+t})$ and query time $q(n)=\littleo(n)$.
\end{theorem}
\begin{proof}
Assume there exists such an algorithm ${\cal A}$ for $d$-vertex-failure global connectivity for some constant $t>0$. We show that it will imply an algorithm for OV restricted to $\hat{m}=\Theta(\hat{n}^{t})$ and $\hat{d}=n^{o(1)}$. We consider a graph $G$ with $V(G)=\{y_{i}\mid 1\leq i\leq \hat{n}\}\cup\{x_{k}\mid 1\leq k\leq\hat{d}\}\cup\{u\}$ and $E(G)=\{(y_{i},x_{k})\mid a^{i}_{k}=1\}\cup\{(x_{k},u)\mid k\}$, where $a^{i}_{k}$ denote the $k$-th coordinate of vector $a^{i}\in A$. We construct a $d$-vertex-failure global connectivity oracle on $G$ using ${\cal A}$ with parameters $n=\hat{n}+\hat{d}+1$ and $d=\hat{d}$ in time $c(n)=\littleo(n^{1+t})$. To get the answer for an OV instance, we process $b^{j}\in B$ one by one, and for each $b^{j}$, it only requires one query on the oracle to check if there exists $a^{i}$ orthogonal to $b^{j}$. Concretely, by setting the failure set $D=\{x_{k}\mid b^{j}_{k}=0\}$, then $a^{i}$ is orthogonal to $b^{j}$ if and only if the vertex $y_{i}$ is not connected with $u$ after failures. Therefore, the OV instance can be answered in $c(n)+\hat{m}q(n)=\littleo(\hat{n}\hat{m})$ time, which means \Cref{conj:SETH} fails.
\end{proof}

%% file: 9-open.tex
\section{Open Problems}
\label{sect:OpenProblems}

While \Cref{thm:main,thm:main lower} together have settled
the complexity of the connectivity oracles with vertex failures up to sub-polynomial factors, they still leave interesting
and questions for this problem. 
\begin{enumerate}
\item If the update and query time can depend only on $d_{\star}$, what
is the best bound we we can hope for? van den Brand and Saranurak
\cite{van2019sensitive} showed a randomized Monte-Carlo oracle with
$O(n^{2})$ space, $O(n^{\omega})$ preprocessing, $O(d^{\omega})$
update, and $O(d^{2})$ query time. Note that both update and query
time are both constant for all $d=O(1)$. Very recently, Pilipczhuk
et al. \cite{pilipczuk2021algorithms} gave the first deterministic
solution for this setting though the bounds seem far from being optimal. 
\item Our $\Ohat(md_{\star})$ preprocessing time is tight only under the
\emph{combinatorial} BMM conjecture. Can we improve the bound using
fast matrix multiplication? 
\end{enumerate}
The following questions about generalization of this problem seem
challenging:
\begin{enumerate}
\item Suppose the vertices to be deleted are given one by one and queries
are given between each deletion, can we still handle all $d$ deletions
in total $\Ohat(d^{2})$ update time? 
\item Is there a vertex-failure $c$-(edge/vertex)-connectivity oracle for
$c>1$? The dynamic $c$-edge connectivity algorithm by Jin and Sun
\cite{jin2022fully} implies a near-optimal \emph{edge}-failure $c$-edge
connectivity oracle for $c=\log^{o(1)}n$, but all other variants
of this problems remain open.
\end{enumerate}

%% file: A-Appendix.tex
\section{Omitted Proof}
\label{Appendix:OmittedProof}
\subsection{Proof of \Cref{thm:ComplexDecomp} and \Cref{Coro:WeightedDecomposition}}
\label{proof:ComplexDecomp}

\begin{proof}[Proof of \Cref{thm:ComplexDecomp}]
We design a procedure $\Decomp(G)$, which takes the graph $G$ as inputs, and returns a collection of vertex-disjoint subgraphs of $G$. The procedure just simply applies Lemma \ref{lemma:CutOrCertify} on graph $G$ with parameters $\epsilon/\log(nU)$ and $r$. It will obtain a cut $(L,S,R)$ such that $w(S)\leq(\epsilon/\log(nU))\cdot\d(L\cup S)$. There are two cases.

\

\noindent\textbf{Case 1.} If $(L,S,R)$ is a balance cut with $\d(L\cup S),\d(R\cup S)\geq \d(V(G))/3$, we do recursion on both $G[L]$ and $G[R]$ and let $\Decomp(G)=\Decomp(G[L])\cup\Decomp(G[R])$.

\

\noindent\textbf{Case 2.} If $\d(R)\geq \d(V(G))/2$ and $h(G[R])\geq \phi$ for some $\phi=(\epsilon/\log(nU))/(\log^{O(r)}U\log^{O(r^{5})}(m\log U))=\epsilon/(\log^{O(r)}U\log^{O(r^{5})}(m\log U))$, then we only recurse on $G[L]$ and let $\Decomp(G)=\Decomp(G[L])\cup\{R\}$. In particular, if we further have $L=\emptyset$, then $\Decomp(G)=\{R\}$.

\

We are going to show ${\cal G}=\Decomp(G,T)$ is our desired output. The first two properties are easy to verified. Each recursion we drop vertices in $S$ to disconnect $L$ and $R$, so there is no edge connecting any two subsets in ${\cal G}$. Every set $V_{i}\in{\cal G}$ has been verified that $h(G[V_{i}])\geq \phi$ in the algorithm. For the last property, in order to bound the total weight of dropped vertices, in each recursion with vertex cut $(L,S,R)$, we charge $w(S)$ averagely on the demands of $L\cup S$, i.e. each unit of demands in $\d(L\cup S)$ will be charged $w(S)/\d(L\cup S)\leq \epsilon/\log(nU)$. Moreover, the depth of the recursion algorithm is at most $O(\log \d(V(G)))$, because every recursive call decrease shrink the total demand of the current graph by a constant factor. Therefore, each unit of demands will be charged at most $O(\log(nU))$ times and the total weight of dropped vertices is bounded by $\epsilon\cdot\d(V(G))$.

The overall running time is $O(m^{1+o(1)+O(1/r)}\cdot U^{1+O(1/r)}\cdot\log^{O(r^{4})}(mU)/\phi)$, because the depth of the recursion is $O(\log \d(V(G)))=O(\log(nU))$ and the total time of applying Lemma \ref{lemma:CutOrCertify} at all recursive step at the same depth is $O(m^{1+o(1)+O(1/r)}\cdot U^{1+O(1/r)}\cdot\log^{O(r^{4})}(mU)/\phi)$.
\end{proof}

\begin{proof}[Proof of \Cref{Coro:WeightedDecomposition}]
We can assume that the weights and demands are real numbers in $\{0\}\cup[1,U]$, because we can scale all weights and demands by multiplying $1/\delta_{w}$ and $1/\delta_{\d}$ respectively and let $\delta_{\d}\epsilon/\delta_{w}>0$ be the new parameter $\epsilon$, where $\delta_{w}=\min_{w(v)>0}w(v)$ and $\delta_{\d}=\min_{\d(v)>0}\d(v)$.

To ensure each vertex has weight at least its demand, we can scale the weights or the demands up depending on whether $\epsilon<1$, and then drop vertices with new demand larger than new weight. Concretely, if $\epsilon \leq 1$, for each $v$ we let its new weight be $w'(v)=w(v)/\epsilon$ and keep its demand $\d'(v)=\d(v)$ unchanged, otherwise we let its new demand be $\d'(v)=\epsilon\cdot\d(v)$ and keep its weight $w'(v)=w(v)$ unchanged. We let $\epsilon'=1$ be the new parameter. The new ratio bound of weights and demands is $U'=U\cdot\max\{1/\epsilon,\epsilon\}\leq nU^{2}$ because the result is trivial if $\epsilon<1/(nU)$ or $\epsilon > nU$. Because $\epsilon'=1$, it's safe to drop vertices $v$ with $w(v)/\d(v)\leq \epsilon'=1$ (namely we add them into the separator in advanced). Let $G'$ be the new graph and w.l.o.g. assume it is still connected. Lastly, we can make weights and demands integral by rounding up the weights and rounding down the demands.

By applying \Cref{thm:ComplexDecomp} on $G'=(V(G'),E(G'),w',\d')$ (let $n'=|V(G')|,m'=|E(G')|$) with parameters $\epsilon'=1$ and $r=\log\log\log(m'U')$, with high probability, we can obtain an $(\epsilon',\phi')$-vertex expander decomposition ${\cal G}'$ of $G'$ where $\phi'=\epsilon'/(m'^{o(1)}\log^{3+o(1)}U')$ in time $O(m'^{1+o(1)}\cdot U'^{o(1)})$. Observe that ${\cal G}={\cal G}'$ is also an $(\epsilon,\phi)$-vertex expander decomposition of $G$ where $\phi=\phi'\epsilon/(4\epsilon')=\epsilon/(4m'^{o(1)}U'^{o(1)})=\epsilon/(m^{o(1)}U^{o(1)})$, where the constant $1/4$ is from the rounding. The running time can be rewritten as $O(m^{1+o(1)}U^{1+o(1)}/\phi)$.
\end{proof}

\subsection{Proof of \Cref{thm:RandMatchingPlayer}}
\label{proof:DetMatchingPlayer}

We start by constructing the following vertex-capacitated graph $G'=(V(G'),E(G'),c)$. The vertex set $V(G')$ is $V(G)\cup\{u_{A}\mid u\in A\}\cup\{v_{B}\mid v\in B\}\cup\{s,t\}$, where $c(v)=\lceil 1/\phi\rceil\cdot w(v)$ for each $v\in V(G)$, $c(u_{A})=\d(u)$ for each $u_{A}$ corresponds to some $u\in A$, $c(v_{B})=\d(v)$ for each $v_{B}$ corresponds to some $v\in B$, and $s$ and $t$ are the source and the sink without capacity constraints. The edge set $E(G')$ is $E(G)\cup\{(s,u_{A}),(u_{A},u)\mid u\in A\}\cup\{(v,v_{B}),(v_{B},t)\mid v\in B\}$.

We will maintain an integral $\d$-matching $M$ between $A$ and $B$ during the algorithm, where $M$ is empty initially. The matching $M$ will be grew gradually by applying the approximate vertex-capacitated max flow algorithm in \Cref{lemma:DetVertexFlow} iteratively, until the value of $M$ reaches $\d(A)-z$ or a desired vertex cut is returned by the max flow algorithm in some round.

At the beginning of the $i$-th round, let $M_{i}$ be the current matching and let $\d_{i}$ be the demand function on $A\cup B$ representing the remaining demands, i.e. for each $v\in A\cup B$, $\d_{i}(v)=\d(v)-M_{i}(v)$, where $M_{i}(v)$ is the total weight of edges adjacent to $v$ in $M_{i}$. Let $G'_{i}=(V(G'),E(G'),c_{i})$ be the same graph with $G'$ except that $c_{i}(u_{A})=\d_{i}(u)$ for each $u\in A$ and $c_{i}(v_{B})=\d_{i}(v)$ for each $v\in B$. Then we apply \Cref{lemma:DetVertexFlow} with $\epsilon=1/2$ on $G'_{i}$. There are two cases.

\

\noindent\textbf{Case 1.} If it returns a cut $S'$ in $G'_{i}$ with total capacities $c_{i}(S')<\d_{i}(A)-z$, we will construct a desired vertex cut in $G$ and terminate the algorithm. Let $(L',S',R')$ be the vertex cut in $G'_{i}$ such that $s\in L'$ and $t\in R'$. We claim that the vertex cut $(L,S,R)$ in $G$ where $L=L'\cap V(G),S=S'\cap V(G)$ and $R=R'\cap V(G)$ is a desired output.
    
First we show that $\d(L),\d(R)> z$. Consider the cut $(L',S',R')$ of $G'_{i}$. For each $u\in A$, if $u\notin L'$, there must be $u_{A}\in S'$ or $u\in S'$, otherwise the path from $s$ through $u_{A}$ to $u$ will bring $u$ to $L'$. Therefore, each $u\in A\setminus L'$ will contribute at least $\min\{c_{i}(u_{A}),c_{i}(u)\}=\min\{d_{i}(u),\lceil 1/\phi\rceil\cdot w(u)\}\geq\d_{i}(u)$ to $c_{i}(S')$, which implies $\d_{i}(A\setminus L')\leq c_{i}(S')<\d_{i}(A)-z$ and $\d_{i}(A\cap L')> z$. Since $A\subseteq V(G)$, we finally get $\d(L)\geq \d_{i}(A\cap L') > z$. A similar argument shows $\d(R)> z$.
    
Next we will show $w(S)\leq \phi\cdot \d(L\cup S)$ and $w(S)\leq\phi\cdot \d(R\cup S)$. Observe that for any vertex $u\in A\cap R'$, we have $u_{A}\in S'$. Because $c(u_{A})=\d_{i}(u)$ is not counted in $c(S)$, it implies $c(S)\leq c(S')-\d_{i}(A\cap R')\leq \d_{i}(A)-z-\d_{i}(A\cap R')$. Since $\d_{i}(A\cap(L\cup S))=\d_{i}(A)-\d_{i}(A\cap R')$ and $c(S)=\lceil 1/\phi\rceil\cdot w(S)$, we get $\lceil 1/\phi\rceil\cdot w(S)\leq \d_{i}(A\cap(L\cup S))\leq \d(L\cup S)$. Similarly, we have $w(S)\leq\phi\cdot\d(R\cup S)$.

\

\noindent\textbf{Case 2.} Otherwise, we are guaranteed that the resulted flow $f$ has value at least $(\d_{i}(A)-z)/3$. By the flow rounding algorithm \cite{KP15}, in time $\tilde{O}(m)$, we can round $f$ into an integral flow with value at least $\lceil (\d_{i}(A)-z)/3 \rceil$. We can further decompose the flow into flow paths implicitly using dynamic tree \cite{ST83} in time $\tilde{O}(m)$\footnote{Here we only decompose the flow and obtain the endpoints of the flow paths, which is sufficient to get the matching.}, and get the corresponding $\d_{i}$-matching $M'$ between $A$ and $B$ with value at least $\lceil (\d_{i}(A)-z)/3 \rceil$ which can be embedded into $G$ with vertex congestion $\lceil 1/\phi\rceil$. Let this embedding be ${\cal P}_{i}$. Observe that $E({\cal P}_{i})$ is exactly edges $e\in E(G)$ with non-zero $f(e)$, which can be computed in $O(m)$ time. We then go to the next iteration by letting $M_{i+1}=M_{i}\cup M'$.

\

We claim that $O(\log (nU))$ rounds suffice. Observe that $M_{i}(A)+\d_{i}(A)=\d(A)$ and $M_{i+1}(A)\geq M_{i}(A)+ \lceil(\d_{i}(A)-z)/3\rceil$, which implies
\[
    \d(A)-M_{i+1}(A)-z\leq (\d(A)-M_{i}(A)-z) - \lceil (\d(A)-M_{i}(A)-z)/3 \rceil.
\]
Therefore, if algorithm does not terminate with a desired vertex cut in $k=O(\log(nU))$ rounds, we will have $\d(A)-M_{k+1}(A)-z\leq 0$ and the matching $M=M_{k+1}$ has $M(A)\geq\d(A)-z$ and it can be embedded into $G$ with vertex-congestion $k\cdot\lceil 1/\phi\rceil=O(\log(nU)/\phi)$. The embedding ${\cal P}$ of $M$ has $E({\cal P})=\bigcup_{i} E({\cal P}_{i})$.

The running time of the algorithm is $m^{1+o(1)}U\log U$ because there are $O(\log (nU))$ rounds and each round will invoke the algorithm in \Cref{lemma:DetVertexFlow} once, which takes $O(m^{1+o(1)}U/\phi)$ time (the maximum capacity ratio of $G'$ is $O(U/\phi)$).

\subsection{Proof of \Cref{lemma:MostBalanceCut} and \Cref{coro:EmbeddingOfExpander}}
\label{proof:MostBalanceCut}

We first prove \Cref{lemma:certificate} and then complete the proof of \Cref{lemma:MostBalanceCut} and \Cref{coro:EmbeddingOfExpander} using it.

\begin{lemma}
Let $G=(V(G),E(G),w,\d)$ be a graph with, for each $v\in V(G)$, integral weight $w(v)\geq 1$ and integral demand $0\leq \d(v)\leq w(v)$. Let $H=(V(H),E(H),w_{H},\d)$ be a graph with $V(H)=V(G)$, the same vertex demand function $\d$, and for each $e\in E(H)$, integral weight $w_{H}(e)\geq 1$. Furthermore, there exists fix parameters $z>0$ and $0<\phi\leq 1$ such that any cut $(A,B)$ of $H$ with $\d(A),\d(B)\geq z$ has $\Psi_{H}(A,B)\geq \phi$. Then, if there exists an embedding ${\cal P}$ of $H$ into $G$ with vertex-congestion $\eta\geq 1$, any vertex cut $(L,S,R)$ of $G$ with $\d(R\cup S), \d(L\cup S)\geq 2z$ has $h_{G}(L,S,R)\geq \phi/(2\eta)$.
\label{lemma:certificate}
\end{lemma}
\begin{proof}
Consider an arbitrary vertex cut $(L,S,R)$ of $G$ with $\d(R\cup S) \geq \d(L\cup S)\geq 2z$. By definition,
\[
h_{G}(L,S,R)=\frac{w(S)}{\d(L\cup S)}.
\]
In what follows, we assume $\d(S)\leq \d(L)$, otherwise we will have $h_{G,T}(L,S,R)=w(S)/\d(L\cup S)\geq \d(S)/\d(L\cup S)\geq 1/2\geq \phi/(2\eta)$ as required.

Let $A=R$ and $B=L\cup S$. Because $\d(R\cup S)\geq \d(L\cup S)\geq 2z$ and $\d(S)\leq \d(L)$, we have $\d(R)\geq \d(L)\geq z$ and $\d(R)\geq \d(L\cup S)/2$, which imply $\d(A),\d(B)\geq z$ and $\d(A)\geq \d(B)/2$. Therefore, $\Psi_{H}(A,B)\geq \phi$ and 
\[
    \Psi_{H}(A,B)=\frac{w_{H}(E_{H}(A,B))}{\min\{\d(A),\d(B)\}}\leq \frac{2w_{H}(E_{H}(A,B))}{\d(B)}=\frac{2w_{H}(E_{H}(A,B))}{\d(L\cup S)}.
\]

Because $H$ can be embedded into $G$ with vertex-congestion at most $\eta$ and each edge $e\in E_{H}(A,B)$ corresponds to exactly $w_{H}(e)$ unweighted flow paths, each of which contains at least one vertex in $S$, we have $w_{H}(E_{H}(A,B))\leq \eta\cdot w(S)$, which implies $h_{G,T}(L,S,R)\geq \phi/(2\eta)$.

\end{proof}

\begin{proof}[Proof of \Cref{lemma:MostBalanceCut}]
We use the cut-matching game framework in \Cref{sect:WeightedCutMatchingGame}, where the cut player uses the algorithm in \Cref{thm:DetCutPlayer} and the matching player uses the algorithm in \Cref{thm:RandMatchingPlayer}. Let $H=(V(H),E(H),w_{H},\d)$ be the graph the cut-matching game is working on, where $V(H)$ is equal to $V(G)$ with the same vertex demand function $\d$, and $E(H)$ is the edge set (initially empty) with edge weight function $w_{H}$. During the game, we also maintain a set $F\subseteq E(H)$ of \textit{fake edges} which collects edges not embedded in $G$.

In the $i$-th round, the cut player applies \Cref{thm:DetCutPlayer} on the current $H$ with parameter $r$, which returns a cut $(A_{i},B_{i})$ of $H$ with $\d(A_{i}),\d(B_{i})\geq \d(V(H))/4$ and $w_{H}(E_{H}(A_{i},B_{i}))\leq \d(V(H))/100$, or a large set $S_{i}\subseteq V(H)$ with $\d(S_{i})\geq \d(V(H))/2$ and $\Psi(H[S_{i}])\geq 1/\log^{O(r^{4})} m$. In the latter case, we let $A_{i}=V(H)\setminus S_{i}$ and $B_{i}=S_{i}$, and the cut-matching game will end right after this round.

The matching player applies \Cref{thm:RandMatchingPlayer} on graph $G$ with the edge cut $(A_{i},B_{i})$ and parameters $z$ and $\phi$. If the matching player returns a vertex cut, we can terminate the algorithm with a desired cut. Otherwise, we get a matching $M_{i}$ between $A_{i}$ and $B_{i}$ with value $M_{i}(A_{i})\geq |A_{i}|-z$, and $M_{i}$ can be embedded into $G$ with vertex-congestion $\eta'=O(\log (nU)/\phi)$. Observe that $M_{i}$ is not a perfect $\d$-matching between $A_{i}$ and $B_{i}$ so it does not satisfy the requirement of the game. We augment it by an arbitrary perfect $\d_{i}$-matching $F_{i}$ between $A_{i}$ and $B_{i}$, where $\d_{i}(v)=\d(v)-M_{i}(v)$ for each $v\in A_{i}\cup B_{i}$. The total weight of edges in $F_{i}$ is $\d(A_{i})-M_{i}(A_{i})\leq z$. The matching player adds $M_{i}\cup F_{i}$ into $H$ and the game proceeds to the next round.

From now we assume the game ends with successfully constructing an expander $H$. Let $\beta>0$ be some large enough constant. By \Cref{lemma:CMGRound}, the number of rounds is at most $k_{\CMG}=O(\log(nU))$. Observe that each $M_{i}$ will have at most $m$ edges, so $|E(H)|\leq k_{\CMG}\cdot m\leq O(m\log(nU))\leq \beta m\log(nU)$. By \Cref{lemma:CMGRound}, $\Psi(H)\geq 1/\log^{O(r^{4})}|E(H)|\geq 1/\log^{\beta r^{4}}(\beta m\log(nU))$ at last. Let $\hat{H}=H-F$, i.e. $\hat{H}=\bigcup_{i}M_{i}$, so $\hat{H}$ can be embedded into $G$ with vertex-congestion at most $\eta=k_{\CMG}\cdot\eta'=O(\log^{2}(nU)/\phi)\leq \beta\log^{2}(nU)/\phi$. The total weight of fake edges $w_{H}(F)\leq k_{\CMG}\cdot z=O(z\log(nU))\leq \beta z\log(nU)$.

\begin{claim}
Any cut $(A,B)$ of $\hat{H}$ with $\d(B)\geq \d(A)\geq 2w_{H}(F)/\Psi(H)$ has $\Psi_{\hat{H}}(A,B)\geq \Psi(H)/2$.
\label{claim:WeightedSparsityWithFake}
\end{claim}
\begin{proof}
Because $w_{H}(E_{H}(A,B))\geq \Psi(H)\cdot \d(A)$ and $\hat{H}=H-F$, we have $w_{H}(E_{\hat{H}}(A,B))\geq \Psi(H)\cdot \d(A)-w_{H}(F)$. When $\d(A)\geq 2w_{H}(F)/\Psi(H)$, we have $w_{H}(E_{\hat{H}}(A,B))\geq \Psi(H)\cdot\d(A)/2$ and $\Psi_{\hat{H}}(A,B)\geq \Psi(H)/2$.
\end{proof}

By \Cref{claim:WeightedSparsityWithFake}, any cut $(A,B)$ of $\hat{H}$ with $\d(B)\geq \d(A)\geq 2\beta z\log(nU)\log^{\beta r^{4}}(\beta m\log(nU))\geq 2w_{H}(F)/\Psi(H)$ has $\Psi_{\hat{H}}(A,B)\geq \Psi(H)/2\geq 1/(2\log^{\beta r^{4}}(\beta m\log(nU)))$. By \Cref{lemma:certificate}, there exists some large enough constant $\alpha>0$ such that every vertex cut of $G$ with $\d(L\cup S),\d(R\cup S)\geq \alpha z\log U\log^{\alpha r^{4}}(m\log U)$ has $h_{G}(L,S,R)\geq \phi/(\alpha\log^{2} U\log^{\alpha r^{4}}(m\log U))$.

There are $O(\log (nU))$ rounds in the cut matching game. Let $m'$ denote the maximum size of $|E(H)|$ and we know $m'\leq m\log(nU)$. Each round the cut player takes $O(m'\cdot(m'U)^{O(1/r)}\cdot\log^{O(r^{2})}(m'U))=O(m\cdot(mU)^{O(1/r)}\cdot\log^{O(r^{2})}(mU))$ time and the matching player takes $O(m^{1+o(1)}\cdot U\log U/\phi)$ time. Therefore, the total running time is $O(m^{1+o(1)+O(1/r)}\cdot U^{1+O(1/r)}\cdot\log^{O(r^{2})}(mU)/\phi)$.

\end{proof}

\begin{proof}[Proof of \Cref{coro:EmbeddingOfExpander}]
    This is a special case of \Cref{lemma:MostBalanceCut} by setting $z=0$ and $r=\log\log\log (mU)$. When $z=0$, there are no fake edges during the algorithm and the expander $H$ with $\Psi(H)\geq 1/\log^{O(r^{4})}(m\log(nU))\geq 1/(m\log U)^{o(1)}$ generated by the cut-matching game will be totally embedded into $G$ by an embedding ${\cal P}$ with vertex-congestion $\eta=O(\log^{2}(nU)/\phi)$. Observe that ${\cal P}$ is actually a union of embeddings ${\cal P}_{i}$ for several matchings $M_{i}$, so the set $E({\cal P})$ is exactly the union of $E({\cal P}_{i})$. Note that each $E({\cal P}_{i})$ can be computed by \Cref{thm:RandMatchingPlayer}. The running time is exactly the running time of algorithm in \Cref{lemma:MostBalanceCut} with the above parameters $z$ and $r$, i.e. $m^{1+o(1)}U^{1+o(1)}/\phi$.
\end{proof}

\subsection{Proof of \Cref{lemma:CutOrCertify}}
\label{proof:BalCutPrune}

We first prove \Cref{lemma:AdjustParameters}, and then complete the proof of \Cref{lemma:CutOrCertify} using it.

\begin{lemma}
Let $G=(V(G),E(G),w,\d)$ be an $n$-vertex $m$-edge graph with, for each $v\in V(G)$, integral weight $1\leq w(v)\leq U$ and integral demand $0\leq \d(v)\leq w(v)$. Given parameters $0<\phi\leq 1,r>1,0<z'<z$, if every vertex cut $(L,S,R)$ in $G$ with $\d(L\cup S),\d(R\cup S)\geq z$ has $h_{G}(L,S,R)> \phi$, then there is a deterministic algorithm that computes a vertex cut $(L,S,R)$ (possibly $L=S=\emptyset$) that has $w(S)\leq \phi\cdot\d(L\cup S)$ and further satisfies
\begin{itemize}
    \item either $\d(L\cup S),\d(R\cup S)\geq \d(V(G))/3$; or
    \item $\d(R)>2\d(V(G))/3$ and every vertex cut $(L',S',R')$ in $G[R]$ with $\d(L'\cup S'),\d(R'\cup S')\geq z'$ has $h_{G[R]}(L',S',R')\geq \phi/(\alpha\log^{2} U\log^{\alpha r^{4}}(m\log U))$, where $\alpha$ is the constant from \Cref{lemma:MostBalanceCut}.
\end{itemize}
The running time of this algorithm is $O((z/z')\cdot m^{1+o(1)+O(1/r)}\cdot U^{1+O(1/r)}\cdot\log^{O(r^{4})}(mU)/\phi)$
\label{lemma:AdjustParameters}
\end{lemma}
\begin{proof}
Let $z^{*}=z'/\alpha\log U\log^{\alpha r^{4}}(m\log U)$. Our algorithm will iteratively apply \Cref{lemma:MostBalanceCut} on the current graph (initially $G$) with parameters $\phi$ and $z^{*}$, and keep prunning the current graph by the sparse vertex cut it returns until we have dropped a constant fraction of demands or we are guaranteed that all balanced vertex cuts (both sides have demands larger than $z'$) in the current graph are not sparse. We use $G_{i}$ denote the current graph at the beginning of the $i$-th round. Initially, $G_{1}=G$. We apply \Cref{lemma:MostBalanceCut} on $G_{i}$ with parameters $\phi$ and $z^{*}$. There are two cases.

\

\noindent\textbf{Case 1.} The output is a vertex cut $(L_{i},S_{i},R_{i})$ in $G_{i}$ with $\d(R_{i})\geq \d(L_{i})> z^{*}$ and $h_{G_{i}}(L_{i},S_{i},R_{i})\leq \phi$. Let $G_{i+1}=G_{i}\setminus(L_{i}\cap S_{i})$. If $\d(V(G_{i+1}))> 2\d(V(G))/3$, we proceed to the next round. Otherwise the algorithm terminates with a vertex cut $(L,S,R)$ where $L=\bigcup_{i'=1}^{i} L_{i'},S=\bigcup_{i'=1}^{i}S_{i'}$ and $R=V(G)\setminus(L\cup S)$. The returned $(L,S,R)$ satisfies the requirements because of the following. First, for all $1\leq i'\leq i$, observe that $R\subseteq R_{i'}$, so there is no edge between $L_{i'}$ and $R$, which implies $(L,S,R)$ is indeed a vertex cut of $G$. Second, for all $1\leq i'\leq i$, we know $w(S_{i'})\leq \phi\cdot\d(L_{i'}\cup S_{i'})$ and all $L_{i'},S_{i'}$ are vertex-disjoint, so $w(S)\leq\phi\cdot\d(L\cup S)$. Lastly, we have $\d(L\cup S)=\d(V(G)\setminus V(G_{i+1}))\geq \d(V(G))-2\d(V(G))/3=\d(V(G))/3$. Also, we have $\d(L\setminus L_{i})\leq \d((L\cup S)\setminus (L_{i}\cup S_{i}))= \d(V(G))-\d(V(G_{i}))$ and $\d(L_{i})\leq \d(V(G_{i}))/2$, so $\d(L)\leq \d(V(G))-\d(V(G_{i}))/2\leq 2\d(V(G))/3$ and $\d(R\cup S)=\d((V(G)\setminus L))\geq \d(V(G))/3$. 

\

\noindent\textbf{Case 2.} We are guaranteed that every vertex cut $(L',S',R')$ in $G_{i}$ with $\d(L'\cup S'),\d(R'\cup S')\geq \alpha z^{*}\log U\log^{\alpha r^{4}}(m\log U)=z'$ has $h_{G_{i}}(L',S',R')\geq \phi/(\alpha\log^{2} U\log^{\alpha r^{4}}(m\log U))$. The algorithm terminates with the cut $(L,S,R)$ where $L=\bigcup_{i'=1}^{i-1}L_{i'},S=\bigcup_{i'=1}^{i-1}S_{i'}$ and $R=V(G_{i})$. By a similar argument, we have $(L,S,R)$ is a vertex cut of $G$, $\d(S)\leq\phi\cdot\d(L\cup S)$ and $\d(R)=\d(V(G_{i}))>2\d(V(G))/3$. In particular, when $i=1$ the vertex cut degenerates and there will be $L=S=\emptyset$ and $R=V(G)$.

\

We now analyse the running time of this algorithm. We claim that the number of rounds is $O(z/z^{*})$. Observe that in the last round $i^{*}$, we have $\d(G_{i^{*}})>2\d(V(G))/3$. Consider the cut $(L^{*},S^{*},R^{*})$ where $L^{*}=\bigcup_{i'=1}^{i'<i^{*}}L_{i'}$, $S^{*}=\bigcup_{i'=1}^{i'<i^{*}}S_{i'}$ and $R^{*}=V(G)\setminus(L^{*}\cup S^{*})$. We have $\d(S^{*})\leq \phi\cdot\d(L^{*}\cup S^{*})$ and $\d(L^{*}\cup S^{*})=\d(V(G))-\d(V(G_{i^{*}}))\leq \d(V(G))/3$, which implies $h_{G}(L^{*},S^{*},R^{*})\leq \phi$, so $\d(L^{*}\cup S^{*})\leq z$ by the property of $G$. Combining the fact that each vertex cut $(L_{i},S_{i},R_{i})$ has $\d(L_{i}\cup S_{i})\geq z^{*}$, we can bound the number of rounds by $O(z/z^{*})$. 

In each round, we apply \Cref{lemma:MostBalanceCut} once, which takes $O(m^{1+o(1)+O(1/r)}\cdot U^{1+O(1/r)}\cdot\log^{O(r^{2})}(mU)/\phi)$ time. Therefore, the overall running time is bounded by $O((z/z')\cdot m^{1+o(1)+O(1/r)}\cdot U^{1+O(1/r)}\cdot\log^{O(r^{4})}(mU)/\phi)$ because $z^{*}=z'/\alpha\log U\log^{\alpha r^{4}}(m\log U)$.

\end{proof}

\begin{proof}[Proof of \Cref{lemma:CutOrCertify}]
Let $\phi_{i}=\epsilon/(4\alpha\log^{2} U\log^{\alpha r^{4}}(m\log U))^{i-1},z_{i}=(2\d(V(G)))^{(r+1-i)/r}$ for all $1\leq i\leq r+1$. In particular, $\phi_{1}=\epsilon$, $z_{1}=2\d(V(G))$, $\phi_{r+1}=\epsilon/(\log^{O(r)}U\log^{O(r^{5})}(m\log U))$ and $z_{r+1}=1$. From now we use the term \textit{$(\phi,z)$-nearly vertex expander} to denote the graph in which every vertex cut $(L,S,R)$ with $\d(L\cup S),\d(R\cup S)\geq z$ has vertex expansion greater than $\phi$.  

The algorithm has $r$ rounds. In the $i$-th round, starting with a $(\phi_{i},z_{i})$-nearly vertex expander, it will apply \Cref{lemma:AdjustParameters} to prune the graph and the remaining graph will turn out to be a $(\phi_{i+1},z_{i+1})$-nearly vertex expander, if \Cref{lemma:AdjustParameters} does not return a balance sparse cut. Therefore, starting with the initial graph $G$ (obviously a $(\phi_{1},z_{1})$-nearly vertex expander), either we will get a balance sparse cut, or the final graph will have vertex expansion at least $\phi=\phi_{r+1}=\epsilon/(\log^{O(r)}U\log^{O(r^{5})}(m\log U))$ as desired. The detailed description and proof are as follows.

Initially, we let $G_{1}=G$ be our current graph. $G$ is a $(\phi_{1},z_{1})$-nearly vertex expander because $z_{1}=2\d(V(G))$ and there is no cut with $\d(L\cup S),\d(R\cup S)\geq z_{1}$. We apply \Cref{lemma:AdjustParameters} on $G_{1}$ with parameters $\phi_{1},r,z=z_{1},z'=z_{2}$. The output is a vertex cut $(L_{1},S_{1},R_{1})$ with $w(S_{1})\leq\phi_{1}\cdot\d(L_{1}\cup S_{1})$ in either of the following two cases. In the first case, we have $\d(L_{1}\cup S_{1}),\d(R_{1}\cup S_{1})\geq \d(V(G))/3$. The algorithm terminates with $(L_{1},S_{1},R_{1})$ as desired. In the second case, $\d(R_{1})\geq 2\d(V(G))/3$, and every vertex cut $(L',S',R')$ in $G[R_{1}]$ with $\d(L'\cup S'),\d(R'\cup S')\geq z_{2}$ has $h_{G[R_{1}]}(L',S',R')\geq \phi_{1}/(\alpha\log^{2} U\log^{\alpha r^{4}}(m\log U))\geq 4\phi_{2}$. We let $G_{2}=G[R_{1}]$ and proceed to the second round.

The $i$-th round for $2\leq i\leq r$ is similar to the first iteration. During all iterations, we maintain the following invariants which initially hold for $i=2$. The first invariant is that $G_{i}$ satisfies that every vertex cut $(L,S,R)$ of $G_{i}$ with $\d(L\cup S),\d(R\cup S)\geq z_{i}$ has $h_{G_{i}}(L,S,R)\geq 4\phi_{i}$, and the second invariant is $\d(G_{i})\geq \d(V(G))/2$. We apply \Cref{lemma:AdjustParameters} on $G_{i}$ with parameters $\phi_{i},r,z=z_{i},z'=z_{i+1}$. The output is a vertex cut $(L_{i},S_{i},R_{i})$ with $w(S_{i})\leq\phi_{i}\cdot\d(L_{i}\cup S_{i})$. There are two cases (actually only one).

\

\noindent\textbf{Case 1.} This case cannot happen as we will show that there cannot be vertex cut $(L_{i},S_{i},R_{i})$ of $G_{i}$ such that $\d(L_{i}\cup S_{i}),\d(R_{i}\cup S_{i})\geq \d(V(G_{i}))/3$. Otherwise there will be $\d(L_{i}\cup S_{i}),\d(R_{i}\cup S_{i})\geq \d(V(G_{i}))/3\geq \d(V(G))/6\geq z_{i}$ (note that $z_{i}\leq z_{2}=(2\d(V(G)))^{(r-1)/r}\leq \d(V(G))/10$ since $1\leq r \leq \lfloor\log_{20}\d(V(G))\rfloor$) and $w(S_{i})\leq\phi_{i}\cdot\d(L_{i}\cup S_{i})\leq 3\phi_{i}\cdot\d(R_{i}\cup S_{i})$ which implies $h_{G_{i}}(L_{i},S_{i},R_{i})\leq 3\phi_{i}$. This violates the first invariant.

\

\noindent\textbf{Case 2.} We will always have $\d(R_{i})\geq 2\d(V(G_{i}))/3$, and every vertex cut $(L',S',R')$ of $G[R_{i}]$ with $\d(L'\cup S'),\d(R'\cup S')\geq z_{i+1}$ has $h_{G[R_{i}]}(L',S',R')\geq \phi_{i}/(\alpha\log^{2} U\log^{\alpha r^{4}}(m\log U))\geq 4\phi_{i+1}$. Let $G_{i+1}=G[R_{i}]$ and obviously the first invariant holds for $i+1$. 
    
It remains to show the second invariant holds for $i+1$, i.e. $\d(V(G_{i+1}))\geq \d(V(G))/2$. Assume for the contradiction that $\d(V(G_{i+1}))<\d(V(G))/2$. 
Let $(\hat{L},\hat{S},\hat{R})$ be a vertex cut of $G_{2}$ where $\hat{L}=\bigcup_{i'=2}^{i}L_{i},\hat{S}=\bigcup_{i'=2}^{i}S_{i}$ and $\hat{R}=R_{i}$. 
Because $\d(V(G_{i+1}))<\d(V(G))/2$ and $\d(V(G_{2}))\geq 2\d(V(G))/3$, we have $\d(\hat{L}\cup\hat{S})=\d(V(G_{2}))-\d(V(G_{i+1}))\geq \d(V(G))/6$. Because $\d(V(G_{i}))\geq \d(V(G))/2$ and $\d(R_{i})\geq 2\d(V(G_{i}))/3$, we have $\d(\hat{R}\cup\hat{S})\geq \d(R_{i})\geq \d(V(G))/3$. We further have $w(\hat{S})\leq \phi_{2}\cdot\d(\hat{L}\cup\hat{S})$, because for all $2\leq i'\leq i$, $w(S_{i})\leq\phi_{i}\cdot\d(L_{i}\cup S_{i})\leq \phi_{2}\cdot\d(L_{i}\cup S_{i})$. Also, $w(\hat{S})\leq 3\phi_{2}\cdot\d(\hat{R}\cap\hat{S})$ because $\d(\hat{L}\cup\hat{S})\leq \d(V(G))\leq 3\d(\hat{R}\cup\hat{S})$. In summary, $(\hat{L},\hat{S},\hat{R})$ is a vertex cut of $G_{2}$ with $\d(\hat{L}\cup\hat{S}),\d(\hat{R}\cup\hat{S})\geq \d(V(G))/6\geq z_{2}$ and $h_{G_{2}}(\hat{L},\hat{S},\hat{R})\leq 3\phi_{2}$, which contradicts the first invariant.

\

At the end of the algorithm, we obtain $G_{r+1}$. Let $L=\bigcup_{i'=1}^{r}L_{i},S=\bigcup_{i'=1}^{r}S_{i}$ and $R=V(G_{r+1})$ with $w(S)\leq\phi_{1}\cdot\d(L\cup S)=\epsilon\cdot\d(L\cup S)$. By the first invariant, $G[R]$ is exactly $G_{r+1}$ in which every vertex cut $(L,S,R)$ with $\d(L\cup S),\d(R\cup S)\geq z_{r+1}=1$ has $h_{G[R]}(L,S,R)\geq \phi_{r+1}=\phi$. Therefore, $h(G[R])\geq \phi$, and by the second invariant, $\d(R)=\d(V(G_{r+1}))\geq \d(V(G))/2$.

There are at most $r$ iterations in the algorithm and each iteration we will apply \Cref{lemma:AdjustParameters} once with $z/z'=\d(V(G))^{1/r}\leq (nU)^{1/r}$. Therefore, the overall running time is $O(m^{1+o(1)+O(1/r)}\cdot U^{1+O(1/r)}\cdot\log^{O(r^{4})}(mU)/\phi)$.

\end{proof}

\subsection{Proof of \Cref{lemma:TreeToInterval}}
\label{proof:TreeToInterval}

Let $\ET_{\full}(\tau)$ be an ordered list representing the trajectory of the same Euler tour as $\ET(\tau)$ but keeping all occurrences of all vertices. Obviously, preprocessing $\ET_{\full}(\tau)$ and $\ET(\tau)$ only takes $O(|V(\tau)|)$ time and $\ET(\tau)$ is a sublist of $\ET_{\full}(\tau)$. It is well-known that by removing all occurrences of failed vertices on $\ET_{\full}(\tau)$, $\ET_{\full}(\tau)$ will be splitted into $O(d\Delta)$ intervals in $O(d\Delta\log(d\Delta))$ time (the $\log(d\Delta)$ factor is from sorting occurrences of failed vertices by $\ET_{\full}(\tau)$), and vertices inside each of the intervals will be contained by a same subtree. Obviously, the set $I_{\tau}$ can be obtained by taking the restrictions of these intervals on $\ET(\tau)$. The owner of each interval is exactly the subtree that contains an arbitrary vertex in this interval, which can be computed in $O(\log d)$ time as shown below.

To compute the subtree containing a given vertex, we need some preparations in the preprocessing phase and the update phase. In the preprocessing phase, we construct in linear time on $\tau$ the data structure which supports lowest common ancestor (LCA) query in $O(1)$ time \cite{BFC00}. In the update phase, we compute in $O(d\log d)$ time the lowest failed ancestor for each failed vertex by the following observation. For each failed vertex $v$, its lowest failed ancestor is also a failed ancestor of $u$, where $u$ is the failed vertex right before $v$ on $\ET(\tau)$. Hence it suffices to find the lowest failed ancestor of $\LCA(u,v)$ by doing a binary search on $u$'s failed ancestors (with supports of some jumping pointers). Regarding answering the query, given a vertex $v\in V(\tau)$, we can identify the subtree $v$ inside by figuring out the lowest failed ancestor of $v$ (for each subtree, we take the edge connecting its root and its root's parent as the identification). Similarly, take the failed vertex $u$ right before $v$ on $\ET(\tau)$ by a binary search on failed vertices. Then the lowest failed ancestor of $\LCA(u,v)$ can be obtained by a binary search on $u$'s failed ancestors. The query time is $O(\log d)$. 

\section{Deterministic Approximate Vertex-Capacitated Max Flow}
\label{sect:DetVertexFlow}

Efficient algorithms for approximate vertex-capacitated max flow problem have been studied in several works, e.g. \cite{CK19,CS21,BGS21}. The currently state-of-the-art algorithm is a randomized one shown in \Cref{lemma:RandVertexFlow} from \cite{BGS21} with almost-linear running time. 

\begin{lemma}[Almost-Linear-Time Approximate Vertex-Capacitated Max Flow, \cite{BGS21}] For any constant $1/\polylog(n)< \epsilon < 1$, consider an $n$-vertex $m$-edge undirected graph $G=(V(G),E(G),c)$ with vertex capacity function $c:V(G)\setminus\{s,t\}\to \mathbb{R}_{>0}$. Let $s,t\in V(G)$ be source and sink vertices without capacity constraints such that $s$ and $t$ are not directly connected by an edge. Then, there is a randomized algorithm that computes, with high probability,
\begin{itemize}
    \item a feasible flow $f$ that sends a $(1-\epsilon)$-fraction of the max flow value from $s$ to $t$, and
    \item a cut $S\subseteq V(G)\setminus\{s,t\}$ whose deletion will disconnect $s$ from $t$, with total capacities no more than $(1+\epsilon)$ times the min cut value.
\end{itemize}
The algorithm runs in $O(m^{1+o(1)}\log\log C)$ time, where $C$ is the maximum capacity ratio.
\label{lemma:RandVertexFlow}
\end{lemma}

\Cref{lemma:DetVertexFlow} is a deterministic variant with slower running time, but it has not been stated explicitly previously, so we show the algorithm below for completeness.

\begin{lemma}[Deterministic Approximate Vertex-Capacitated Max Flow] For any constant $1/\polylog(n)< \epsilon < 1$, consider an $n$-vertex $m$-edge undirected graph $G=(V(G),E(G),c)$ with vertex capacity function $c:V(G)\to [1,C]\cup\{\infty\}$. Let $s,t\in V(G)$ be source and sink vertices such that $s$ and $t$ are not directly connected by an edge and $c(s),c(t)=\infty$, and they are the only vertices with infinite capacity. Then, there is a deterministic algorithm that computes,
\begin{itemize}
    \item a feasible flow $f$ that sends a $(1-\epsilon)$-fraction of the max flow value from $s$ to $t$, and
    \item a cut $S\subseteq V(G)\setminus\{s,t\}$ whose deletion will disconnect $s$ from $t$, with total capacities no more than $(1+\epsilon)$ times the min cut value.
\end{itemize}
The algorithm runs in $O(m^{1+o(1)}\polylog C+n^{o(1)}C_{\ssum})$ time, where $C$ is the maximum capacity and $C_{\ssum}=\sum_{v\in V(G)\setminus\{s,t\}}c(v)$ denotes the total capacities of finite capacitated vertices.
\label{lemma:DetVertexFlow}
\end{lemma}

\Cref{lemma:DetVertexFlow} follows the approach of computing approximate max flow via the multiplicative-weight-update (MWU) framework and decremental single-source shortest paths (SSSP) algorithms. The algorithm are almost indentical to those in \cite{CK19,CS21} except that we plug in the the currently state-of-the-art decremental SSSP data strucutre in \cite{BGS21}.

\begin{algorithm}[H]
\caption{An approximate algorithm for $s$-$t$ flow in edge-capacitated digraph}
\begin{algorithmic}[1]
\Require A directed graph $G=(V,E)$ with edge capacity function $c:E\to \mathbb{R}_{>0}\cup\{\infty\}$, a source $s$, a sink $t$ and an accuracy parameter $0<\epsilon<1$.
\Ensure An $s$-$t$ flow which is capacity-feasible.
\State $\delta\gets (2m)^{-1/\epsilon}$
\State For all $e\in E$, initialize the edge length by $\ell(e)\gets\delta/c(e)$
\State Initialize a flow function $f\equiv 0$
\While{$\sum_{e\in E}c(e)\ell(e)<1$}
\State $P\gets$ a $(1+\epsilon)$-approximate $\ell$-shortest $s$-$t$ path
\State $c\gets\min_{e\in P}c(e)$
\State $f(P)\gets f(P)+c$ \Comment{Augment along the path $P$.}
\State For all $e\in E(P)$, $\ell(e)\gets\ell(e)(1+\frac{\epsilon c}{c(e)})$
\EndWhile
\State return the scaled down flow $f/\log_{1+\epsilon}(\frac{1+\epsilon}{\delta})$
\end{algorithmic}
\label{algo:EdgeFlow}
\end{algorithm}

\begin{lemma}[Decremental SSSP, \cite{BGS21}]
Given an undirected, decremental graph $G=(V,E,w)$, a fixed source vertex $s\in V$, and any $\epsilon > 1/\polylog(n)$, there is a deterministic data structure that maintains a $(1+\epsilon)$-approximation of the distance from $s$ to every vertex $t$ in $V$ explicitly in total update time $m^{1+o(1)}\polylog L$, where $L$ is the maximum weight ratio. The data structure can further answers queries for an $(1+\epsilon)$-approximate shortest $s$-to-$t$ path $\pi(s,t)$ in time $|\pi(s,t)|n^{o(1)}$.
\label{lemma:DecreSSSP}
\end{lemma}

Let $G=(V(G),E(G),c)$ be the vertex-capacitated undirected graph in \Cref{lemma:DetVertexFlow}. By the standard reduction, we can obtain the corresponding edge-capacitated digraph $G'=(V(G'),E(G'),c')$, where $V(G')=\{v^{\iin},v^{\out}\mid v\in V(G)\}$ and $E(G')=\{(u^{\out},v^{\iin},\infty),(v^{\out},u^{\iin},\infty)\mid (u,v)\in E(G)\}\cup \{(v^{\iin},v^{\out},c(v))\mid v\in V(G)\}$, and it suffices to work on $G'$ with source $s^{\iin}$ and sink $t^{\out}$. We run \Cref{algo:EdgeFlow} on graph $G'$ and \Cref{lemma:VertexFlowCorrectness} follows the analysis in \cite{CS21}, which implies a desired flow solution of $G$ by choosing $\epsilon/\alpha$ as the actual error parameter for some large enough constant $\alpha$.

\begin{lemma}
The flow $f/\log_{1+\epsilon}(\frac{1+\epsilon}{\delta})$ is a $(1-O(\epsilon))$-approximate max flow on $G'$.
\label{lemma:VertexFlowCorrectness}
\end{lemma}

We speed up \Cref{algo:EdgeFlow} by plugging in the decremental SSSP data structure as shown in \Cref{lemma:DecreSSSP}. Let $\ell'$ denote the edge length function during executing \Cref{algo:EdgeFlow}. To compute approximate shortest paths in $G'$ in line 5, it suffices to compute approximate shortest $s$-$t$ paths in the original graph $G$ if we assign each vertex $v\in V(G)$ with vertex length $\ell(v)$ equal to $\ell'(e)$ of $e=(v^{\iin},v^{\out})$. Observe that for each vertex length $\ell(v)=\delta/c(v)$ initially and keep increasing during the execution, but there will always be $\ell(v)< 1/c(v)$ in the while loop. Therefore, following the idea in \cite{CK19} that discretizing all possible lengths by the power of $(1+\epsilon/3)$ for each vertex, we construct the graph $G''=(V(G''),E(G''),c'')$ as follows. Let $K=\lceil\log_{(1+\epsilon/3)}1/\delta\rceil$. Then we set $V(G'')=\{(v_{i},\delta(1+3/\epsilon)^{i}/c(v))\mid v\in G,0\leq i\leq K\}\cup\{(s',0),(t',0)\}$ and $E(G'')=\{(u_{i},v_{j})\mid (u,v)\in E(G),0\leq i,j\leq K\}\cup\{(s',s_{i}),(t_{i},t')\mid 0\leq i\leq K\}$. Over the course of the algorithm, we maintain the invariant that for each $v\in V(G)$, vertices $v_{0},...,v_{i-1}$ have been deleted and $v_{i}$ is still in $G''$, where $i$ is the minimum integer such that $c''(v_{i})=\delta(1+3/\epsilon)^{i}/c(v)\geq \ell(v)$. Obviously, a $(1+3/\epsilon)$ approximate $s'$-$t'$ shortest path in $G''$ implies a $(1+\epsilon)$ approximate $s$-$t$ shortest path in $G$, namely a desired $(1+\epsilon)$-approximate $s^{\iin}$-$t^{\out}$ shortest path in $G'$. Technically, when applying decremental SSSP algorithm on $G''$, we can turn vertex lengths into edge lengths by assigning each edge $e=(u,v)$ with length $\ell(e)=(\ell(u)+\ell(v))/2$. Observe that the edge lengths are at most $(1+\epsilon/3)$ and at least $\delta/(2C)$ (the only zero-length vertices are $s$ and $t$ but they are not connected), so the maximum length ratio is $L=O(C/\delta)$.

It remains to analyse the running time of the whole algorithm. The bottleneck is running the decremental SSSP algorithm on $G''$. Since $|V(G'')|=O(K|V(G)|)=\Otil(n)$, $|E(G'')|=O(K^{2}|E(G)|)=\Otil(m)$ and $L=O(C/\delta)$, the total update time is $|E(G'')|^{1+o(1)}\polylog L$, namely $m^{1+o(1)}\polylog C$. The total query time is almost linear to the total length of augmentation paths by \Cref{lemma:DecreSSSP}. Since each time we will augment at least one unit of flow along a path, and each augmentation path $P$ with length $|P|$ will go through $\Theta(|P|)$ finite capacitated edges, so it will add at least $\Theta(|P|)$ units of flows to finite capacitated edges. Further, we know finite capacitated edges have total capacities the same as $C_{\ssum}$. Combining that $f/\log_{1+\epsilon}(\frac{1+\epsilon}{\delta})$ is a feasible flow on $G'$, the total length of augmentation paths is bounded by $C_{\ssum}\log_{1+\epsilon}(\frac{1+\epsilon}{\delta})$, namely $C_{\ssum}\polylog(m)$. Therefore, the total query time is $n^{o(1)}C_{\ssum}$ and we conclude that the total running time is $m^{1+o(1)}\polylog C+n^{o(1)}C_{\ssum}$.

Finally, we note that an approximate cut solution can be computed by the MWU framework as a byproduct in nearly linear additional running time deterministically. See Section 7.2 in \cite{CK19} for detailed discussions.

\section{Expander Decomposition for Hypergraphs}
\label{sect:HyperDecomp}

In this section, we slightly generalize our techniques on vertex expander decomposition in \Cref{sect:VertexExpander} to hypergraphs. This part is not needed for our vertex-failure connectivity oracle, but we expect expander decomposition will be useful in other applications.

We consider a hypergraph $H=(V(H),E(H))$ such that each edge $e\in E(H)$ is represented by a non-empty subset of vertices, i.e. $e\subseteq V(H)$ and $e\neq\emptyset$. Given a subset $U\subseteq V(H)$, we use $H[U]$ to denote the vertex-induced subgraph with vertex set $V(H[U])=U$ and edge set $E(H[U])=\{e\cap U\mid e\in E(H),e\cap U\neq\emptyset\}$\footnote{We consider $E(H[U])$ as a multiset. Namely, for each $e\in E(H),e\cap U\neq\emptyset$, we create a copy of $e\cap U$ in $E(H[U])$.}. Note that in our notation, for a hyperedge $e$ partially intersects $U$, we keep the part $e\cap U$ in $E(H[U])$. Here we only show the result for the notion of \textit{conductance} on hypergraphs (see \Cref{def:HyperConductance}), and it's likely to generalize the result to the setting considering other notion of sparse cuts.

\begin{definition}[Conductance on Hypergraphs]
Let $H=(V(H),E(H))$ be a hypergraph. For any cut $(S,\bar{S})$ with $\emptyset \subset S\subset V(H)$ and $\bar{S}=V(H)\setminus S$, its \textit{conductance} is 
\[
    \Phi_{H}(S)=\frac{|E_{H}(S,\bar{S})|}{\min\{|E(H[S])|,|E(H[\bar{S}])|\}},
\]
where $E_{H}(S,\bar{S})=\{e\mid e\in E(H),e\cap S\neq\emptyset,e\cap\bar{S}\neq\emptyset\}$ denote edges crossing the cut. The conductance of $H$, denoted by $\Phi(H)$, is the minimum conductance of any cut $(S,\bar{S})$ with $\emptyset \subset S\subset V(H)$.
\label{def:HyperConductance}
\end{definition}

\begin{definition}[Expander Decomposition on Hypergraphs]
Let $H=(V(H),E(H))$ be a hypergraph. For parameters $0\leq\phi\leq\epsilon\leq 1$, an $(\epsilon,\phi)$-expander decomposition of $H$ is a partition ${\cal H}=(V_{1},...,V_{k})$ of $V(H)$ such that
\begin{itemize}
    \item for all $V_{i}\in{\cal H}$, $\Phi(H[V_{i}])\geq \phi$; and
    \item the number of crossing edges, i.e. edges $e\in E(H)$ such that $\forall V_{i}\in {\cal H},e\nsubseteq V_{i}$, is at most $\epsilon|E(H)|$. 
\end{itemize}
\label{def:HyperDecomp}
\end{definition}

\begin{theorem}
Let $H=(V(H),E(H))$ be an $n$-vertex $m$-edge hypergraph with $p=\sum_{e\in E(H)}|e|$. Given parameters $0<\epsilon\leq 1$, there is an algorithm that computes an $(\epsilon,\phi)$-expander decomposition of $H$, where $\phi=\epsilon/m^{o(1)}$. The algorithm can be deterministic with running time $O(p^{1+o(1)}/\phi)$, or randomized with running time $O(p^{1+o(1)})$ and high correct probability.
\label{thm:HyperDecomp}
\end{theorem}

\Cref{thm:HyperDecomp} is the main result on expander decomposition on hypergraphs defined in \Cref{def:HyperDecomp}. As we will show below, the expander decomposition problem on a hypergraph can be reduced to computing a \textit{neighbor-included vertex expander decomposition} on some weighted bipartite graph. We define this variant in \Cref{def:DecompWithNeighbors} and obtain a fast algorithm in \Cref{thm:DecompWithNeighbors}. We first describe the algorithm in \Cref{thm:DecompWithNeighbors}, and then complete the proof of \Cref{thm:HyperDecomp} using it.

\begin{definition}
Let $G=(V(G),E(G),w,\d)$ be a graph with weight $w(v)\geq 0$ and demand $\d(v)\geq 0$ for each vertex $v\in V(G)$. For parameters $\epsilon>\phi> 0$, an \textit{$(\epsilon,\phi)$-neighbor-included vertex expander decomposition} of $G$ is a collections of pairs of subgraphs ${\cal G}=\{(V_{1},W_{1}),...,(V_{k},W_{k})\}$ such that $V_{1},...,V_{k}$ are disjoint (some $V_{i}$ can possibly be empty). Let $X=V(G)\setminus (\bigcup_{i} V_{i}$) denote the \textit{separator}. Then ${\cal G}$ further satisfies the following.
\begin{itemize}
    \item[(1)] For all $(V_{i},W_{i})\in{\cal G}$, $V_{i}\cup N_{G}(V_{i})\subseteq W_{i}$.
    \item[(2)] $X=\bigcup_{i}W_{i}\setminus V_{i}$.
    \item[(3)] For all $(V_{i},W_{i})\in{\cal G}$, $h(G[W_{i}])\geq\phi$.
    \item[(4)] $\sum_{i} w(W_{i}\setminus V_{i})\leq \epsilon\cdot \d(V(G))$.
\end{itemize}
\label{def:DecompWithNeighbors}
\end{definition}

We emphasize that \Cref{def:DecompWithNeighbors} is slightly different from \Cref{def:Decomp}. First, the decomposition will guarantee that, for each $V_{i}\in{\cal G}$, the neighbor-included subgraph $G[W_{i}]$ (i.e. $V_{i}\cup N_{G}(V_{i})\subseteq W_{i}\subseteq V_{i}\cup X$) is an vertex expander, rather than $G[V_{i}]$ itself. Furthermore, it requires a stronger bound on the weight of the separator $X$. Concretely, instead of bounding $w(X)$, we can even bound the total weight (counting multiplicity) of vertices in the multiset $\sum_{i}W_{i}\setminus V_{i}$, where $\bigcup_{i}W_{i}\setminus V_{i}=X$.

\begin{theorem}
Let $G=(V(G),E(G),w,\d)$ be an $n$-vertex $m$-edge graph with for each $v\in V(G)$, integral weight $1\leq w(v)\leq U$ and integral demand $0\leq d(v)\leq w(v)$. Given parameters $0<\epsilon\leq 1$, there is a algorithm that computes an $(\epsilon,\phi)$-neighbor-included vertex expander decomposition of $G$, where $\phi=\epsilon/(m^{o(1)}U^{o(1)})$. The algorithm can be deterministic with running time is $O(m^{1+o(1)}U^{1+o(1)}/\phi)$, or randomized with running time $O(m^{1+o(1)}U^{o(1)})$ and high correct probability.
\label{thm:DecompWithNeighbors}
\end{theorem}

\begin{lemma}
Let $G=(V(G),E(G),w,\d)$ be an $n$-vertex $m$-edge graph with, for each $v\in V(G)$, integral weight $1\leq w(v)\leq U$ and integral demand $0\leq \d(v)\leq w(v)$. Given parameters $0<\epsilon\leq 1/2$ and $1\leq r\leq \lfloor\log_{20}\d(V(G))\rfloor$, there is a deterministic algorithm that computes a vertex cut $(L,S,R)$ (possibly $L=S=\emptyset$) of $G$ with $w(S)\leq \epsilon\cdot\d(L\cup S)$, which further satisfies
\begin{itemize}
    \item either $\d(L\cup S),\d(R\cup S)\geq \d(V(G))/3$; or
    \item $\d(R\cup S)\geq \d(V(G))/2$ and $h(G[R\cup S])\geq \phi$ for some $\phi=\epsilon/\log^{O(r)}U\log^{O(r^{5})}(m\log U)$.
\end{itemize}
The running time of this algorithm is $O(m^{1+o(1)+O(1/r)}\cdot U^{1+O(1/r)}\cdot\log^{O(r^{4})}(mU)/\phi)$.
\label{lemma:CutOrCertifyBoundary}
\end{lemma}

The proof of \Cref{thm:DecompWithNeighbors} relies on \Cref{lemma:CutOrCertifyBoundary} which is a variant of \Cref{lemma:CutOrCertify}. The difference is that, in the case of returning an unbalanced sparse vertex cut $(L,S,R)$, \Cref{lemma:CutOrCertifyBoundary} will certify that $G[R\cup S]$ is a vertex expander rather than $G[R]$. The algorithm of \Cref{lemma:CutOrCertifyBoundary} is the same with that of \Cref{lemma:CutOrCertify}, except that each time when using the cut $(L,S,R)$ from \Cref{lemma:MostBalanceCut} to prune the graph, the pruned part is $L$ rather than $L\cup S$. The correctness and running time basically follows the same proof in \Cref{proof:BalCutPrune} and we omit it here.

\begin{proof}[Proof of \Cref{thm:DecompWithNeighbors}]
Let $\epsilon'=\epsilon/(10\log(nU))$ and $\phi=\epsilon'/(\log^{2}U\log^{O(r^{5})}(m\log U))$. %
We design a procedure $\Decomp(G)$, which takes a graph $G$ as input and returns a collection of subgraphs. By directly applying \Cref{lemma:CutOrCertifyBoundary} on $G$ with parameters $\epsilon',r$, we get a cut $(L,S,R)$ with $w(S)\leq\epsilon'\cdot\d(L\cup S)$. There are two cases.

\

\noindent\textbf{Case 1.} If $\d(L\cup S),\d(R\cup S)\geq\d(V(G))/3$, the algorithm will recurse on both side. Let $G_{L}=G[L\cup S]\setminus E(G[S])$ be the graph from $G[L\cup S]$ after deleting edges inside $S$, and obtain $G_{R}=G[R\cup S]\setminus E(G[S])$ similarly. Let ${\cal G}_{L}=\Decomp(G_{L})$ and ${\cal G}_{R}=\Decomp(G_{R})$. The output ${\cal G}=\{(V_{i}\setminus S,W_{i})\mid (V_{i},W_{i})\in{\cal G}_{L}\cup{\cal G}_{R}\}$.

\

\noindent\textbf{Case 2.} In the other case, we have $\d(R\cup S)\geq \d(V(G))/2$ and $h(G[R\cup S])\geq\phi$ for some $\phi=\epsilon'/(\log^{O(r)}U\log^{O(r^{5})}(m\log U))$. We only recurse on the left side and obtain ${\cal G}_{L}=\Decomp(G_{L})$, where $G_{L}$ is defined as above. The output is ${\cal G}=\{(V_{i}\setminus S,W_{i})\mid (V_{i},W_{i})\in{\cal G}_{L}\}\cup\{(R,R\cup S)\}$. In particular, if $L$ and $S$ are empty, we return ${\cal G}=\{(R,R)\}$ without any further recursion.

\

We claim that the maximum recursion depth $j^{*}$ is at most $10\log(nU)$. At each recursive step, we have $\d(S)\leq w(S)\leq \epsilon'\d(L\cup S)\leq \d(V(G))/10$, so we have $\d(L\cup S),\d(R\cup S)\leq (2/3+1/10)\d(V(G))$ in both cases, which implies $j^{*}\leq\log\d(V(G))/(-\log(2/3+1/10))\leq 10\log(nU)$.

We now show the correctness of the algorithm by induction. Let the induction hypothesis be that, for all $1\leq j\leq j^{*}$, each recursive step $\Decomp(G_{j})$ at depth $j$ will return an $(\epsilon_{j},\phi)$-neighbor-included vertex expander decomposition of $G_{j}$, where $\epsilon_{j}=2(j^{*}-j+1)\epsilon'(1+\epsilon')^{j^{*}-j}$. Observe that the hypothesis holds for every leaf-step $\Decomp(G_{j})=\{(G_{j},G_{j})\}$. Consider a recursive step $\Decomp(G)$ with output ${\cal G}$ at depth $j$ and assume w.l.o.g. it invokes two recursive steps $\Decomp(G_{L})$ and $\Decomp(G_{R})$ with outputs ${\cal G}_{L}$ and ${\cal G}_{R}$ at depth $j+1$, where $G_{L}$ and $G_{R}$ are defined as above. The case with only one recursive branch is analogous. First of all, all $V_{i}$ in ${\cal G}$ are obviously disjoint because $(L,S,R)$ is a vertex cut and the induction hypothesis on ${\cal G}_{L}$ and ${\cal G}_{R}$. We then verify the properties as follows.

\

\noindent\textbf{Property (1).} We consider any $(V_{i},W_{i})\in {\cal G}$. We have $V_{i}=V'_{i}\setminus S$ and $W_{i}=W'_{i}$ for some $(V'_{i},W'_{i})\in{\cal G}_{L}\cup{\cal G}_{R}$. Assume w.l.o.g. $(V'_{i},W'_{i})\in{\cal G}_{L}$. From the induction hypothesis, $V'_{i}\cup N_{G_{L}}(V'_{i})\subseteq W'_{i}$. Because $V_{i}=V_{i}'\setminus S$ and $V_{i}\subseteq L$ is not adjacent to $R$, we have $N_{G}(V_{i})\subseteq L\cup S$ and $N_{G}(V_{i})\subseteq N_{G_{L}}(V'_{i})\cup (V'_{i}\setminus V_{i})=N_{G_{L}}(V'_{i})\cup(V'_{i}\cap S)$, which implies $V_{i}\cup N_{G}(V_{i})\subseteq V'_{i}\cup N_{G_{L}}(V'_{i})\subseteq W'_{i}=W_{i}$.

\

\noindent\textbf{Property (2).} From the induction hypothesis, we know $(L\cup S)\setminus \bigcup_{(V_{i},W_{i})\in{\cal G}_{L}}V_{i}=\bigcup_{(V_{i},W_{i})\in{\cal G}_{L}}W_{i}\setminus V_{i}$. Combining $V_{i}\subseteq W_{i}$ for each $(V_{i},W_{i})\in{\cal G}_{L}$, we have $L\cup S=\bigcup_{(V_{i},W_{i})\in{\cal G}_{L}}W_{i}$. Analogously $R\cup S=\bigcup_{(V_{i},W_{i})\in{\cal G}_{R}}W_{i}$. Therefore, $V(G)=\bigcup_{(V_{i},W_{i})\in{\cal G}}W_{i}$, which implies property (2).

\

\noindent\textbf{Property (3).} It is simply from the induction hypothesis and the facts that $G_{L}$ and $G_{R}$ are subgraphs of $G$ and that ${\cal G}$ collects all $W_{i}$ in ${\cal G}_{L}$ and ${\cal G}_{R}$.

\

\noindent\textbf{Property (4).} For each $(V_{i},W_{i})$, it corresponds to some $(V'_{i},W'_{i})\in{\cal G}_{L}\cup{\cal G}_{R}$ such that $V_{i}=V'_{i}\setminus S$ and $W_{i}=W'_{i}$, so $w(W_{i}\setminus V_{i})=w(V_{i}'\cap S)+w(W'_{i}\setminus V'_{i})$. First, we have $\sum_{i} w(V'_{i}\cap S)\leq 2w(S)\leq 2\epsilon'\d(V(G_{j}))$ because $V'_{i}$ from the same side are disjoint and $w(S)\leq \epsilon'\d(L\cup S)\leq \epsilon'\d(V(G_{j}))$. Then by the induction hypothesis, $\sum_{i} w(W'_{i}\setminus V'_{i})\leq \epsilon_{j+1}\cdot(\d(L\cup S)+\d(R\cup S))=\epsilon_{j+1}\cdot(\d(V(G_{j}))+\d(S))\leq \epsilon_{j+1}(1+\epsilon')\d(V(G_{j}))$, where the last inequality is because $\d(S)\leq w(S)\leq \epsilon'\d(V(G_{j}))$. Therefore, we have $\sum_{i}w(W_{i}\setminus V_{i})\leq (2\epsilon'+\epsilon_{j+1}(1+\epsilon'))\d(V(G_{j}))\leq \epsilon_{j}\d(V(G_{j}))$.

\

We finally claim that ${\cal G}$ is an $(\epsilon,\phi)$-neighbor-included vertex expander decomposition of $G$ because $\epsilon_{1}\leq \epsilon$. Regarding the running time, observe that for a recursive step which cuts a graph $G$ by $(L,S,R)$, its two branches on $G_{L}$ and $G_{R}$ satisfy $|E(G_{L})|+|E(G_{R})|\leq |E(G)|$ because we remove all edges inside $S$. Thus the total number of edges in all recursive steps at each depth is still $O(m)$ and all recursive step at one depth takes $O(m^{1+o(1)+O(1/r)}\cdot U^{1+O(1/r)}\cdot\log^{O(r^{4})}(mU)/\phi)$. Summing over $O(\log(nU))$ levels and set $r=\log\log\log(mU)$, the overall running time is $O(m^{1+o(1)}U^{1+o(1)}/\phi)$ and the conductance $\phi$ is at least $\epsilon/(m^{o(1)}U^{o(1)})$. Like \Cref{remark:DetDecomp}, we can substitute the subroutine in \Cref{lemma:DetVertexFlow} with that in \Cref{lemma:RandVertexFlow} to obtain a faster randomized algorithm with running time improved by an $O(U/\phi)$ factor.
\end{proof}

\begin{proof}[Proof of \Cref{thm:HyperDecomp}]
We construct a bipartite graph $G$ with $V(G)=V(H)\cup V_{E}$ where $V_{E}=\{x_{e}\mid e\in E(H)\}$ and $E(G)=\bigcup_{e\in E(H)}\{(x_{e},v)\mid v\in e\}$. Namely, we convert $H$ into a graph $G$ by substituting each $e\in E(H)$ with a vertex $x_{e}$ and edges connecting $x_{e}$ and $e$'s endpoints. We set the vertex weight and demand functions on $V(G)$ as follows. For each $x_{e}\in\{x_{e}\mid e\in E(H)\}\subseteq V(G)$, it has weight $w(x_{e})=1$ and demand $\d(x_{e})=1$, while for each $v\in V(H)\subseteq V(G)$, it has weight $w(v)=m$ and demand $\d(v)=0$. The weights and demands are bounded by $U=m$. Let ${\cal G}$ be an $(\epsilon,\phi)$-neighbor-included vertex expander decomposition of $G$ obtained by applying \Cref{thm:DecompWithNeighbors} with parameter $\epsilon$, where $\phi=\epsilon/(|V(G)|^{o(1)}U^{o(1)})=\epsilon/m^{o(1)}$. We construct ${\cal H}$ by letting ${\cal H}=\{V'_{i}\cap V(H)\mid (V'_{i},W'_{i})\in{\cal G},V'_{i}\cap V(H)\neq\emptyset\}$. We claim that ${\cal H}$ is an $(\epsilon,\phi)$-expander decomposition of $H$. 

To show ${\cal H}$ is a collection of disjoint subsets of $V(H)$. First all $V_{i}$ in ${\cal H}$ are disjoint because all $V'_{i}$ in ${\cal G}$ are disjoint. To see $\bigcup_{i} V_{i}=V(H)$, we will show $\bigcup_{i} W'_{i}=V(G)$ and $V_{i}=W'_{i}\cap V(H)$. The former is directly from properties (1) and (2) in \Cref{def:DecompWithNeighbors}. The latter is from the following reasons. On one direction, we have $V_{i}=V'_{i}\cap V(H)\subseteq W'_{i}\cap V(H)$. On the other direction, assume there is some $v\in W'_{i}\cap V(H)$ not in $V_{i}$. It implies $v\in W'_{i}\setminus V'_{i}$ and $\sum_{i}w(W'_{i}\setminus V'_{i})\geq w(v')=m> \epsilon\cdot\d(V(G))$, a contradiction.

For each $V_{i}\in{\cal H}$ corresponding to $(V'_{i},W'_{i})\in{\cal G}$, we claim that $\Phi(H[V_{i}])\geq h(G[W'_{i}])\geq\phi$. To see this, we first claim $V_{i}=W'_{i}\cap V(H)$ implies $V_{i}\cup N_{G}(V_{i})=W'_{i}$. Observe that $V_{i}\cup N_{G}(V_{i})\subseteq V'_{i}\cup N_{G}(V'_{i})\subseteq W'_{i}$. For any $x_{e}\in W'_{i}\setminus V(H)$, it must be connected to some $v\in W'_{i}\cap V(H)$ from the definition of $E(G)$ and the fact that $G[W_{i}]$ is an expander. Thus $W'_{i}\setminus V(H)\subseteq N_{G}(V_{i})$ and $W'_{i}\subseteq N_{G}(V_{i})\cup V_{i}$. To prove $\Phi(H[V_{i}])\geq h(G[W'_{i}])$, consider any cut $(A,B)$ of $H[V_{i}]$ where $B=V_{i}\setminus A$. it corresponds to a vertex cut $(L,S,R)$ in $G[W'_{i}]$ where $S=\{x_{e}\mid e\in E_{H[V_{i}]}(A,B)\}$, $L=A\cup N_{G[W'_{i}]}(A)\setminus S$ and $R=B\cup N_{G[W'_{i}]}(B)\setminus S$. Observe that $E(H[A])$ corresponds to vertex set $N_{G}(A)$, which is the same as $N_{G[W'_{i}]}(A)$ because $W'_{i}=V_{i}\cup N_{G}(V_{i})$. Thus $N_{G}(A)\subseteq L\cup S$ and $\d(L\cup S)\geq |N_{G}(A)|=|E(H[A])|$. Similarly, we have $\d(R\cup S)\geq |E(H[B])|$. Therefore,
\[
\Phi_{H[V_{i}]}(A,B)=\frac{|E_{H[V_{i}]}(A,B)|}{\min\{|E(H[A])|,|E(H[B])|\}}\geq\frac{w(S)}{\min\{\d(L\cup S),\d(R\cup S)\}}=h_{G[W'_{i}]}(L,S,R),
\]
and $\Phi(H[V_{i}])\geq h(G[W'_{i}])$.

To bound the number of crossing hyperedges, we claim that each $x_{e}\in V'_{i}\cap V_{E}$ for any $V'_{i}$ in ${\cal G}$ corresponds to a non-crossing hyperedge $e$ only connecting vertices in $V_{i}$. Otherwise let $v$ in another $V_{i'}$ be an endpoint of $e$. By the definition of $E(G)$, there is an edge $(x_{e},v)\in E(G)$, but $v\notin W'_{i}$ because $v\notin V_{i}$ and $V_{i}=W'_{i}\cap V(H)$, which contradicts the fact that $W'_{i}=V_{i}\cup N_{G}(V_{i})$. Therefore, all crossing hyperedges correspond to some $x_{e}\in (V(G)\setminus \bigcup_{i}V'_{i})\cap V_{E}$, and we can conclude that the number of crossing hyperedges is at most $w(V(G)\setminus \bigcup_{i}V'_{i})\leq w(\sum_{i}W'_{i}\setminus V'_{i})\leq \epsilon\cdot\d(V(G))=\epsilon m$.

Regarding the running time, if we allow randomization, the algorithm can run in $O(|E(G')|^{1+o(1)}\cdot U^{o(1)}) = O(p^{1+o(1)})$ time by \Cref{thm:DecompWithNeighbors}. For the deterministic algorithm, we can only obtain an $O((pm)^{1+o(1)}/\phi)$ bound directly from \Cref{thm:DecompWithNeighbors}. However, a tighter bound $O(p^{1+o(1)}/\phi)$ can be obtained by a more careful analysis. The bottleneck of the decomposition algorithm is the subroutine in \Cref{lemma:DetVertexFlow} and the bottleneck of the algorithm in \Cref{lemma:DetVertexFlow} is the flow decomposition barrier. However, we will run the algorithm on $G'$ slightly modified from $G$ (see \Cref{proof:DetMatchingPlayer}). Because $G$ is a bipartite graph with one side on which every vertex has unit weight, such vertices have capacity $O(1/\phi)$ and the length of flow paths in $G'$ is at most linear to the total capacity of these vertices, i.e. $O(n/\phi)$. Therefore, the running time of the deterministic algorithm only increases by an $O(1/\phi)$ factor.
\end{proof}

%% file: main.bbl
\newcommand{\etalchar}[1]{$^{#1}$}